\documentclass[12pt]{article}
\usepackage{amsmath}
\usepackage{graphicx}
\usepackage{enumerate}
\usepackage{natbib}
\usepackage{multibib}
\usepackage{url} 

\usepackage{amsthm, amssymb, bm, bbm}
\usepackage{mathtools}
\usepackage{subcaption}
\usepackage{booktabs}
\usepackage{newfloat}
\usepackage{framed}
\usepackage[usenames,dvipsnames,svgnames,table]{xcolor}
\usepackage[pdfencoding=auto, psdextra,
            colorlinks=true,
            linkcolor=red,
            urlcolor=black,
            citecolor=MidnightBlue!70,
            hypertexnames=false]{hyperref}

\usepackage{setspace}

\newcites{Sup}{References}

\DeclareFloatingEnvironment[name=Algorithm]{myalgorithm}

\allowdisplaybreaks
\def\bal#1\eal{\begin{align}#1\end{align}}
\def\balnn#1\ealnn{\begin{align*}#1\end{align*}}

\newtheorem{theorem}{Theorem}
\newtheorem{lemma}{Lemma}
\newtheorem{corollary}{Corollary}

\newtheorem{proposition}{Proposition}
\newtheorem{assump}{Assumption}

\DeclareMathOperator*{\argmin}{arg\,min}
\DeclareMathOperator*{\argmax}{arg\,max}
\DeclareMathOperator*{\logit}{logit}
\DeclareMathOperator*{\diag}{diag}
\DeclareMathOperator*{\tr}{tr}

\DeclarePairedDelimiter\braket{\langle}{\rangle}
\DeclarePairedDelimiter\set{\{}{\}}
\DeclarePairedDelimiterX{\norm}[1]{\lVert}{\rVert}{#1}
\DeclarePairedDelimiterX{\abs}[1]{\lvert}{\rvert}{#1}
\DeclarePairedDelimiterX{\floor}[1]{\lfloor}{\rfloor}{#1}
\DeclarePairedDelimiterX{\ceil}[1]{\lceil}{\rceil}{#1}

\DeclarePairedDelimiterX{\expectarg}[1]{[}{]}{%
    \ifnum\currentgrouptype=16 \else\begingroup\fi
    \activatebar#1
    \ifnum\currentgrouptype=16 \else\endgroup\fi
}

\DeclarePairedDelimiterX{\variancearg}[1]{(}{)}{%
    \ifnum\currentgrouptype=16 \else\begingroup\fi
    \activatebar#1
    \ifnum\currentgrouptype=16 \else\endgroup\fi
}

\DeclarePairedDelimiterX{\klarg}[1]{(}{)}{%
    \ifnum\currentgrouptype=16 \else\begingroup\fi
    \activatedoublebar#1
    \ifnum\currentgrouptype=16 \else\endgroup\fi
}

\newcommand{\innermid}{\nonscript\;\delimsize\vert\nonscript\;}
\newcommand{\activatebar}{%
    \begingroup\lccode`\~=`\|
    \lowercase{\endgroup\let~}\innermid
    \mathcode`|=\string"8000
}

\newcommand{\doublemid}{\nonscript\;\delimsize\vert\delimsize\vert\nonscript\;}
\newcommand{\activatedoublebar}{%
    \begingroup\lccode`\~=`\|
    \lowercase{\endgroup\let~}\doublemid
    \mathcode`|=\string"8000
}

\newcommand{\Cov}{\operatorname{Cov}\variancearg}
\newcommand{\E}{\mathbb{E} \, \expectarg}

\newcommand{\iidsim}{\overset{\text{iid}}\sim}
\newcommand{\indsim}{\overset{\text{ind.}}\sim}
\newcommand{\Bern}{\operatorname{Bernoulli}}
\newcommand{\PG}{\operatorname{PG}}
\newcommand{\GP}{\operatorname{GP}}
\newcommand{\GIG}{\operatorname{GIG}}
\newcommand{\Reals}[1]{\mathbb{R}^{#1}}

\newcommand{\bY}{\mathbf{Y}}
\newcommand{\by}{\mathbf{y}}
\newcommand{\bX}{\mathbf{X}}
\newcommand{\bx}{\mathbf{x}}
\newcommand{\bt}{\mathbf{t}}
\newcommand{\bO}{\mathbf{O}}
\newcommand{\bU}{\mathbf{U}}
\newcommand{\bu}{\mathbf{u}}
\newcommand{\bv}{\mathbf{v}}
\newcommand{\bw}{\mathbf{w}}
\newcommand{\be}{\mathbf{e}}
\newcommand{\bb}{\mathbf{b}}
\newcommand{\bbeta}{\boldsymbol\beta}
\newcommand{\bTheta}{\boldsymbol\Theta}
\newcommand{\brho}{\boldsymbol\rho}
\newcommand{\blambda}{\boldsymbol\lambda}
\newcommand{\bLambda}{\boldsymbol\Lambda}
\newcommand{\bmu}{\boldsymbol\mu}
\newcommand{\bomega}{\boldsymbol\omega}

\newcommand{\bC}{\mathbf{C}}
\newcommand{\bW}{\mathbf{W}}
\newcommand{\bOmega}{\boldsymbol\Omega}
\newcommand{\bSigma}{\boldsymbol\Sigma}
\newcommand{\bD}{\mathbf{D}}
\newcommand{\bGamma}{\boldsymbol\Gamma}
\newcommand{\bP}{\mathbf{P}}
\newcommand{\bE}{\mathbf{E}}
\newcommand{\bV}{\mathbf{V}}
\newcommand{\bB}{\mathbf{B}}
\newcommand{\btheta}{\boldsymbol\theta}

\newcommand{\boldeta}{\boldsymbol\eta}
\newcommand{\bphi}{\boldsymbol\phi}
\newcommand{\bA}{\mathbf{A}}

\def\mOne{{\mathbbm{1}}}
\newcommand{\ind}[1]{\mOne_{\{#1\}}}

\newcommand{\blind}{1}

\begin{document}

\def\spacingset#1{\renewcommand{\baselinestretch}%
{#1}\small\normalsize} \spacingset{1}


\if1\blind
{
  \title{\bf \Large Fast Variational Inference of Latent Space Models for Dynamic Networks Using Bayesian P-Splines}
  \author{Joshua Daniel Loyal\hspace{.2cm}\\
    Department of Statistics, Florida State University}
  \date{\vspace{-3em}}
  \maketitle
} \fi

\if0\blind
{
  \bigskip
  \bigskip
  \bigskip
  \begin{center}
    {\LARGE\bf Fast Variational Inference of Latent Space Models for Dynamic Networks Using Bayesian P-Splines}
\end{center}
  \medskip
} \fi

\bigskip
\begin{abstract}
Latent space models (LSMs) are often used to analyze dynamic (time-varying) networks that evolve in continuous time. Existing approaches to Bayesian inference for these models rely on Markov chain Monte Carlo algorithms, which cannot handle modern large-scale networks. To overcome this limitation,  we introduce a new prior for continuous-time LSMs based on Bayesian P-splines that allows the posterior to adapt to the dimension of the latent space and the temporal variation in each latent position. We propose a stochastic variational inference algorithm to estimate the model parameters. We use stochastic optimization to subsample both dyads and observed time points to design a fast algorithm that is linear in the number of edges in the dynamic network. Furthermore, we establish non-asymptotic error bounds for point estimates derived from the variational posterior. To our knowledge, this is the first such result for Bayesian estimators of continuous-time LSMs. Lastly, we use the method to analyze a large data set of international conflicts consisting of 4,456,095 relations from 2018 to 2022. 
\end{abstract}

\noindent%

\onehalfspacing

\section{Introduction}

Network data is ubiquitous in modern applications from various scientific disciplines. In general, a network describes the relations, or edges, between pairs of entities, or nodes. Much of the statistical network analysis literature focuses on static networks~\citep{kolaczyk2014, goldenberg2010}, meaning inferences are drawn from a single set of edges observed at one point in time. However, real-world systems are often time-varying, or dynamic, with the relations between nodes changing over time. In this work, we focus on a time series of networks on a common set of $n$ nodes observed at $M$ distinct time points $\set{t_m}_{m=1}^M$ on a compact time interval $\mathcal{T} \subset \Reals{}$ with edges that can change over time. We consider network time series measured  in continuous-time, meaning the observed time points can be irregularly spaced. Furthermore, we allow dyadic covariate information to accompany these networks. Such dynamic network data with covariates appears in diverse fields such as neuroscience~\citep{xiaojingzhu2023} and international relations, cf. Section~\ref{sec:application}.

There is a rapidly growing literature on statistical models for dynamic network data. A prevalent approach extends models for static networks to the dynamic setting. For instance, there exist dynamic versions of various stochastic block models~\citep{yang2011, xing2010, matias2017},  the exponential random graph model (ERGM)~\citep{hanneke2010, krivitsky2014stergm}, random dot product graphs (RDPGs)~\citep{passino2021, chen2023, macdonald2023}, and latent space models (LSMs)~\citep{sakar2006, durante2014, sewell2015}. See \citet{kim2018} for a more complete review. In this work, we focus on the continuous-time dynamic LSM introduced by \citet{durante2014}, which represents each node $i$ with a latent position in a $d$-dimensional Euclidean space that evolves in continuous-time via a vector-valued function $\bu_i(t) : \mathcal{T} \rightarrow \Reals{d}$ called a latent trajectory. The model's advantages are that it is expressive enough to capture complex network structures, incorporates dyadic covariates, and allows for meaningful visualization. 

Despite the empirical success of continuous-time dynamic LSMs, Bayesian inference of their parameters is computationally infeasible for modern large-scale networks and lacks theoretical support. Initially, \citet{durante2014} modeled the latent trajectories as Gaussian processes (GPs) and introduced a Markov chain Monte Carlo (MCMC) algorithm that scales cubically in the number of time points $M$. Subsequent works used specific GPs with state-space representations~\citep{durante2016, guhaniyogi2020} to reduce the run time to linear in $M$. Nevertheless, these existing methods inherit LSMs' usual quadratic scaling in the number of nodes. As such, Bayesian inference can take hours or days for dynamic networks with only a few hundred nodes or time points. In terms of theoretical properties, these previous works verified the large support property of the GP priors but did not address posterior consistency.

In this paper, we introduce a Bayesian inference procedure for continuous-time dynamic LSMs with theoretical guarantees that scales to large dynamic networks. Instead of modeling the latent trajectories with GPs, we approximate them using a finite series of spline basis functions to improve computational tractability. Recently, spline approximations have attracted attention for modeling time-varying parameters in other dynamic network models. Both \citet{lee2020} and \citet{park2022} used splines to parameterize time-varying coefficients in dynamic ERGMs. In addition, \citet{artico2023} used  classical penalized splines~\citep{eilers1996} to approximate latent trajectories in an LSM for relational event data. The parametrization most similar to ours is functional adjacency spectral embedding (FASE) proposed by \citet{macdonald2023}, which uses splines to model the latent trajectories in a dynamic RDPG. However, unlike our methodology, FASE cannot incorporate dyadic covariate information or provide uncertainty quantification.

Under this spline approximation, our first contribution is eliciting an appropriate prior for the basis coefficients. While spline approximations can be effective, their quality heavily depends on the choice of basis dimension, which controls the overall variability of the approximate latent trajectories. Accordingly, influenced by Bayesian P-splines~\citep{lang2004}, we introduce a prior over the basis coefficients designed to ensure the posterior adapts appropriately to the variation in the unknown latent trajectories and the latent space dimension. We call this new prior the P-spline prior for dynamic LSMs. 

For fast inference, we construct estimates based on a variational approximation~\citep{wainwright2008} to the parameter's fractional posterior~\citep{walker2001, bhattacharya2019}. We show that these estimates are consistent with an error rate that adapts to the inherent variation in the true latent trajectories. To our knowledge, this is the first such result for Bayesian estimators of continuous-time dynamic LSMs. This result adds to the literature on the theoretical properties of variational inference for estimating network models~\citep{celisse2012, bickel2013, zhangsbm2020}. In particular, our theoretical results follow the $\alpha$-variational Bayes framework~\citep{yang2020}, which has been used to demonstrate the consistency of the variational approach for discrete-time dynamic LSMs~\citep{LiCh2022, zhao2022}.

Our final contribution is a scalable stochastic variational inference (SVI) procedure~\citep{hoffman2013} to obtain the variational posterior. En route, we introduce a new Polya-gamma augmentation scheme~\citep{polson2013} for conjugate fractional posterior inference, which may have independent interest. Our SVI algorithm scales as the maximum number of edges in a single network observation. As a result, the method can scale to larger networks and perform estimation significantly faster than existing Bayesian approaches. 

The remainder of the article is organized as follows. Section~\ref{sec:model} introduces our spline approximation for continuous-time dynamic LSMs. Section~\ref{sec:priors} develops the proposed P-spline prior for dynamic LSMs. We establish the variational approximation and its theoretical properties in Section~\ref{sec:theory}. We derive the SVI procedure in Section~\ref{sec:estimation}. Section~\ref{sec:sim_study} presents a simulation study, while we apply the methodology to a real conflict network in Section~\ref{sec:application}. The supplementary material contains all proofs and additional technical details. A Python package for the proposed method is available at \url{https://github.com/joshloyal/splinetlsm}.

\section{A Spline Approximation of Dynamic LSMs}\label{sec:model}

\subsection{Notation}

To begin, we establish some notation. For a continuous function $f \, : \, [0,1] \rightarrow \Reals{}$, $\norm{f}_{L_{\infty}[0,1]} = \text{ess\,sup}_{x \in [0,1]} \abs{f(x)}$ denotes the supremum norm. The notation $[\cdot]_{ij}$ denotes the $(i,j)$-th entry of a matrix or the $(i,j)$-th tube fiber of a three-way tensor. For a matrix $\bC$, we denote its minimum singular value as $\sigma_{min}(\bC)$, its Frobenius norm as $\norm{\bC}_{F}$, and its operator norm as $\norm{\bC}_{op}$. We use $\mathcal{O}_d$ to denote the group of $d$-dimensional orthogonal matrices. We let $\mathbf{0}_d$ denote the $d$-dimensional vector of zeros. We use $\indsim$ and $\iidsim$ to denote independently distributed and independently and identically distributed, respectively. For two densities $p$ and $q$, we use $D_{KL}(p, q)$ to denote the Kullback-Leibler (KL) divergence between $p$ and $q$. For sequences $a_n$ and $b_n$, we write $a_n \lesssim b_n$ (or $b_n \gtrsim a_n$) to imply that $a_n \leq c b_n$ for some constant $c$ independent of $n$. The notation $a_n = O(b_n)$ implies $a_n \lesssim b_n$ while $a_n \asymp b_n$ implies $a_n \lesssim b_n$ and $a_n \gtrsim b_n$. We use $a_n \ll b_n$ to mean $\lim_{n \rightarrow \infty} a_n / b_n = 0$.

\subsection{The Continuous-Time Dynamic Latent Space Model}

We model the dynamic network data as a collection of time-index $n \times n$ binary adjacency matrices $\set{\bY_{t_m}}_{m=1}^M$ with random entries $Y_{ij,t_m} = [\bY_{t_m}]_{ij} \in \set{0, 1}$. We assume each network is undirected so that each $\bY_{t_m}$ is symmetric, and we allow self-loops.  We use lower-case letters to denote the observed values of the random adjacency matrices so that $y_{ij,t_m} = 1$ and $y_{ij,t_m} = 0$ indicate the presence or absence of an observed edge between node $i$ and node $j$ at time $t_m$, respectively.  We use $\mathcal{Y} = \set{y_{ij,t_m} : 1 \leq i \leq j \leq n,  1\leq m \leq M}$ to denote the collection of all observed dyadic relations. Additionally, we measure dyadic covariates at each time point $\set{\bx_{ij,t_m} = (x_{ij1,t_m}, \dots x_{ijp,t_m})^{\top} \in \Reals{p} \, : \, 1 \leq i,j \leq n}$, which we collect in a covariate tensor $\mathcal{X}_{t_m} \in \Reals{n \times n \times p}$ with the $(i,j)$-th tube fiber $[\mathcal{X}_{t_m}]_{ij} = \bx_{ij,t_m}$. Since the networks are undirected, we require $\bx_{ij,t_m} = \bx_{ji,t_m}$. We denote the collection of all covariate tensors as $\mathcal{X} = \set{\mathcal{X}_{t_m}}_{m=1}^M$. For the remainder of this article, we assume the time interval $\mathcal{T} = [0, 1]$ since we can always re-scale the data so that this equality holds.

We adopt the continuous-time dynamic LSM proposed by \citet{durante2014cov}, which assumes the edges are independent conditioned on the dyadic covariates and a collection of latent functions so that for $1 \leq i \leq j \leq n$ and $m = 1, \dots, M$, 
\begin{equation}\label{eq:dynlsm}
    Y_{ij,t_m} \mid \bx_{ij,t_m} \indsim \Bern\left\{\textrm{logit}^{-1}([\bTheta_{t_m}]_{ij})\right\}
\end{equation}
with
\begin{align}\label{eq:dynlsm_lr}
    [\bTheta_{t_m}]_{ij} &= \text{logit}\{\mathbb{P}(Y_{ij,t_m} = 1 \mid \bx_{ij,t_m})\} =  \bbeta(t_m)^{\top} \bx_{ij,t_m} + \bu_i(t_m)^{\top}  \bu_j(t_m).
\end{align}
In model~(\ref{eq:dynlsm})--(\ref{eq:dynlsm_lr}), $\bTheta_{t_m} \in \Reals{n \times n}$ has entries $[\bTheta_{t_m}]_{ij}$ indicating the log-odds of an edge forming between nodes $i$ and $j$ at time $t_m$,  $\bbeta(t) = (\beta_1(t), \dots, \beta_p(t))^{\top} : [0, 1] \rightarrow \Reals{p}$ is a vector-valued function of time-varying coefficients associated with the dyadic covariates, and $\bu_i(t) = (u_{i1}(t), \dots, u_{id}(t))^{\top}: [0, 1] \rightarrow \Reals{d}$ is the latent trajectory of node $i$. We collect the latent trajectories into the rows of an $n \times d$ matrix-valued function $\bU(t) = (\bu_1(t), \dots, \bu_n(t))^{\top}$. 

Model~(\ref{eq:dynlsm})--(\ref{eq:dynlsm_lr}) provides an intuitive description for edge formation in dynamic networks. The coefficient function's $k$-th coordinate $\beta_k(t_m)$ measures the extent of homophily in the network attributed to the $k$-th covariate at time $t_m$. Furthermore, the log-odds of an edge forming between two nodes at time $t_m$ increases with the inner-product similarity between their latent positions at time $t_m$. Although the model assumes conditional independence between dyads during a single time point and across time points, endogenous dyadic covariates appearing in Equation~(\ref{eq:dynlsm_lr}) can capture certain temporal dependencies. For example, one can set $x_{ij,t_m} = y_{ij,t_{m-1}}$ to capture edge persistence~\citep{friel2016}.

\subsection{B-Spline Basis Expansions of the Latent Functions}

In this work, we approximate the latent functions using a linear combination of $\ell$ spline basis functions for fast inference. This proposal contrasts with existing Bayesian nonparametric approaches that use GPs to model the latent functions, which results in the usual computational bottlenecks associated with GPs. Specifically, we parameterize the coordinate functions as follows: $u_{ih}(t) = \bw_{ih}^{\top}\bb(t)$  for $1 \leq i \leq n, 1 \leq h \leq d$ and $\beta_k(t) = \bw_{k}^{\top}\bb(t)$ for  $1 \leq k \leq p$, where $\bb(t) = (b_{1}(t), \dots, b_{\ell}(t))^{\top}$ denotes a vector of known spline basis functions and $\bw_{ih}, \bw_{k} \in \Reals{\ell}$ are vectors of basis coefficients. For simplicity, we assume a common basis for all latent functions, but this could be relaxed. Although these parameterizations approximate the unknown latent functions, we will show that we can recover the true latent functions asymptotically when $\ell$ increases appropriately with network size.
 
We adopt the spline basis used by penalized splines~\citep{eilers1996} and its Bayesian counterpart~\citep{lang2004}. Specifically, we choose $\bb(t)$ to consist of B-spline functions of degree $q$ with $K$  equally-spaced internal knots as well as boundary knots so that $\ell = K + q + 1$.  In this article, we set $q = 3$, so that $\bb(t)$ is a cubic B-spline basis, which is a common choice, although this specific degree value is not necessary. 

Lastly, we define some more notation. We collect the latent trajectories' basis coefficients in the tensor $\mathcal{W}_u \in \Reals{n \times d \times \ell}$ with $(i,h)$-th tube-fibers $[\mathcal{W}_u]_{ih} = \bw_{ih}$ and the coefficient function's basis coefficients in the matrix $\bW_{\beta} = (\bw_1, \dots, \bw_k)^{\top} \in \Reals{p \times \ell}$. We denote the collection of all basis coefficients as $\mathcal{W} = \set{\mathcal{W}_u, \bW_{\beta}}$. Also, we use $p(\mathcal{Y} \mid \mathcal{W}, \mathcal{X})$ to denote the Bernoulli likelihood obtained by substituting the spline approximations into model (\ref{eq:dynlsm})--(\ref{eq:dynlsm_lr}). 

\section{Prior Specification}\label{sec:priors}

The success of the proposed approximation relies on a prior for the basis coefficients that allows the posterior to adapt to misspecifications of the model's two primary parameters: the latent space dimension $d$ and the basis dimension $\ell$. The latent space dimension $d$ controls the model's static complexity, that is, the expressiveness of the log-odds matrix $\bTheta_{t_m}$ at each time point $t_m$. The basis dimension $\ell$ controls the model's dynamic complexity, that is, the temporal variability of the latent functions. Accordingly, it is crucial to select a prior that penalizes both levels of complexity so that the posterior can determine the appropriate level for the observed dynamic network to avoid underfitting and overfitting.

\subsection{The P-Spline Prior for Dynamic LSMs}\label{subsec:pspline}

Based on the Bayesian approach to P-splines developed by \citet{lang2004}, we propose the following Gaussian Markov random field (GMRF) priors for the basis coefficients suited for controlling the static and dynamic complexity of dynamic LSMs:
\begin{align}
    &u_{ih}(t) = \bw_{ih}^{\top} \bb(t), \quad \bw_{ih} \indsim N(\mathbf{0}_{\ell}, \gamma_h^{-1} \bOmega_{i}^{-1}), \quad 1 \leq i \leq n, \quad 1 \leq h \leq d, \label{eq:pspline_u}\\
    &\beta_k(t) = \bw_{k}^{\top}\bb(t), \quad \bw_{k} \indsim N(\mathbf{0}_{\ell}, \bOmega_{\beta_k}^{-1}), \quad 1 \leq k \leq p, \label{eq:pspline_b}
\end{align}
where 
\begin{equation}\label{eq:precision}
    \bOmega_{i} = \frac{\bD_{\ell}^{(1) \, \top}\bD_{\ell}^{(1)}}{\sigma_{i}^2} + \frac{\be_{1} \be_{1}^{\top}}{\tau^2},\qquad
    \bOmega_{\beta_k} = \frac{\bD_{\ell}^{(r_k) \, \top}\bD_{\ell}^{(r_k)}}{\sigma_{\beta_k}^2} + \sum_{s=1}^{r_k} \frac{\be_{s} \be_{s}^{\top}}{\tau^2_{\beta}}.
\end{equation}
In the previous expressions, $\bD_{\ell}^{(r)}$ is an $(\ell - r) \times \ell$ matrix representing the $r$-th order finite-difference operation acting on $\bw \in \Reals{\ell}$, $\be_s \in \Reals{\ell}$ is the $s$-th standard basis vector, and the variance parameters $\set{\gamma_h}_{h=1}^d$, $\set{\sigma_i^2}_{i=1}^n$, $\set{\sigma_{\beta_k}^2}_{k=1}^p$, $\tau^2$, and $\tau_{\beta}^2$ take on strictly positive values. We refer to this prior as the P-spline prior for dynamic LSMs.

Under Equations~(\ref{eq:pspline_u})--(\ref{eq:precision}), the basis coefficients follow Gaussian random walks. For the latent trajectories, each $\bw_{ih}$ follows a first-order random walk with initial variance $\gamma_h \tau^2$ and a node-specific transition variance $\gamma_h \sigma_i^2$. Crucially, when $\gamma_h^{-1} \approx 0$, all the $u_{ih}(t)$ functions associated with the $h$-th latent space dimension will concentrate near zero. As such, priors that induce shrinkage of $\gamma_h^{-1}$ to zero can potentially control the model's static complexity by removing unnecessary dimensions. For the coefficient function, each $\bw_k$ follows a $r_k$-th order random walk  with the initial $r_k$ basis coefficients having initial variances $\tau_{\beta}^2$ and a covariate-specific transition variance $\sigma_{\beta_k}^2$. We fix $\tau^2 = 1$ to identify the shrinkage parameters $\set{\gamma_h^{-1}}_{h=1}^d$ and $\tau^2_{\beta} = 100$ to induce a broad prior on the coefficient function. 

The $r$-th order random walk priors on the basis coefficients shrink the associated latent functions towards a polynomial of degree $r - 1$, which controls the model's dynamic complexity by promoting smoothness in the latent functions. The transition variances determine the amount of deviation from this base polynomial. For example, a first-order random walk forces shrinkage towards a constant function, and a second-order random walk forces shrinkage towards a linear trend. We place first-order random walk priors on the basis coefficients associated with the latent trajectories because we expect these functions to be constant over time a priori. However, for the coefficient function, we allow mixed orders that can differ between coordinates depending on the application. Lastly,  we note that the proposed prior implicitly places non-stationary GP priors on the latent functions whose properties we describe in Appendix~\ref{subsec:gp} of the supplementary material.

\subsection{Prior Specification for the Variance Parameters}\label{subsec:var_priors}

Next, we define the priors for the variance parameters. We place a multiplicative Gamma process prior~\citep{bhattacharya2011} on the shrinkage parameters, which has been used for Bayesian learning of the latent space dimension in previous LSMs~\citep{durante2014, gwee2023}. The prior takes the following form 
\[
    \gamma_h = \prod_{s=1}^h \nu_{s}, \qquad \nu_1 \indsim \text{Gamma}(a_1, 1), \qquad \nu_{s} \iidsim \text{Gamma}(a_2, 1), \qquad 2 \leq s \leq d.
\]
As shown in \citet{bhattacharya2011}, under this prior, the shrinkage parameters $\gamma_h^{-1}$ are stochastically decreasing towards zero as $h$ increases when $a_2 > 1$, which allows the posterior to effectively remove unnecessary dimensions. As suggested by \citet{durante2017letters}, we set $a_1 = 2$ and $a_2 = 3$, which performs well overall, especially when $d$ is small.

For the transition variances, we require priors that promote shrinkage towards zero to control the model's dynamic complexity. For this reason, we adopt priors
\[
    \sigma_i^2 \iidsim \text{Gamma}(c_{\sigma}/2, d_{\sigma}/2), \quad i = 1, \dots, n, \quad \sigma_{\beta_k}^2 \iidsim \text{Gamma}(c_{\sigma}/2, d_{\sigma}/2), \qquad k = 1, \dots, p.
\]
In this work, we fix $c_{\sigma} = 2$ and $d_{\sigma} = 1$. Unlike traditional inverse-gamma priors~\citep{simpson2017}, we show that these gamma priors place sufficient mass around zero to appropriately regularize the latent trajectories. Furthermore, the use of gamma priors on low-level variance parameters has been used in discrete-time dynamic LSMs~\citep{zhao2022} and Bayesian hierarchical models~\citep{gelman2006} to better promote shrinkage towards zero.

\section{Variational Inference}\label{sec:theory}

We introduced the P-spline prior for dynamic LSMs in order to construct a fast inference procedure that can recover the true latent functions asymptotically based on estimates of the basis coefficients $\mathcal{W}$ and the variance parameters $\brho = \set{\set{\nu_h}_{h=1}^d, \set{\sigma_i^2}_{i=1}^n, \set{\sigma_{\beta_k}^2}_{k=1}^p}$. To this end, we adopt the fractional posterior framework~\citep{walker2001, bhattacharya2019}, where a fractional power of the likelihood, $\set{p(\mathcal{Y} \mid \mathcal{W}, \mathcal{X})}^{\alpha}$ for $\alpha \in (0, 1]$, is combined with a prior using the usual Bayes formula to arrive at a fractional posterior $p_{\alpha}(\mathcal{W}, \brho \mid \mathcal{Y}, \mathcal{X}) \propto \set{p(\mathcal{Y} \mid \mathcal{W}, \mathcal{X})}^{\alpha} p(\mathcal{W} \mid \brho)p(\brho)$.  For $\alpha = 1$, we recover the usual posterior; however, fractional posteriors with purely fraction powers $(\alpha < 1)$ require less conditions than the usual posterior to ensure consistent point estimation. For scalability, we construct estimates based on a variational approximation to this fractional posterior.

\subsection{The \texorpdfstring{$\alpha$}{alpha}-Variational Posterior}

In general, variational inference approximates the posterior distribution by its closest member in a pre-specified parametric family of distributions $\mathcal{Q}$ with respect to the KL divergence. Variational approximations of fractional posteriors have recently gained popularity~\citep{alquier2020, yang2020}. In this setting, we seek a parametric distribution that approximates the fractional posterior distribution, that is,
\begin{align}\label{eq:kl_vi}
    \hat{q}(\mathcal{W}, \brho) &= \argmin_{q(\mathcal{W}, \brho) \in \mathcal{Q}} D_{KL}\set{q(\mathcal{W}, \brho) \mid \mid p_{\alpha}(\mathcal{W}, \brho \mid \mathcal{Y}, \mathcal{X})} \nonumber \\
    &= \argmax_{q(\mathcal{W}, \brho) \in \mathcal{Q}} \mathbb{E}_{q(\mathcal{W}, \brho)}\left[\log\left\{\frac{p_{\alpha}(\mathcal{Y}, \mathcal{W}, \brho \mid \mathcal{X})}{q(\mathcal{W}, \brho)}\right\}\right], 
\end{align}
where the second objective function is called the evidence-lower bound (ELBO) and $\hat{q}(\mathcal{W}, \brho)$ is the $\alpha$-variational posterior, which equals the traditional variational posterior when $\alpha = 1$.

To complete the variational specification, we select a variational family $\mathcal{Q}$. We choose
\begin{equation}\label{eq:var_fam}
    \mathcal{Q} = \left\{q(\mathcal{W}, \brho) \, : \, q(\mathcal{W}, \brho) =  \prod_{k=1}^p q(\bw_k) q(\sigma_{\beta_k}^2) \prod_{i=1}^n \left[q(\sigma_i^2) \prod_{h=1}^d q(\bw_{ih}) \right]\prod_{h=1}^d q(\nu_h) \right\}.
\end{equation}
This variational family only maintains the dependencies between the basis coefficients associated with a single spline approximation. Importantly, there is no dependence between the basis coefficients and the variance parameters so that $q(\mathcal{W}, \brho) = q(\mathcal{W})q(\brho)$. Next, we show that this variational family is sufficient to recover the true latent functions asymptotically with a rate that is adaptive to the overall dynamic complexity of the latent functions.

\subsection{Theoretical Properties of the \texorpdfstring{$\alpha$}{alpha}-Variational Posterior}

We establish non-asymptotic consistency results as $n$ and $M$ grow for estimates constructed from global variational solution under the P-spline prior for dynamic LSMs. Specifically, we present error bounds for the $\alpha$-variational posterior means at the observed time points, that is, $\widehat{\bTheta}_{t_m} = \mathbb{E}_{\hat{q}(\mathcal{W}, \brho)}[\bTheta_{t_m}]$, $\hat{\bU}(t_m) = \mathbb{E}_{\hat{q}(\mathcal{W}, \brho)}[\bU(t_m)]$, and $\hat{\bbeta}(t_m) = \mathbb{E}_{\hat{q}(\mathcal{W}, \brho)}[\bbeta(t_m)]$, where $\mathbb{E}_{\hat{q}(\mathcal{W}, \brho)}[\cdot]$ denotes an expectation with respect to the $\alpha$-variational posterior defined in Equations (\ref{eq:kl_vi})--(\ref{eq:var_fam}). Furthermore, we show that this bound adapts to the variability of the true latent functions.

We assume that the observed dynamic network data $\mathcal{Y}$ is generated from model (\ref{eq:dynlsm})--(\ref{eq:dynlsm_lr}) with true latent functions $\bbeta_0(t) = (\beta_{01}(t), \dots, \beta_{0k}(t))^{\top}$ and $\bu_{0i}(t) = (u_{0i1}(t), \dots, u_{0id}(t))^{\top}$ for $1 \leq i \leq n$. Also, we let $\bU_0(t) = (\bu_{01}(t)^{\top}, \dots, \bu_{0n}(t)^{\top})^{\top}$ and $\set{\bTheta_{0t_m}}_{m=1}^M$ denote the true dyad-wise log-odds matrices at the observed time points.   Let $\mathbb{P}_0$ be the probability measure under this true data-generating process. We assume the latent space dimension $d$ is fixed and known, and the number of dyadic covariates $p$ is fixed. In the remainder of this section, we let $\lambda$ denote the Lebesgue measure on the unit interval. Below, we outline the assumptions about the true latent functions and covariates sufficient to achieve our results.

\begin{assump}\label{assump:func_space}
    The true latent functions belong to the Sobolev space $L_{\infty}^1[0,1] = \{f : [0, 1] \rightarrow \Reals{}: \, f \text{ is absolutely continuous on } [0,1] \text{ and } \norm{f'}_{L_{\infty}[0, 1]} < \infty\}$, that is, $u_{0ih} \in L_{\infty}^1[0,1]$ and $\beta_{0k} \in L_{\infty}^1[0,1]$ for $1 \leq i \leq n, 1 \leq h \leq d$, and $1 \leq k \leq p$.
\end{assump}
\begin{assump}\label{assump:lp}
    The true latent functions have Lipschitz constants that are upper bounded by a constant that is independent of network size. That is, let $\max_{1 \leq k \leq p}\norm{\beta_{0k}'}_{L_{\infty}[0,1]} = L_{\beta}$ and $\max_{1 \leq i \leq n, 1 \leq h \leq d}\norm{u_{0ih}'}_{L_{\infty}[0,1]} = L_u$, then $L = \max(L_{\beta}, L_u) = O(1)$.
\end{assump}
\begin{assump}\label{assump:x}
    The dyadic covariates are upper bounded by a constant that is independent of network size, that is, $\max_{i,j,m} \norm{\bx_{ij,t_m}}_2 \leq K_x$ for some constant $K_x > 0$.
\end{assump}

Assumption~\ref{assump:func_space} requires the latent functions to be sufficiently smooth, a common condition in the literature on nonparametric regression. In particular, $L^1_{\infty}[0,1]$ is equivalent to the space of almost-everywhere differentiable Lipschitz continuous functions on the unit interval. Assumption~\ref{assump:lp} places an upper bound on the maximum variation of the true latent trajectories. In particular, since $\set{t_m}_{m=1}^M \subseteq [0, 1]$, we have that the total variation in the latent trajectories $\sum_{m=2}^M \sum_{i=1}^n \norm{\bu_{0i}(t_m) - \bu_{0i}(t_{m-1})}_2 = O(n)$. As such, for a fixed $n$, the distance the latent positions travel between time points should decrease as $M$ increases. This behavior is reasonable when we observe the network over an increasingly finer grid of time points but may not be reasonable if we observe the network over a progressively longer period of time. Assumption~\ref{assump:x} is used to bound the entries of the log-odds matrices which is common in the LSM literature~\citep{wu2017, ma2020}. Importantly, Assumption \ref{assump:func_space} and Assumption \ref{assump:x} imply that the networks are dense. 

With these assumptions, we state the non-asymptotic error bound for the recovery of the true log-odds matrices under the $\alpha$-variational posterior at the observed time points.
\begin{theorem}[Error bound for the log-odds under $\alpha$-VB]
    \label{thm:vb_consistency}
    Suppose the true data generating process satisfies model (\ref{eq:dynlsm})--(\ref{eq:dynlsm_lr}) with true latent functions $\bU_0(t)$ and $\bbeta_0(t)$ and observed covariates $\mathcal{X}$ that satisfy Assumptions \ref{assump:func_space}--\ref{assump:x}. Then, under the priors defined in Equations~(\ref{eq:pspline_u})--(\ref{eq:precision}) with $r_1 = \dots =  r_p = 1$ and $\bb(t)$ a B-spline basis of dimension $\ell \asymp (nM)^{1/5}$, we have with $\mathbb{P}_0$-probability tending to one as $n,M \rightarrow \infty$ that for $\lambda$-almost all $\set{t_m}_{m=1}^M$ and any $\alpha \in (0,1)$
    \[
        \frac{1}{M n^2} \sum_{m=1}^M \norm{\hat{\bTheta}_{t_m} - \bTheta_{0t_m}}_F^2 \lesssim \max\left\{\left(\frac{L}{nM}\right)^{2/5},  \frac{\log nM}{nM}\right\}.
    \]
\end{theorem}

\noindent As desired, Theorem~\ref{thm:vb_consistency} shows that point estimates constructed from the $\alpha$-variational posterior under the P-spline prior for dynamic LSMs have an error rate that is adaptive to the variation in the true latent functions. Specifically, for fixed $n$ and $M$, the rate is an increasing function in $L$, implying that less variable functions lead to better rates. However, the rate cannot be faster than $\log(nM)/nM$, which is the minimax rate up to a logarithmic factor for recovering a matrix of static latent positions given $O(n^2M)$ observations. 

To transfer the error bound in Theorem~\ref{thm:vb_consistency} for the recover of the log-odds matrices to the recovery of the latent functions, we require an additional identifiability condition. To this end, we define the following quantity:
\begin{align}
    &r(\mathcal{X}) = \min_{m = 1, \dots, M} \left\{\sup_{\bbeta \in \Reals{p}} \frac{\norm{\mathcal{X}_{t_m} \, \bar{\times}_3 \, \bbeta}_{op}^2}{\norm{\mathcal{X}_{t_m} \, \bar{\times}_3 \, \bbeta}_{F}^2}\right\}^{-1}\label{eq:stable_rank_t},
\end{align}
where $\bar{\times}_3$ denotes tensor-vector multiplication in the 3rd mode. In the case of a single dyadic covariate, we can drop the supremum in Equation~(\ref{eq:stable_rank_t}) as $\norm{\mathcal{X}_{t_m}}_F^2 / \norm{\mathcal{X}_{t_m}}_{op}^2$ is the stable rank of $\mathcal{X}_{t_m}$, so $r(\mathcal{X})$ becomes the minimum stable rank of the dyadic-covariate matrices over all time points. We make the following assumption relating $r(\mathcal{X})$ to $d$.
\begin{assump}\label{assump:stable_rank}
    For $n$ and $M$ large enough, $r(\mathcal{X}) > 2 d$. 
\end{assump}
\noindent When we observe a single network, Assumption~\ref{assump:stable_rank} reduces to an existing condition used for the identifiability of static network LSMs~\citep{ma2020, wu2017}. For dynamic networks, we require this condition to hold for all observed time points. 

With this additional assumption, we have the following non-asymptotic error bounds for the $\alpha$-variational posterior mean estimates of the latent functions.
\begin{corollary}[Error bounds for $\bU_0(t)$ and $\bbeta_0(t)$ under $\alpha$-VB]\label{thm:vb_function_consistency}
    Suppose the same conditions as Theorem~\ref{thm:vb_consistency} and that Assumption~\ref{assump:stable_rank} holds. Define $\kappa_{\mathcal{X},d} = 1 - \sqrt{2d/r(\mathcal{X})}$ and $\sigma_{min}(\bU_0) = \min_{m = 1, \dots, M}\sigma_{min}\{\bU_0(t_m)\}$. Then, we have with $\mathbb{P}_0$-probability tending to one as $n,M \rightarrow \infty$ that for $\lambda$-almost all $\set{t_m}_{m=1}^M$ and any $\alpha \in (0, 1)$
    \begin{align*}
        \frac{1}{Mnd}\sum_{m=1}^M \min_{\bO_m \in \mathcal{O}_d} \norm*{\hat{\bU}(t_m) -  \bU_0(t_m)\bO_m}_F^2 &\lesssim \frac{n}{\kappa_{\mathcal{X}, d} \, \sigma_{min}^2(\bU_0)} \, \max\left\{\left(\frac{L}{nM}\right)^{2/5},  \frac{\log nM}{nM}\right\}, \\ 
        \frac{1}{M n^2}\sum_{m=1}^M \sum_{1 \leq i,j \leq n} \left[\{\hat{\bbeta}(t_m) - \bbeta_{0}(t_m)\}^{\top}\bx_{ij,t_m}\right]^2 &\lesssim \frac{1}{\kappa_{\mathcal{X},d}} \, \max\left\{\left(\frac{L}{nM}\right)^{2/5},  \frac{\log nM}{nM}\right\}.
    \end{align*}
\end{corollary}
\noindent Corollary~\ref{thm:vb_function_consistency} gives the estimation error for the remaining identifiable quantities in the model. Specifically, LSMs with inner-product similarity functions are well known to be only identifiable up to an orthogonal transformation of the latent positions. As such, the error in the latent trajectories is stated up to a collection of orthogonal transformations that can change between time points. According to the bound in Corollary~\ref{thm:vb_function_consistency}, a sufficient condition for the recovery of the latent trajectories is that $\sigma_{min}(\bU_0)^2 \asymp n$. This scaling requirement is common in static LSMs~\citep{ma2020}, and for example holds when $d \ll n$ and the entries of $\bU_0(t_m)$ are i.i.d. random variables with bounded variance for all $1 \leq m \leq M$. In addition, the coefficient functions are identifiable up to the linear predictors, e.g.,  $\bx_{ij,t_m}^{\top} \bbeta(t_m)$. To transfer the bound to the coefficient functions themselves would require conditions on the distribution of the covariates to avoid collinearity. Overall, the error rates remain adaptive to the underlying variation in the true latent functions.

\section{Estimation}\label{sec:estimation}

    Next, we develop a stochastic variational inference (SVI) algorithm~\citep{hoffman2013} for computing the $\alpha$-variational posterior that scales to large networks. We assume familiarity with SVI; however, we review the essential concepts in Appendix~\ref{sec:svi_overview} of the supplement. 

\subsection{P\'olya-Gamma Augmentation for \texorpdfstring{$\alpha$}{alpha}-Variational Bayes}\label{subsec:pg}

An immediate problem with finding the $\alpha$-variational posterior defined in Equations~(\ref{eq:kl_vi})--(\ref{eq:var_fam}) is that its optimal factors are not members of known parametric families. Furthermore, the SVI framework proposed by \citet{hoffman2013} requires the model parameters' full-conditional distributions to be in the exponential family, which the proposed model does not satisfy. To solve both problems, we propose a new P\'{o}lya-gamma augmentation scheme~\citep{polson2013, choihobert2013} that produces optimal closed-form $\alpha$-variational posteriors in a large class of logistic models. A possible alternative to this augmentation scheme is the tangent-transform approach proposed by \citet{jaakkola2000}, which has been used to obtain $\alpha$-variational posteriors for discrete-time dynamic LSMs~\citep{zhao2022}. However, the tangent-transform cannot be used to derive an SVI algorithm because it lacks the necessary probabilistic interpretation. 

Under our proposed P\'olya-gamma augmentation scheme, we introduce a set of local latent P\'{o}lya-gamma random variables associated with each dyad in the network, that is, for $1 \leq i \leq j \leq n$ and $1 \leq m \leq M$, we introduce
\[
    \omega_{ij,t_m} = \omega_{ji,t_m} \iidsim \PG(\alpha, [\bTheta_{t_m}]_{ij}), 
\]
so that the augmented likelihood is
\begin{align*}
    p(\mathcal{Y}, \bomega \mid \mathcal{W}, \mathcal{X}) = p(\mathcal{Y} \mid \mathcal{W}, \mathcal{X}) p_{\alpha}(\bomega \mid \mathcal{W} ,\mathcal{X}) = p(\mathcal{Y} \mid \mathcal{W}, \mathcal{X}) \prod_{m=1}^M \prod_{i \leq j} \PG(\omega_{ij,t_m} \mid \alpha, [\bTheta_{t_m}]_{ij}).
\end{align*}
In the previous expressions, $\PG(b,c)$ and $\PG(\omega \mid b, c)$ denote the distribution and density of a P\'{o}lya-gamma random variable with parameters $b > 0$ and $c \in \Reals{}$ and $\bomega$ denotes the collection of all P\'{o}lya-gamma latent variables in the model. When $\alpha = 1$, this scheme recovers standard P\'{o}lya-gamma augmentation, which has been used for Bayesian inference of existing dynamic LSMs~\citep{durante2014, sewell2017}.

For inference, we consider the augmented fractional posterior density $p_{\alpha}(\mathcal{W}, \brho, \bomega \mid \mathcal{Y}) \propto \set{p(\mathcal{Y} \mid \mathcal{W}, \mathcal{X})}^{\alpha} p_{\alpha}(\bomega \mid \mathcal{W}, \mathcal{X}) p(\mathcal{W}\mid \brho)p(\brho)$. Clearly, the marginal fractional posterior density obtained by integrating out the local P\'{o}lya-gamma latent variables is the original fractional posterior analyzed in Section~\ref{sec:theory}. Furthermore, the parameters have full-conditional distributions in the exponential family under the augmented model, which we use to derive an SVI algorithm in Section~\ref{subsec:svi}. As such, we seek a variational approximation to this augmented fractional posterior by maximizing the corresponding ELBO
\begin{align}\label{eq:aug_elbo}
    \hat{q}(\mathcal{W}, \brho)\hat{q}(\bomega) = \argmax_{q(\mathcal{W}, \brho)q(\bomega) \in \mathcal{Q} \times Q_{\bomega}} \mathbb{E}_{q(\mathcal{W},\brho)q(\bomega)} \left[\log\left\{\frac{p_{\alpha}(\mathcal{Y}, \bomega, \mathcal{W}, \brho \mid \mathcal{X})}{q(\mathcal{W}, \brho)q(\bomega)}\right\}\right],
\end{align}
where
\begin{align}\label{eq:augmented_lik}
    p_{\alpha}(\mathcal{Y}, \bomega, \mathcal{W}, \brho \mid \mathcal{X}) = \left[\prod_{m=1}^M \prod_{i \leq j} p_{\alpha}(y_{ij,t_m}, \omega_{ij,t_m} \mid [\bTheta_{t_m}]_{ij})\right] p(\mathcal{W} \mid \brho) p(\brho),
\end{align}
and $p_{\alpha}(y_{ij,t_m}, \omega_{ij,t_m} \mid [\bTheta_{t_m}]_{ij}) =\set{e^{y_{ij,t_m}[\bTheta_{t_m}]_{ij}}/(1 + e^{[\bTheta_{t_m}]_{ij}})}^{\alpha}\PG(\omega_{ij,t_m} \mid \alpha, [\bTheta_{t_m}]_{ij})$. We denote the ELBO in Equation~(\ref{eq:aug_elbo}) by $\textsf{ELBO}[q(\mathcal{W}, \brho)q(\bomega)]$ to reflect its dependence on the variational posterior. We select $\mathcal{Q}_{\bomega}$ as a mean-field variational family, that is, $Q_{\bomega} = \set{q(\bomega) \, : \, q(\bomega) = \prod_{m=1}^M \prod_{i \leq j} q(\omega_{ij,t_m})}$, and keep $\mathcal{Q}$ as defined in Equation~(\ref{eq:var_fam}). Furthermore, we set each variational factor to its optimal parametric form, that is, the same exponential family as its associated parameter's full-conditional distribution~\citep{bishop2006}.

\subsection{The Stochastic Variational Inference Algorithm}\label{subsec:svi}

For scalable estimation, SVI optimizes the ELBO through stochastic gradient ascent~\citep{robbinsmonro1951}. To motivate the algorithm, we re-express the ELBO as
\begin{align}\label{eq:new_elbo}
    \textsf{ELBO}[q(\mathcal{W}, \brho)q(\bomega)] &= \sum_{m=1}^M \sum_{i \leq j} \mathbb{E}_{q(\mathcal{W})q(\omega_{ij,t_m})}\left\{\log p_{\alpha}(y_{ij,t_m}, \omega_{ij,t_m} \mid [\bTheta_{t_m}]_{ij}) - \log q(\omega_{ij,t_m})\right\} \nonumber \\
    &\qquad\qquad- D_{KL}\{q(\mathcal{W}, \brho) \mid \mid p(\mathcal{W}\mid\brho)p(\brho)\}.
\end{align}
A computational bottleneck when calculating the gradient of this objective is the summation over all time points and dyads, which has a computational complexity of $O(Mn^2)$. SVI reduces this computational cost by using an unbiased estimate of the gradient that is faster to compute. In addition, SVI achieves further computational gains by replacing the estimate of the standard gradient with an estimate of the natural gradient~\citep{amari1982}. 

To form an unbiased natural gradient estimate, we replace the summation in Equation (\ref{eq:new_elbo}) with a summation over a random subsample of time points and dyads. We use $\mathcal{N}_{i, t_m} = \set{j \, : \, y_{ij, t_m} = 1, 1 \leq j \leq n}$ to denote the neighborhood of node $i$ at time $t_m$, so that $\mathcal{N}_{i,t_m}^c$ is the set of nodes not connected to node $i$ at time $t_m$. Often networks get sparser as $n$ grows, so that formally $\abs{\mathcal{N}_{i,t_m}} \ll \abs{\mathcal{N}_{i,t_m}^c}$. On the other hand, the summation can still be computationally demanding for moderately sized $n$ when the dynamic network contains many time points $M$. As such, we construct an unbiased estimate of the ELBO by randomly sampling both non-edges and time points according to the following proposition.

\begin{proposition}\label{prop:unbiased_grads}
    Consider the following summations
    \[
        H_i = \sum_{m=1}^M \sum_{j = 1}^n h_{ij,t_m}, \qquad H = \sum_{m=1}^M \sum_{i \leq j} h_{ij,t_m} = \frac{1}{2} \sum_{i=1}^n H_i,
    \]
    where $h_{ij,t_m} = h_{ji,t_m} \in \Reals{}$ for $1 \leq i \leq j \leq n$ and $1 \leq m \leq M$. Let $\mathcal{M} \subseteq \set{1, \dots, M}$ denote a uniform random sample without replacement of time points and $\mathcal{N}_{i,t_m}^{c \, *} \subseteq \mathcal{N}_{i,t_m}^c$ denote a uniform random sample without replacement of nodes not connected to node $i$  conditioned on the event $m \in \mathcal{M}$ and the empty set otherwise, then an unbiased estimator of $H_i$ is 
    \begin{align}
        \mathcal{B}_i(H_i) &= \frac{M}{\abs{\mathcal{M}}} \sum_{m \in \mathcal{M}} \left(\sum_{j \in \mathcal{N}_{i,t_m}} h_{ij,t_m} + \frac{\abs{\mathcal{N}_{i,t_m}^c}}{\abs{\mathcal{N}_{i,t_m}^{c \, *}}} \sum_{j \in \mathcal{N}_{i,t_m}^{c \, *}} h_{ij,t_m}\right).
    \end{align}
   Furthermore, $\mathcal{B}(H) = \frac{1}{2} \sum_{i=1}^n \mathcal{B}_i(H_i)$ is an unbiased estimator of $H$.
\end{proposition}

Applying Proposition~\ref{prop:unbiased_grads} to the summation in Equation~(\ref{eq:new_elbo}), we arrive at the following unbiased estimator of the ELBO
\begin{align}
    \widehat{\textsf{ELBO}}[q(\mathcal{W}, \brho)q(\bomega)] &= \frac{1}{2} \frac{M}{\abs{\mathcal{M}}}\sum_{m \in \mathcal{M}} \sum_{i=1}^n \left(\sum_{j \in \mathcal{N}_{i,t_m}} e_{ij,t_m} + \frac{\abs{\mathcal{N}_{i,t_m}^c}}{\abs{\mathcal{N}_{i,t_m}^{c \, *}}} \sum_{j \in \mathcal{N}_{i,t_m}^{c \, *}} e_{ij,t_m} \right) \nonumber \\ 
    &\qquad- D_{KL}\set{q(\mathcal{W}, \brho) \mid\mid p(\mathcal{W} \mid \brho) p(\brho)}, \label{eq:stochastic_elbo}
\end{align}
where $e_{ij,t_m} = \mathbb{E}_{q(\mathcal{W})q(\omega_{ij,t_m})}\{\log p_{\alpha}(y_{ij,t_m}, \omega_{ij,t_m} \mid [\bTheta_{t_m}]_{ij}) - \log q(\omega_{ij,t_m})\}$. It is possible to use other subsampling schemes to construct an unbiased estimate.  For example, in an MCMC algorithm for static LSMs, \citet{raftery2012} postulated that uniformly subsampling non-edges might misrepresent the network structure. As such, they proposed a stratified sampling scheme based on shortest path lengths. In addition, in an SVI algorithm for static LSMs, \citet{aliverti2022} used an adaptive sampling scheme that stratified non-edges based on the current parameter estimates. However, they found that the computational cost of constructing these subsamples was rarely worth the gain in performance. Therefore, we settle for a fast sampling scheme that performs well in practice.

Proposition \ref{prop:lp_weights} and Proposition \ref{prop:beta_weights} derive unbiased estimators for the natural gradients of $q(\mathcal{W})$'s parameters based on Equation~(\ref{eq:stochastic_elbo}). The calculations involve performing Bayesian linear regression-type updates using subsamples of the time points and dyads. Because the variance parameters only appear in the KL divergence term in Equation~(\ref{eq:stochastic_elbo}), which does not depend on the subsample, their variational factors are updated using full (non-stochastic) natural gradients presented in Appendix~\ref{sec:init} of the supplementary material.

\begin{proposition}\label{prop:lp_weights}
    For $1 \leq i \leq n$ and $1 \leq h \leq d$, under the variational family $\mathcal{Q} \times \mathcal{Q}_{\bomega}$ defined in Section~\ref{subsec:pg}, the optimal factor $q(\bw_{ih})$ has the form $N(\bmu_{\bw_{ih}}, \bSigma_{\bw_{ih}})$ with natural parameters $\blambda_{ih} \in \Reals{\ell}$ and $\bLambda_{ih} \in \Reals{\ell \times \ell}$, that is, $\bmu_{ih} = \bLambda_{ih}^{-1} \blambda_{ih}$ and $\bSigma_{\bw_{ih}} = \bLambda_{ih}^{-1}$.  Also, unbiased estimators of the natural gradients are $\hat{\nabla}_{\blambda_{ih}} \mathsf{ELBO}[q(\mathcal{W}, \brho)q(\bomega)] = -\blambda_{ih} + \mathcal{B}_i(\bar{\blambda}_{ih})$ and $\hat{\nabla}_{\bLambda_{ih}} \mathsf{ELBO}[q(\mathcal{W}, \brho)q(\bomega)] = -\bLambda_{ih} + \mathbb{E}_{q(\brho)}[\gamma_h \bOmega_{i}] + \mathcal{B}_i(\bar{\bLambda}_{ih})$, where
    \begin{align*}
        \bar{\blambda}_{ih} &=  \sum_{m=1}^M \sum_{j = 1}^n \bigg[\alpha (y_{ij,t_m} - 1/2) -\mathbb{E}_{q(\omega_{ij,t_m})}[\omega_{ij,t_m}]\xi_{ij,t_m} \bigg]  \mathbb{E}_{q(\bw_{jh})}[u_{jh}(t_m)] \bb(t_m), \\
        \bar{\bLambda}_{ih} &=  \sum_{m=1}^M \sum_{j =1}^n \mathbb{E}_{q(\omega_{ij,t_m})}[\omega_{ij,t_m}] \mathbb{E}_{q(\bw_{jh})}[u_{jh}(t_m)^2] \bb(t_m) \bb(t_m)^{\top},
    \end{align*}
    $\xi_{ij,t_m} = \sum_{k=1}^p \mathbb{E}_{q(\bw_k)}[\beta_k(t_m)] x_{ijk,t_m} + \sum_{g \neq h} \mathbb{E}_{q(\bw_{ig})}[u_{ig}(t_m)] \mathbb{E}_{q(\bw_{jg})}[u_{jg}(t_m)]$, and $\mathcal{B}_i(\cdot)$ is constructed based on a random sample of dyads and time points as in Proposition~\ref{prop:unbiased_grads}.
\end{proposition}

\begin{proposition}\label{prop:beta_weights}
    For $1 \leq k \leq p$, under the variational family $\mathcal{Q} \times \mathcal{Q}_{\bomega}$ defined in Section~\ref{subsec:pg}, the optimal variational distribution $q(\bw_{k})$ is $N(\bmu_{\bw_{k}}, \bSigma_{\bw_{k}})$ with natural parameters $\blambda_{k} \in \Reals{\ell}$ and $\bLambda_{k} \in \Reals{\ell \times \ell}$, that is, $\bmu_{k} = \bLambda_{k}^{-1} \blambda_{k}$ and $\bSigma_{\bw_{k}} = \bLambda_{k}^{-1}$. Also, unbiased estimators of the natural gradients are $\hat{\nabla}_{\blambda_{k}} \mathsf{ELBO}[q(\mathcal{W}, \brho)q(\bomega)] = -\blambda_{k} + \mathcal{B}(\bar{\blambda}_{k})$ and $\hat{\nabla}_{\bLambda_k} \mathsf{ELBO}[q(\mathcal{W}, \brho)q(\bomega)] = -\bLambda_k + \mathbb{E}_{q(\brho)}[\bOmega_{\beta_k}] + \mathcal{B}(\bar{\bLambda}_k)$, where
\begin{align*}
    \bar{\blambda}_{k} &=\sum_{m=1}^M \sum_{i \leq j} \bigg[\alpha (y_{ij,t_m} - 1/2) - \mathbb{E}_{q(\omega_{ij,t_m})}[\omega_{ij,t_m}] \nu_{ij,t_m}\bigg] x_{ijk,t_m}\bb(t_m), \\
    \bar{\bLambda}_{k} &=  \sum_{m=1}^M \sum_{i \leq j} \mathbb{E}_{q(\omega_{ij,t_m})}[\omega_{ij,t_m}] x_{ijk,t_m}^2 \bb(t_m) \bb(t_m)^{\top},
\end{align*}
    $\nu_{ij,t_m} = \sum_{g \neq k} \mathbb{E}_{q(\bw_{g})}[\beta_{g}(t_m)] x_{ij\ell,t_m} + \sum_{h=1}^d \mathbb{E}_{q(\bw_{ih})}[u_{ih}(t_m)] \mathbb{E}_{q(\bw_{jh})}[u_{jh}(t_m)]$, and $\mathcal{B}(\cdot)$ is constructed based on a random sample of dyads and time points as in Proposition~\ref{prop:unbiased_grads}.
\end{proposition}

Algorithm~\ref{alg:svi} presents our proposed SVI algorithm for obtaining the $\alpha$-variational posterior in Equation~(\ref{eq:aug_elbo}), which can be easily modified to exclude self-loops if necessary.  Appendix~\ref{sec:init} of the supplementary material discusses technical details concerning initialization, the stopping criteria, and post-processing to address identifiability issues. In addition, we derive the algorithm in Appendix~\ref{sec:der_svi} of the supplement.  The algorithm sets the step size using the step size schedule proposed by \citet{hoffman2013}, where $\kappa \in (0.5, 1)$ and $\tau > 0$.  The hyperparameters $0 < m_0 \leq M$, and $n_{0i,t_m} \geq 0$ control the subsample size used to construct the stochastic natural gradients. Specifically, $m_0$ and $n_{0i,t_m}$ are the number of time points included in the subsample and the number of  non-edges associated with node $i$ at time $t_m$ included in the subsample, respectively. 

\begin{myalgorithm}
    \begin{framed}
    Given the previous parameters at step $s$, update the current parameters as follows:
    \begin{enumerate}
        \item Set the step size $\rho_s = (s + \tau)^{-\kappa}$.
        \item Construct a subsample of time points and non-edges.
        \begin{enumerate}
            \item ({\it Time point subsample}). Sample $m_0$ time points to form $\mathcal{M}$.
            \item ({\it Non-edge subsample}). For $m \in \mathcal{M}$ and $1 \leq i \leq n$, sample $n_{0i,t_m}$ nodes from $\mathcal{N}_{i,t_m}^c$ without replacement to form $\mathcal{N}_{i,t_m}^{c \, *}$. 
        \end{enumerate}
    \item Optimize the local variational parameters for the subsampled dyads. 
        
        For $m \in \mathcal{M}$, $1 \leq i \leq n$, and $j \in \mathcal{N}_{i,t_m} \cup \mathcal{N}_{i,t_m}^{c\, *}$, update 
        \[
            q(\omega_{ij,t_m}) = \PG(\alpha, c_{ij,t_m})
        \]
            using Algorithm~\ref{alg:svi_local} in Appendix~\ref{sec:init} of the supplementary material.
    \item Update $q(\bw_{ih}) = N(\bmu_{\bw_{ih}}, \bSigma_{\bw_{ih}})$ for $i \in \set{1, \dots, n}$ and $h \in \set{1, \dots, d}$.

        Update natural parameters using natural gradients defined in Proposition~\ref{prop:lp_weights}:
                    \[
                        \blambda_{ih}^{(s+1)} = (1 - \rho_s) \blambda_{ih}^{(s)} + \rho_s \mathcal{B}_i(\bar{\blambda}_{ih}^{(s)}),\quad
                        \bLambda_{ih}^{(s+1)} = (1 - \rho_s) \bLambda_{ih}^{(s)} + \rho_s \set{\mathbb{E}_{q(\brho)}[\gamma_h \bOmega_i] + \mathcal{B}_i(\bar{\bLambda}_{ih}^{(s)})},
                    \]
                    and set $\bmu_{\bw_{ih}}^{(s+1)} = \left[\bLambda_{ih}^{(s+1)}\right]^{-1} \blambda_{ih}^{(s+1)}$ and $\bSigma_{\bw_{ih}}^{(s+1)} = \left[\bLambda_{ih}^{(s+1)}\right]^{-1}$.
            \item Update $q(\bw_k) = N(\bmu_{\bw_k}, \bSigma_{\bw_{k}})$ for $k \in \set{1, \dots, p}$.

                Update natural parameters using natural gradients defined in Proposition~\ref{prop:beta_weights}:
                \[
                    \blambda_{k}^{(s+1)} = (1 - \rho_s) \blambda_{k}^{(s)} + \rho_s \mathcal{B}(\bar{\blambda}_{k}^{(s)}), \qquad
                    \bLambda_{k}^{(s+1)} = (1 - \rho_s) \bLambda_{k}^{(s)} + \rho_s \set{\mathbb{E}_{q(\brho)}[\bOmega_{\beta_k}] + \mathcal{B}(\bar{\bLambda}_{k}^{(s)})}
                \]
            and set $\bmu_{\bw_k}^{(s+1)} = \left[\bLambda_{k}^{(s+1)}\right]^{-1} \blambda_{k}^{(s+1)}$ and $\bSigma_{\bw_k}^{(s+1)} = \left[\bLambda_{k}^{(s+1)}\right]^{-1}$.
        \item Update the variance parameters using Algorithm~\ref{alg:svi_var} in Appendix~\ref{sec:init} of the supplementary material.
    \end{enumerate}
    \end{framed}
    \caption{The stochastic variational inference algorithm.}
    \label{alg:svi}
\end{myalgorithm}

To ensure Algorithm~\ref{alg:svi} scales to large networks, we use a subsample size on the order of the maximum number of edges in an observed network $E_{max} = \max_{1 \leq m \leq M} \sum_{i\leq j} y_{ij,t_m}$. To achieve this scaling, we set $m_0 = \min(\ceil{\gamma_M M}, 100)$ and $n_{0i,t_m} = \min(\floor{\gamma_n \abs{\mathcal{N}_{i,t_m}}}, \abs{\mathcal{N}_{i,t_m}^c})$, where $\gamma_M  \in (0, 1]$ and $\gamma_n \geq 1$. Under these choices, performing all natural gradient updates takes $O(m_0 E_{max})$ operations. In the sparse network setting with a fixed $M$ or, generally, when $M$ increases, we have that $m_0 E_{max} \ll Mn^2$. As such, the proposed method is much more computationally efficient than existing algorithms that process all dyadic observations.

Under this subsampling scheme, Algorithm~\ref{alg:svi} has four hyperparameters: $\kappa, \tau, \gamma_M$, and $\gamma_n$. We set $\kappa = 0.75$ and $\tau = 1$ based on the recommendation of \citet{aliverti2022} for an SVI algorithm proposed for a static LSM. Furthermore, we set $\gamma_M = 0.25$ and $\gamma_n = 2$ based on the results of a sensitivity study in Appendix~\ref{sec:more_results} of the supplement. Overall, we found that smaller values of $\gamma_n$ are preferred for very sparse networks, and the algorithm's performance was roughly the same for $\gamma_M$ values above 0.25 for $M$ as large as 500.

\section{Simulation Study}\label{sec:sim_study}

We performed a simulation study that evaluated the proposed SVI algorithm's ability to recover the model parameters and compared it to existing methods. We analyzed the algorithm's sensitivity to the subsample fractions $\gamma_n$ and $\gamma_M$ in Appendix~\ref{sec:more_results} of the supplement. 

\subsection{Simulation Settings}\label{subsec:sim_setup}

For various values of $n$ and $M$, we generated synthetic dynamic networks observed at equally spaced time points $0 = t_1 < \dots < t_M = 1$ from model (\ref{eq:dynlsm})--(\ref{eq:dynlsm_lr}) with latent functions $\set{\bU_0(t), \bbeta_0(t)}$ and a latent space dimension $d = 2$. To describe the data generating procedure, we use $\delta_{\bv}$ to denote a point mass centered at a vector $\bv$ and $\GP(0, C)$ to denote a mean-zero Gaussian process with covariance function $C$. We generated the latent trajectory of each node as $\bu_{0i}(t) = \bmu_i + \tilde{\bu}_{0i}(t)$, where $\bmu_i \iidsim (1/3) \delta_{(1.5, 0)^{\top}} + (1/3) \delta_{(-1.5, 0)^{\top}} + (1/3) \delta_{(0, 1)^{\top}}$ and $\tilde{u}_{0ih}(t) \iidsim \GP(0, C)$ for $1 \leq h \leq d$. We included an intercept and two static dyadic covariates with entries independently drawn from a standard normal distribution so that $p = 3$. The values of the intercept function at the observed time points, that is, $\set{\beta_{01}(t_m)}_{m=1}^M$, were chosen to fix the expected density of the observed networks to a given value. We generated the remaining coefficient functions as $(\beta_{02}(t), \beta_{03}(t))^{\top} = (1, -1)^{\top} + (\tilde{\beta}_{02}(t), \tilde{\beta}_{03}(t))^{\top}$, where $\tilde{\beta}_{0k}(t) \iidsim \GP(0, C)$ for $k = 2, 3$. For all GPs, we used an exponential covariance function $C(t,t') = a^2 \exp\{(t - t')^2 / 2b\}$ with standard deviation $a = 0.5$ and length scale $b = 0.2$, so that the true latent functions are relatively smooth. We excluded the adjacency matrices' diagonal entries during estimation to match the application in Section~\ref{sec:application}.

In all simulations, we estimated the $\alpha$-variational posterior under the P-spline prior for dynamic LSMs with first-order random walk GMRFs on all basis coefficients using the SVI algorithm and hyperparameter settings proposed in Section~\ref{sec:estimation}. We set $d = 6$ and the fractional power $\alpha = 0.95$. Moderate changes in $\alpha$ produced comparable results. We set the number of internal knots $K = \ceil{(nM)^{1/5}}$ to match the theory in Section~\ref{sec:theory}. All parameter estimates refer to their means under the $\alpha$-variational posterior in the subsequent sections.

\subsection{Parameter Recovery}\label{subsec:recovery}

Here, we evaluate the SVI algorithm's ability to recover the true latent functions for different network sizes and densities. We measured the estimates' accuracy using three root-mean-squared errors (RMSEs): $\{(nMd)^{-1} \min_{\bO \in \mathcal{O}_d} \sum_{m=1}^M  \norm{\hat{\bU}(t_m) - \bU_0(t_m) \bO}_F^2\}^{1/2}$,  $\{(Mp)^{-1} \sum_{m=1}^M \norm{\hat{\bbeta}(t_m) - \bbeta_0(t_m)}_2^2\}^{1/2}$, $\{2(n(n-1)M)^{-1} \sum_{m=1}^M \sum_{i < j}([\hat{\bTheta}_{t_m}]_{ij} - [\bTheta_{0t_m}]_{ij})^2\}^{1/2}$.
We calculated the RMSE for the latent trajectories using the first two estimated latent space dimensions; however, the log-odds matrices was calculated using all six dimensions. 

In Figure~\ref{fig:recovery_node} and Figure~\ref{fig:recovery_time}, we report the results for synthetic networks generated according to the simulation setup described in Section~\ref{subsec:sim_setup} with expected edge densities 0.1, 0.2, and 0.3. In Figure~\ref{fig:recovery_node}, we vary the number of nodes $n \in \set{100, 200, 500, 1000}$ for a fixed number of time points $M = 100$. In Figure~\ref{fig:recovery_time}, we vary the number of time points $M \in \set{50, 100, 250, 500}$ for a fixed number of nodes $n = 250$. In all settings, we calculated the error metrics over 50 independent replicates. The SVI method performed well in all cases, with its average error decreasing as $n$, $M$, or the expected edge density increased.

\begin{figure}[tb]
\centering \includegraphics[width=\textwidth, keepaspectratio]{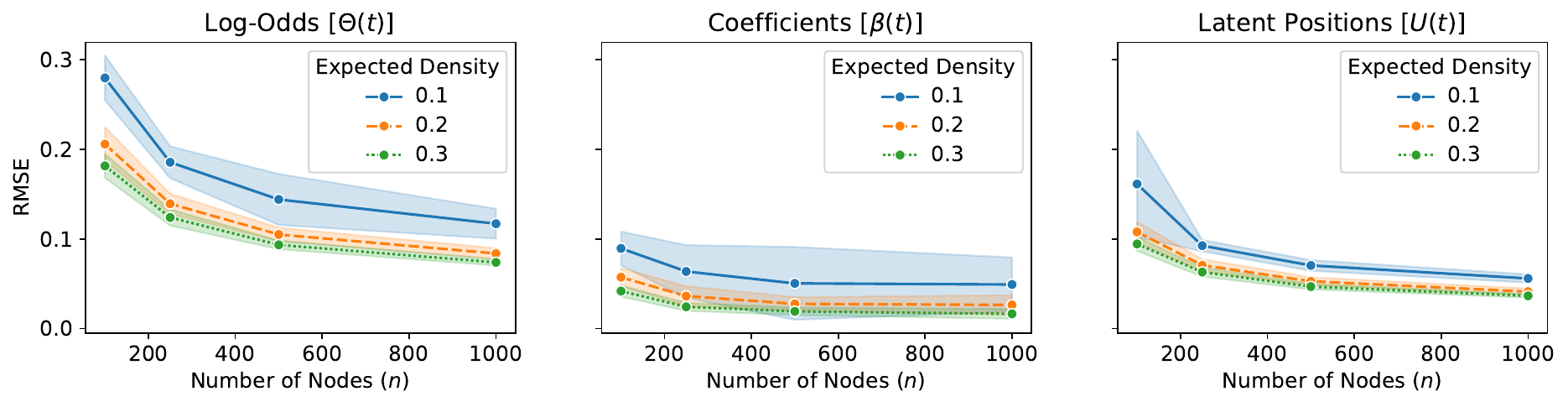}
    \caption{RMSEs for a fixed $M = 100$ and varying $n$. The curves and shaded regions indicate averages and one standard deviation over 50 independent replicates, respectively.}
    \label{fig:recovery_node}
\end{figure}

\begin{figure}[tb]
\centering \includegraphics[width=\textwidth, keepaspectratio]{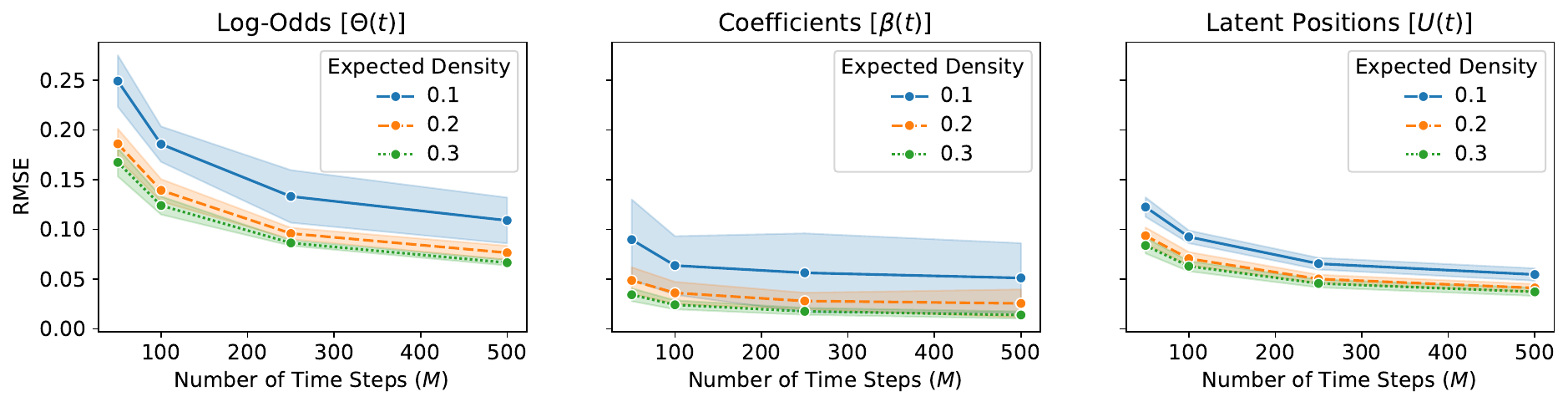}
    \caption{RMSEs for a fixed $n = 250$ and varying $M$. The curves and shaded regions indicate averages and one standard deviation over 50 independent replicates, respectively.}
\label{fig:recovery_time}
\end{figure}

\subsection{Method Comparison}\label{subsec:comp}

Next, we compared the proposed method to two competitors. The first competitor is the original GP-based dynamic LSM~\citep{durante2014}, which we label GP. We estimated the model with $d = 6$ latent space dimensions using 2,500 posterior samples drawn using MCMC after a burn-in of 2,500 samples. The second competitor is FASE~\citep{macdonald2023} estimated using gradient descent with hyperparameters chosen using their NGCV criterion. Appendix~\ref{sec:more_results}  in the supplement contains further estimation details.

For this comparison, we used the same simulation setup described in Section~\ref{subsec:sim_setup}; however, we only included an intercept because FASE cannot incorporate dyadic covariates. We compared the methods using two criteria: (1) the Pearson correlation coefficient (PCC) between the true and estimated dyad-wise probabilities and (2) the overall computation time. The PCC ranges from $-1$ to $1$, with a larger value being better. 

\begin{table}
    \centering\small
    \begin{tabular}{@{}clcc@{}} \toprule
        $(n, M)$    & Method & PCC & Computation Time (seconds) \\
    \bottomrule
        $(100, 10)$ & GP   & 0.94  &  1045 (418)\\
                    & FASE & 0.94  &  95 (15) \\
                    & P-Spline (Proposed) & 0.96  & 11 (1) \\[0.5em]
        \hline
        $(100, 20)$ & GP & 0.96   & 11422 (2987) \\
                    & FASE & 0.97   & 129 (23) \\
                    & P-Spline (Proposed) & 0.97  & 10 (2)\\[0.5em]
        \hline 
        $(200, 10)$ & GP & 0.97   & 4683 (1884) \\
                    & FASE & 0.97   & 229 (16) \\
                    & P-Spline (Proposed) & 0.98  & 37 (7)\\[0.5em]
    \bottomrule
    \end{tabular}
    \caption{Average PCCs and computation times for the competing methods over the 50 replications. The values in parentheses indicate one standard deviation. The standard deviations for the PCCs are not included because they are all less than 0.01.}
    \label{tab:comp}
\end{table}

Table~\ref{tab:comp} reports the results aggregated over 50 independent replicates for various network sizes and an expected edge density of 0.2. Table~\ref{tab:comp_appendix} in Appendix~\ref{sec:more_results} of the supplement contains the same results for networks with edge densities of 0.1 and 0.3. The following conclusions also hold for these settings. Regarding recovering the dyad-wise probabilities, the three methods performed similarly, with the proposed method achieving the best or equivalent to the best PCC in all scenarios. The proposed method is expected to outperform FASE because the data comes from model (\ref{eq:dynlsm})--(\ref{eq:dynlsm_lr}). However, the proposed model also outperformed the GP model, which more closely resembles the true data-generating process.

The benefit of the proposed SVI algorithm is highlighted by its fast computation time. Even for these small network sizes, the GP model took hours to compute, underscoring the need for a scalable Bayesian method. Furthermore, the proposed SVI algorithm is roughly an order of magnitude faster than FASE. The computation time for FASE includes performing a search over 18 parameter combinations; however, we believe this is a fair comparison since the P-spline prior for dynamic LSMs performs the equivalent selection of $d$ and $\ell$. In summary, the proposed method provides accurate estimates with an order of magnitude faster computation time than competitors while also providing approximate uncertainty quantification and adapting to critical sources of model complexity.

\section{Application to Weekly International Conflict Data}\label{sec:application}

In this section, we employ the proposed methodology on a longitudinal data set of international conflicts between nations. Specifically, we consider a dynamic network of $n = 186$ nations measured over $M = 259$ weeks between January 2018 and December 2022 for a total of 4,456,095 observed dyadic relations. An edge ($y_{ij,t_m} = 1$) indicates that a material conflict as defined by the PLOVER ontology~\citep{halterman2023} occurred between nation $i$ and nation $j$ on the $m$-th week. We constructed the network by aggregating weekly relational event data recorded in the POLECAT database~\citep{halterman2023b}. The database assigns each event one of four categories: verbal cooperation, material cooperation, verbal conflict, and material conflict. We selected the material conflict events, which, for example, include military conflicts. We limited the analysis to the 186 nations that participated in at least one material conflict during the five year period.

Our analysis aims to quantify the time-varying effects of specific dyadic covariates on the occurrence of material conflict and to identify any time-varying latent structure in the network. To do so, we applied model (\ref{eq:dynlsm})--(\ref{eq:dynlsm_lr}) with a latent space dimension of $d = 6$ and four dyadic covariates so that the log-odds of a material conflict occurring between nations $i$ and $j$ at time $t_m$ is
\begin{align*}
    \text{logit}\{\mathbb{P}(Y_{ij,t_m} = 1 \mid \bx_{ij,t_m})\} &= \beta_1(t_m) + \beta_{2}(t_m) y_{ij,t_{m-1}} + \beta_3(t_m) \texttt{ConCoopDiff}_{ij,t_{m-1}} \\
    &\quad+ \beta_4(t_m) \texttt{CommLang}_{ij} + \beta_5(t_m) \log(\texttt{Dist}_{ij}) + \bu_i(t_m)^{\top}\bu_j(t_m).
\end{align*}
In the previous expression, $\texttt{ConCoopDiff}_{ij,t_{m-1}}$ is the difference between the number of verbal conflict events and total cooperation events that occurred between nations $i$ and $j$ during the previous week, $\texttt{CommLang}_{ij}$ is a binary indicator variable for shared language, and $\texttt{Dist}_{ij}$ is the population-weighted harmonic distance between nations $i$ and $j$. In addition, we included a single endogenous covariate, $y_{ij,t_{m-1}}$, to capture edge persistence. We modeled the latent functions using the proposed P-spline prior for dynamic LSMs with first-order random walk GMRFs on all basis coefficients.

We estimated the model using the proposed SVI algorithm with the same hyperparameter values used in the simulation study. The algorithm converged in eight minutes on a laptop with an Apple M1 Pro processor. The AUC (area under the operator characteristic curve) for classifying edges is 0.94, indicating a good fit to the dynamic conflict network.

Figure~\ref{fig:coefficients} displays the means of the coefficient function's coordinates and their 95\% pointwise credible intervals according to the $\alpha$-variational posterior. The pointwise credible intervals indicate that all coefficients are significant during the observation period. The large positive magnitude of $\hat{\beta}_2(t)$ indicates a strong propensity for material conflicts to persist over time. Furthermore, a valuable observation for forecasting is that an excess of verbal conflicts over cooperation events increases the log-odds of material conflict occurring during the following week. In addition, the negative coordinate functions indicate that sharing a common language or increasing the distance between nations decreases the log-odds of a material conflict occurring. Lastly, we see a significant increase in the magnitude of the effect of $\texttt{ConCoopDiff}_{ij,t_{m-1}}$ midway through 2020. We posit that this may be due to the gradual change in the geopolitical climate after the COVID-19 pandemic.

\begin{figure}[bt]
\centering 
\includegraphics[width=\textwidth, height=0.19\textheight, keepaspectratio]{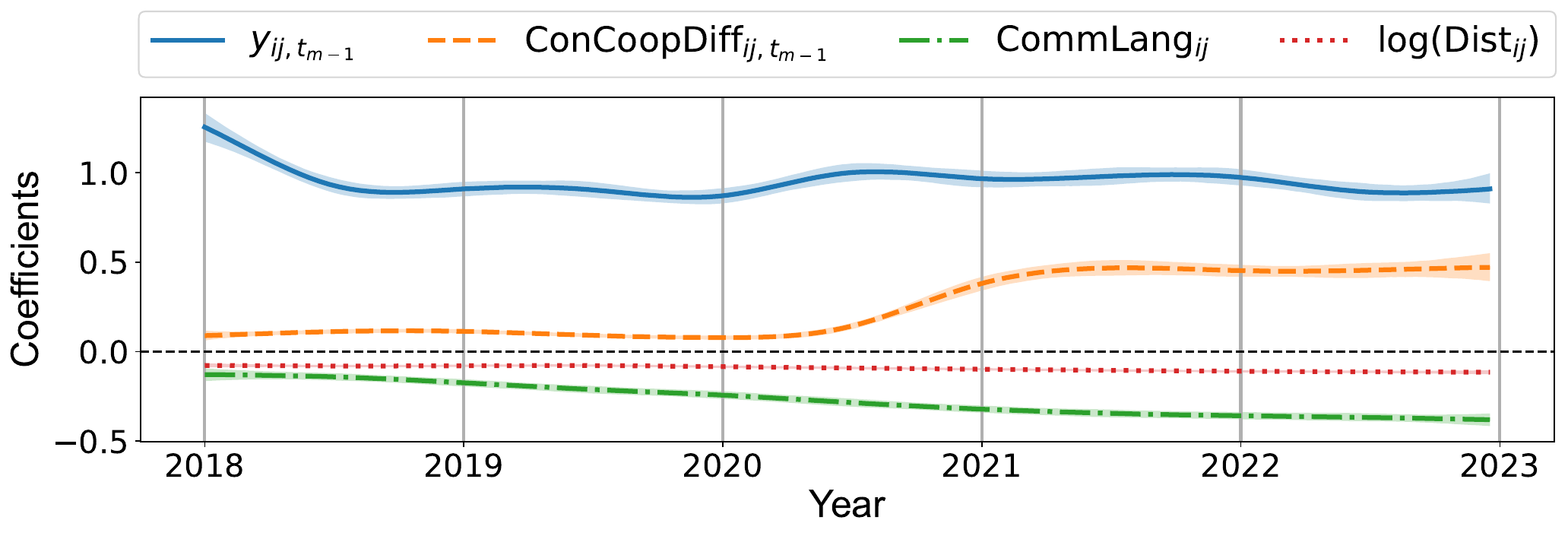}
    \caption{Dynamic covariate effects on material conflict. The curves are the $\alpha$-variational posterior means and the shaded regions indicate 95\% pointwise credible intervals.} 
\label{fig:coefficients}
\end{figure}

Next, we use the latent trajectories to detect temporal variation in network structure. A useful property of the proposed prior is that the nodewise transition variances $\set{\sigma_i^2}_{i=1}^n$ provide a way to rank the latent trajectories' temporal variation. Figure~\ref{fig:nodewise_vars} in Appendix~\ref{sec:more_results} of the supplement shows the ten nations with the largest transition variances. These ten nations participated in major material conflicts during 2018 to 2022. To demonstrate this observation, we further analyzed three of these nations: Ukraine, Venezuela, and Ethiopia.

Figure~\ref{fig:degree} shows the three nations' observed degree time series and their $\alpha$-variational posterior predictive distributions. Each time series contains a significant increase in degree around a major material conflict. Specifically, the spikes in degree occurred during the Venezuelan presidential crisis in early 2019, the Tigray War in Ethiopia beginning in late 2020, and the Russo-Ukrainian War beginning in 2022. Furthermore, these plots indicate that the model does well in capturing gradual changes in network structure with 95\% pointwise credible intervals that have good coverage despite the variational approximation. However, the model tends to over-smooth abrupt changes during the start of the conflicts. We briefly discuss a possible model extension to address this lack of fit in Section~\ref{sec:discussion}. 

\begin{figure}[tb]
\centering
\begin{subfigure}[b]{0.32\textwidth}
    \centering 
    \includegraphics[width=\textwidth, keepaspectratio]{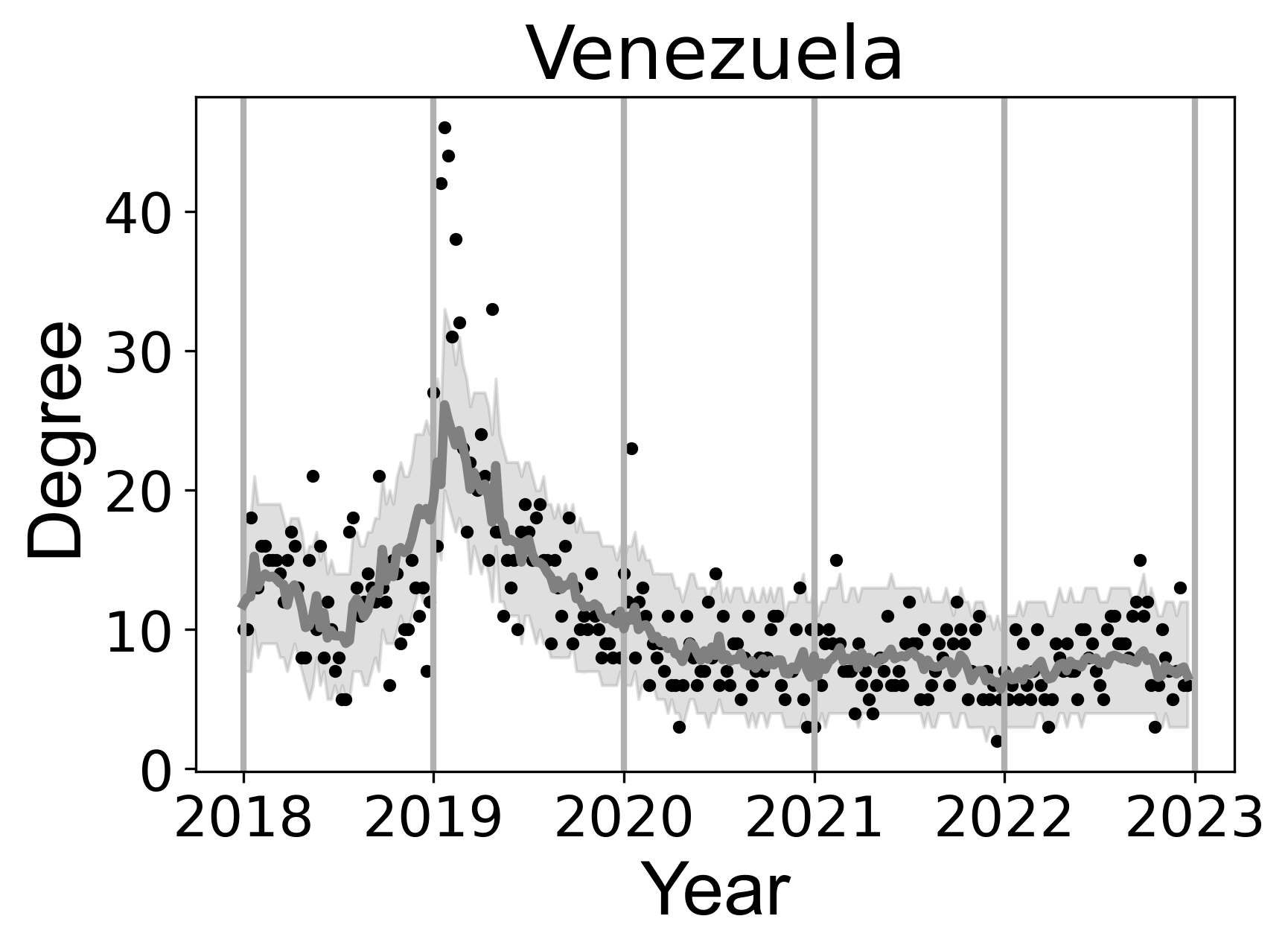}
\end{subfigure}
\begin{subfigure}[b]{0.32\textwidth}
    \centering 
    \includegraphics[width=\textwidth, keepaspectratio]{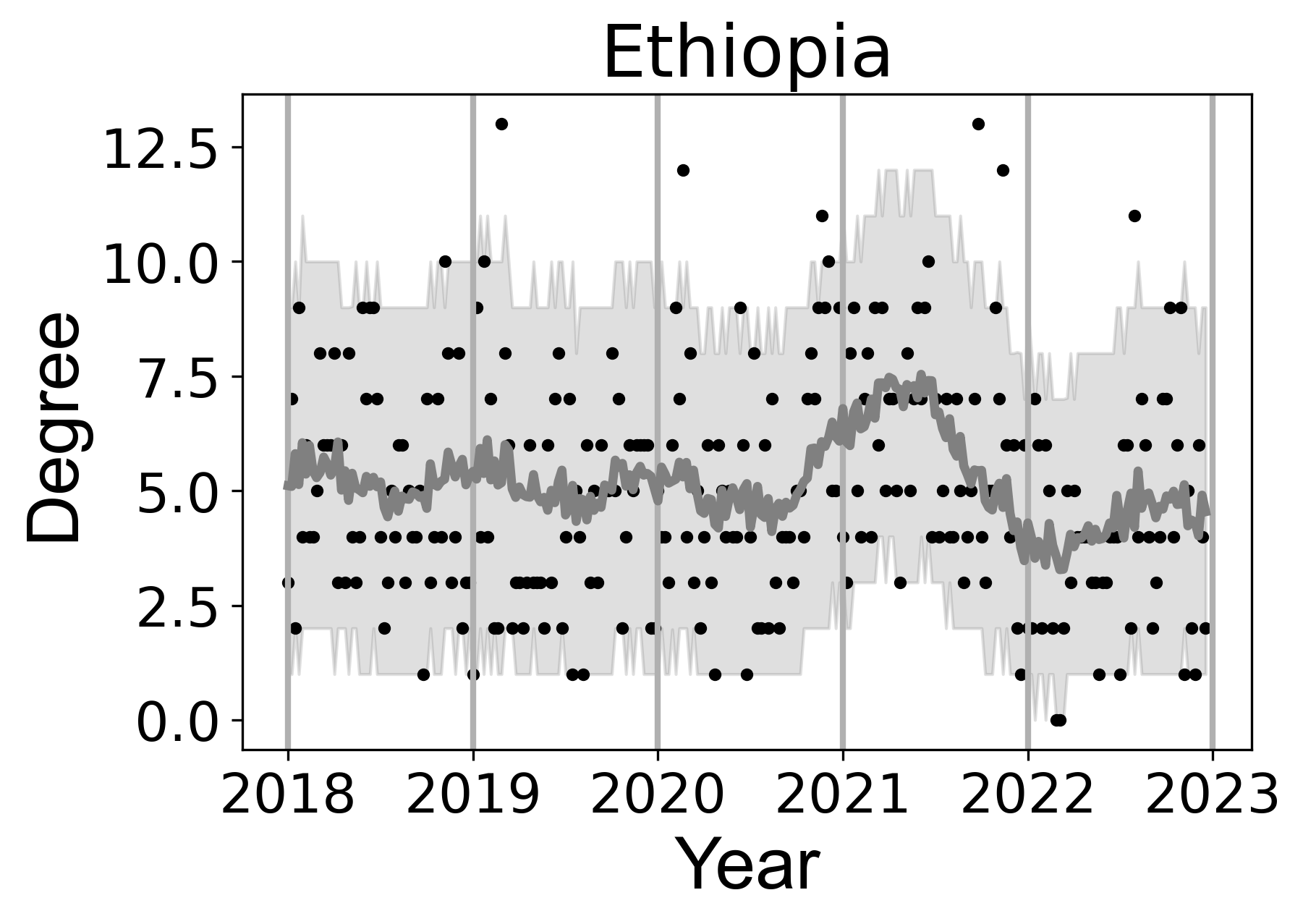}
\end{subfigure}
\begin{subfigure}[b]{0.32\textwidth}
    \centering 
    \includegraphics[width=\textwidth, keepaspectratio]{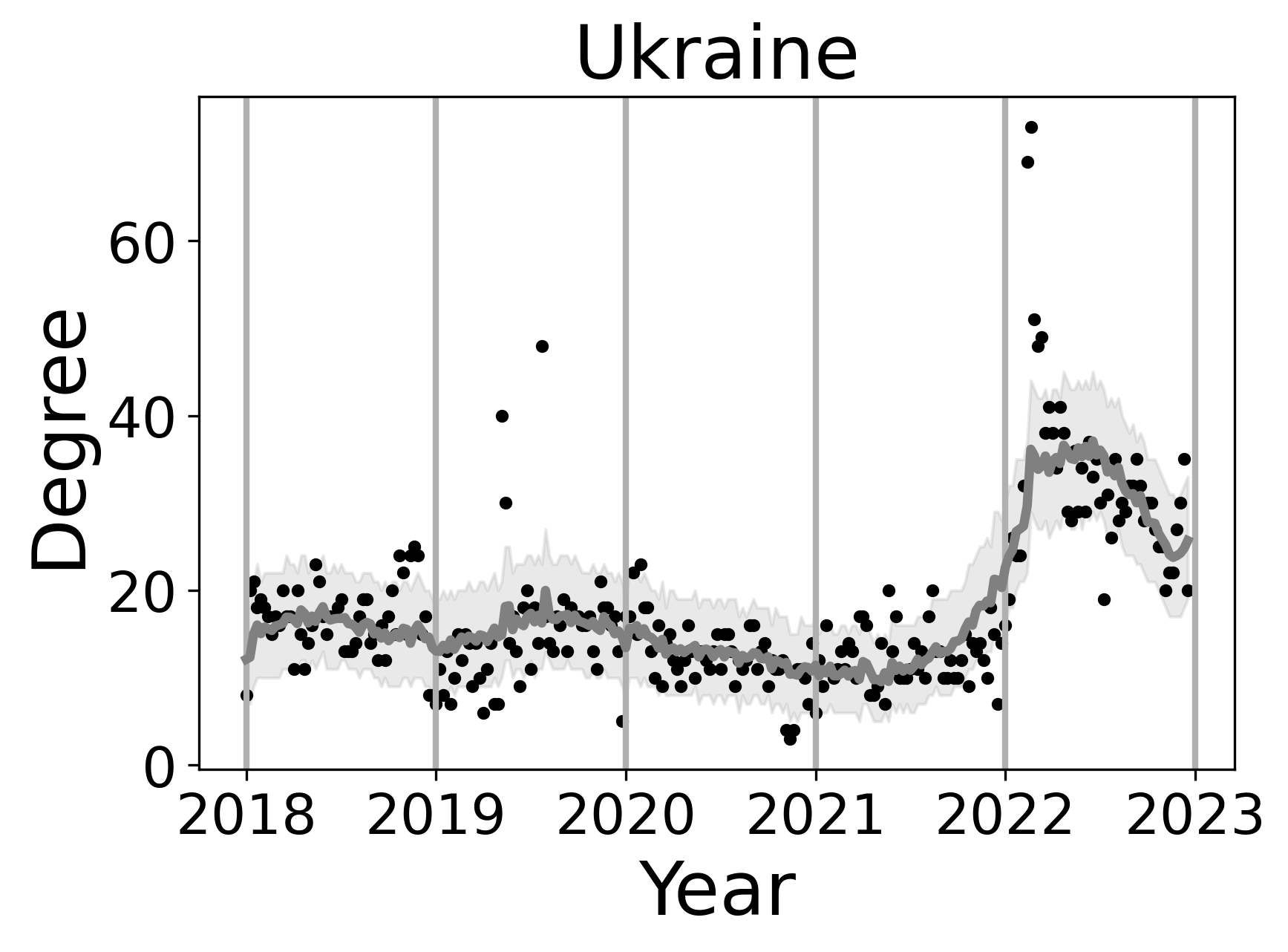}
\end{subfigure}
    \caption{The degree time series for Venezuela, Ethiopia, and Ukraine. The black dots indicate observed values. The gray curves are the $\alpha$-variational posterior means and the shaded regions indicate 95\% pointwise credible intervals.}
    \label{fig:degree}
\end{figure}

Lastly, we visualize the latent space to further understand the network's dynamics. Figure~\ref{fig:latent_space} displays the  posterior means of the first two latent space dimensions at three time points during the Venezuelan presidential crisis (left), Tigray War (center), and Russo-Ukrainian War (right). We selected two dimensions because the remaining four shrinkage parameters were concentrated near zero. See Figure~\ref{fig:shrinkage} in Appendix~\ref{sec:more_results} for details. First, we see that the latent positions are clustered by geographical region, which is especially apparent along the second latent dimension. The first latent dimension separates active from inactive nations and accounts for the USA's high degree. Furthermore, we observe that the movement of Venezuela's (VEN), Ethiopia's (ETH), and Ukraine's (UKR) latent positions are consistent with the aforementioned conflicts. In particular, each nation's latent position changes substantially during the conflict primarily affecting the country.

\begin{figure}[tb]
\centering \includegraphics[width=\textwidth, height=0.2\textheight, keepaspectratio]{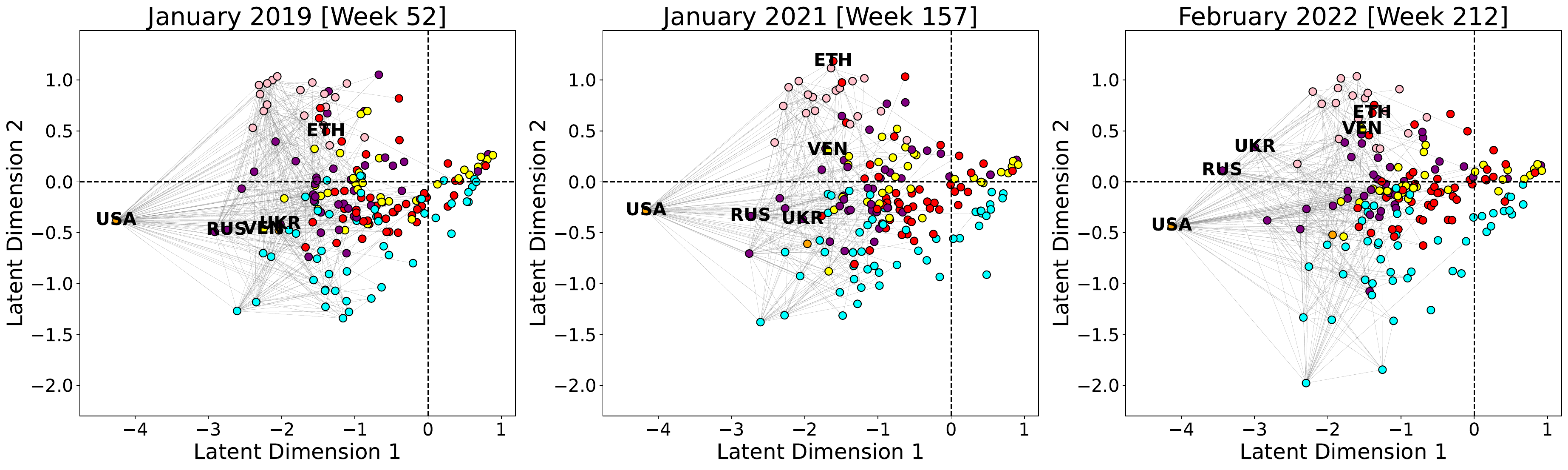}
    \caption{First ($x$-axis) and second ($y$-axis) latent space dimension evaluated at three time points. The gray lines indicate observed edges in the network. Points are colored by geographical region. Red: Africa,  Cyan: Asia-Pacific, Purple: Europe, Pink: Middle East, Orange: North America, Yellow: Latin America and the Caribbean.}
\label{fig:latent_space}
\end{figure}

\section{Discussion}\label{sec:discussion}

In this paper, we developed a Bayesian inference procedure for continuous-time dynamic LSMs with theoretical guarantees that scales to large dynamic networks. Specifically, we introduced a new prior based on Bayesian P-splines that allows the posterior to adapt to the static and dynamic complexity of the observed data and derived an SVI algorithm that is orders of magnitude faster than existing Bayesian estimation procedures. We provided theoretical and empirical support for the methodology on simulated and real data.

There are various directions for future research. Although the methodology and theory can easily be modified to accommodate networks with Gaussian edge distributions, an extension to general exponential-family distributions, such as the Poisson distribution for count-valued dynamic networks, is an area of future study. Next, as observed in the real data application, the method can be improved by using locally-adaptive spline approximations~\citep{wahba1995} to capture time-varying smoothness in network structure, which we plan to pursue in future work. In terms of theory, our results only apply to the statistical properties of the global variational solution without P\'olya-gamma augmentation. The conditions under which the current algorithm using data augmentation converges to this solution, which is contained in the augmented variational family, is an open problem.

\bigskip

\bibliographystyle{apalike}
\bibliography{references}

\newpage

\setcounter{page}{1}
\setcounter{section}{0}
\setcounter{equation}{0}
\setcounter{table}{0}
\setcounter{figure}{0}
\setcounter{assump}{3}
\pagenumbering{arabic}
\renewcommand{\thesection}{\Alph{section}}
\renewcommand{\thesubsection}{\Alph{section}.\arabic{subsection}}
\renewcommand\theequation{S.\arabic{equation}}
\renewcommand\thetable{S.\arabic{table}}
\renewcommand\thefigure{S.\arabic{figure}}
\renewcommand{\themyalgorithm}{S.\arabic{myalgorithm}}
\renewcommand{\thetheorem}{S.\arabic{theorem}}
\renewcommand{\theproposition}{S.\arabic{proposition}}
\renewcommand{\thecorollary}{S.\arabic{corollary}}
\renewcommand{\thelemma}{S.\arabic{lemma}}

\begin{center}
{\Large\bf Supplementary Material for \\
    ``Fast Variational Inference of Latent Space Models for Dynamic Networks Using Bayesian P-Splines''}\\[2em]
{\large Joshua Daniel Loyal}
\end{center}

\section{Notation}

Here, we review the notation used throughout the main article and introduce some new notation used throughout the supplement. For a continuous function $f \, : \, [0,1] \rightarrow \Reals{}$, $\norm{f}_{L_{\infty}[0,1]} = \text{ess\,sup}_{x \in [0,1]} \abs{f(x)}$ denotes the supremum norm. The notation $[\cdot]_{ij}$ denotes the $(i,j)$-th entry of a matrix or the $(i,j)$-th tube fiber of a three-way tensor. For a matrix $\bC$, we denote its minimum singular value as $\sigma_{min}(\bC)$, its Frobenius norm as $\norm{\bC}_{F}$, and its operator norm as $\norm{\bC}_{op}$. We use $\mathcal{O}_d$ to denote the group of $d$-dimensional orthogonal matrices. We use $\diag(v_1, \dots, v_k)$ to denote a $k \times k$ diagonal matrix with $v_1, \dots, v_k$ on its diagonal. We let $\mathbf{0}_d$ denote the $d$-dimensional vector of zeros. We use $\indsim$ and $\iidsim$ to denote independently distributed and independently and identically distributed, respectively. For two densities $p$ and $q$, we use $D_{KL}(p, q)$ to denote the Kullback-Leibler (KL) divergence between $p$ and $q$. For sequences $a_n$ and $b_n$, we write $a_n \lesssim b_n$ (or $b_n \gtrsim a_n$) to imply that $a_n \leq c b_n$ for some constant $c$ independent of $n$. The notation $a_n = O(b_n)$ implies $a_n \lesssim b_n$ while $a_n \asymp b_n$ implies $a_n \lesssim b_n$ and $a_n \gtrsim b_n$. We use $a_n \ll b_n$ to mean $\lim_{n \rightarrow \infty} a_n / b_n = 0$.

\section{Connections to Gaussian Process Priors}\label{subsec:gp}

As observed in Section~\ref{subsec:pspline} of the main text, the latent trajectories under the proposed P-spline prior for dynamic LSMs implicitly have mean-zero GP priors with non-stationary covariance functions. Formally, the covariance function for the $h$-th coordinate function of the $i$-th latent trajectory is
\[
    \text{Cov}\{u_{ih}(s), u_{ih}(t)\} = \bb(s)^{\top} \text{Var}(\bw_{ih}) \bb(t) = \gamma_{h}^{-1} \bb(s)^{\top}\bOmega_{i}^{-1}\bb(t).
\]
It is often necessary to incorporate non-stationarity to realistically describe dynamic network data~\citepSup{durante2016}. As such, a benefit of the proposed prior is that it can model non-stationary dynamic networks in a computationally efficient manner without using a large number of parameters.

Next, we use this connection with GPs to provide more insight into the proposed Bayesian P-spline prior for dynamic LSMs. Specifically, we examine the effect that the latent trajectories' variance parameters have on the induced prior moments of the log-odds matrices. As such, all moments in this section are condition on the coefficient function $\bbeta(t)$, so that they describe the residual dependencies induced by the latent trajectories.  We have the following results on the induced moments of the log-odds matrices, whose proof is provided at the end of this section. For brevity, we let $v_i(s,t) = \bb(s)^{\top}\bOmega_i^{-1}\bb(t)$ and $v_i(t) = v_i(t, t)$. All moments are conditioned on the variance parameters, which we suppressed for clarity. 
\begin{proposition}\label{prop:induced_dependence}
    Under the P-spline prior for dynamic LSMs defined in Equations (\ref{eq:pspline_u})--(\ref{eq:precision}), the induced prior on the log-odds matrices' elements have the following first and third central moments conditioned on the variance parameters and the coefficient function:
    \begin{align*}
        \mathbb{E}([\bTheta_t]_{ij} \mid \bbeta(t)^{\top}\bx_{ij,t}) &= \bbeta(t)^{\top}\bx_{ij,t}, \\
        \mathbb{E}(\xi_{ij,t}\xi_{jk,t}\xi_{ki,t} \mid \bbeta(t)^{\top}\bx_{ij,t}, \bbeta(t)^{\top}\bx_{jk,t}, \bbeta(t)^{\top}\bx_{ki,t})  &=  v_i(t) v_j(t) v_k(t) \sum_{h=1}^d \gamma_h^{-3},
    \end{align*}
    where $\xi_{ij,t} = [\bTheta_t]_{ij} - \bbeta(t)^{\top}\bx_{ij,t}$, $t \in \set{t_m}_{m=1}^M$, and $1 \leq i \neq  j  \neq k \leq n$, and the covariance between any two time points $s,t \in \set{t_m}_{m=1}^M$ is given by
    \begin{align*}
        \Cov{[\bTheta_s]_{ij}, [\bTheta_t]_{ij} \mid \bbeta(t)^{\top}\bx_{ij,s}, \bbeta(t)^{\top}\bx_{ij,t}} &=  v_i(s,t) v_j(s,t) \sum_{h=1}^d \gamma_h^{-2}.
    \end{align*}
\end{proposition}

According to Proposition~\ref{prop:induced_dependence}, a priori the mean log-odds of forming an edge is linear in the covariates after marginalizing out the latent trajectories. Furthermore, the dependence induced by the latent trajectories on elements of the log-odds matrices after conditioning on the covariate effects is controlled by the shrinkage parameters $\set{\gamma_{h}^{-1}}_{h=1}^d$ and the $v_i(s,t)$ functions. The second central moment describes autocorrelation for a specific dyad and the third central moment describes the dependence between transitive triplets in the network. The temporal variation in these higher-order moments is determined by the node-specific $v_i(s,t)$ functions while their overall magnitude increases with the sum of the shrinkage parameters. When all $\gamma_h^{-1}$ are near zero, the higher-order moments are roughly zero indicating that the covariates explain most of the dependence in the network. However, as each $\gamma_h^{-1}$ grows, the latent trajectories explain more of the residual network dependence. 

To better describe the central-moments in the previous proposition, we have the following expression for the $v_i(s,t)$ functions when $\bb(t)$ is a B-spline basis. The proof is provided at the end of this section.
\begin{proposition}\label{prop:rw_v} 
    Let $u_{ih}(t) = \bw_{ih}^{\top}\bb(t)$ where $\bb(t)$ is an $\ell$-dimensional B-spline basis, then under the first-order Gaussian random walk prior for $\bw_{ih}$ defined in Equation~(\ref{eq:pspline_u}), we have 
\[
    v_i(s,t) = \sigma_i^2 \left\{\sum_{g = 1}^{\ell} \sum_{g' = 1}^{\ell} b_{g}(s) b_{g'}(t) \min(g, g')\right\} + \tau^2 = \sigma_i^2 v_{\bb}(s,t) + \tau^2,
\]
    so that $\Cov{u_{ih}(t), u_{ih}(s)} = \gamma_h^{-1} \{\sigma_i^2 v_{\bb}(s,t) + \tau^2\}$.
\end{proposition}
According to Proposition~\ref{prop:rw_v}, the basis-specific function $v_{\bb}(s,t)$ determines the functional form of $v_i(s,t)$. In addition, $v_i(s,t)$'s magnitude increases with $\sigma_i^2$ and $\tau^2$. Crucially, when $\sigma_i^2$ is zero, the $i$-th node's covariance function is not time-varying. Moreover, it is known in the Bayesian P-spline literature that the magnitude of $v_{\bb}(s,t)$ increases with the basis dimension $\ell$~\citepSup{ventrucci2016}. This behavior explains why, to obtain consistent estimates of the unknown latent functions, we must apply sufficient shrinkage on $\sigma_i^2$ to counterbalance the increase in variation caused by $\ell$ growing with network size.

    \begin{proof}[Proof of Proposition~\ref{prop:induced_dependence}]
        We start by stating some properties of the latent trajectories under the P-spline prior for dynamic LSMs. Under the proposed prior, we have that $\E{u_{ih}(t)} = \E{\bw_{ih}}^{\top}\bb(t) = 0$, so that $\E{\bu_i(t)} = \mathbf{0}_d$ for all $1 \leq i \leq n$ and $t \in [0,1]$. Next, let $\bGamma^{-1} = \diag(\gamma_1^{-1}, \dots, \gamma_d^{-1})$. The covariance matrix and autocovariance matrix of $\bu_i(t)$ for all $1 \leq i \leq n$ and $s,t \in [0,1]$ are $\textrm{Var}\set{\bu_i(t)} = v_i(t) \bGamma^{-1}$ and $\mathbb{E}[\bu_i(t) \bu_i(s)^{\top}] = v_i(t, s) \bGamma^{-1}$, respectively.
        
        Now, we prove the various properties asserted in the proposition. Take a fixed time $t \in \set{t_m}_{m=1}^M$ and dyad $1 \leq i,j \leq n$, then
        \begin{align*}
            \mathbb{E}([\bTheta_{t}]_{ij} \mid \bbeta(t)^{\top}\bx_{ij,t}) &= \bbeta(t)^{\top}\bx_{ij,t} + \E{\bu_i^{\top} \bu_j} \\
            &= \bbeta(t)^{\top}\bx_{ij,t} + \E{\bu_i}^{\top} \E{\bu_j} = \bbeta(t)^{\top}\bx_{ij,t}.
        \end{align*}
        Next, we consider a triplet of nodes $1 \leq i \neq j \neq k \leq n$. We have 
        \begin{align*}
            \mathbb{E}[\xi_{ij,t} \xi_{jk,t} \xi_{k i,t} \mid \bbeta(t)^{\top}\bx_{ij,t}, \bbeta(t)^{\top}\bx_{jk, t} &\bbeta(t)^{\top}\bx_{k i, t}] =
            \E{\bu_i(t)^{\top}\bu_j(t)\bu_j(t)^{\top}\bu_{k}(t)\bu_{k}(t)^{\top}\bu_i(t)} \\
            &= \E{\tr\{\bu_i(t)^{\top}\bu_j(t)\bu_j(t)^{\top}\bu_{k}(t)\bu_{k}(t)^{\top}\bu_i(t)\}} \\
            &= \tr\{\E{\bu_i(t)\bu_i(t)^{\top}} \E{\bu_j(t)\bu_j(t)^{\top}}\E{\bu_{k}(t)\bu_{k}(t)^{\top}}\} \\
            &=\tr\{v_i(t) v_j(t) v_{k}(t)  \bGamma^{-3}\} \\
            &= v_i(t) v_j(t) v_{k}(t) \sum_{h=1}^d \gamma_h^{-3}.
        \end{align*}

        Finally, we calculate the autocovariance of the log-odds over time. For two time points $s,t \in \set{t_m}_{m=1}^M$, we have
        \begin{align*}
            \text{Cov}\{[\bTheta_{t}]_{ij}, [\bTheta_{s}]_{ij} \mid \bbeta(t)^{\top}\bx_{ij,t}, \bbeta(t)^{\top}\bx_{ij,s}\} &= \E{\bu_i(t)^{\top}\bu_j(t) \bu_j(s)^{\top} \bu_i(s)} \\
            &= \E{\tr\{\bu_i(s)\bu_i(t)^{\top} \bu_j(t) \bu_j(s)^{\top}\}} \\
            &= \tr\{\E{\bu_i(s)\bu_i(t)^{\top}}\E{\bu_j(t) \bu_j(s)^{\top}}\} \\
            &= \tr\{v_i(s,t) v_j(t,s) \bGamma^{-2}\} \\
            &= v_i(s,t) v_j(s,t) \sum_{h=1}^d \gamma_h^{-2},
        \end{align*}
        where in the last line we used the fact that $v_i(s,t) = v_i(t,s)$ for all $t,s \in [0,1]$. 
    \end{proof}

    \begin{proof}[Proof of Proposition~\ref{prop:rw_v}]
        Under the first-order Gaussian random walk prior for $\bw_{ih}$, we have that $[\bOmega_i^{-1}]_{g,g'} = \sigma_i^2 \min(g, g') + \tau^2$ for $1 \leq g,g' \leq \ell$. It follows that
        \begin{align*}
            v_i(s,t) &= \bb_u(s)^{\top}\bOmega^{-1}_i \bb_u(t) \\
            &= \sum_{g=1}^{\ell} \sum_{g'=1}^{\ell} b_{g}(s) b_{g'}(t)\left[ \sigma_i^2 \min(g, g') + \tau^2\right] \\
            &= \sigma^2 \left\{\sum_{g=1}^{\ell} \sum_{g'=1}^{\ell} b_{g}(s) b_{g'}(t) \min(g, g')\right\} + \tau^2 \sum_{g=1}^{\ell} b_{g}(s) \left\{\sum_{g'=1}^{\ell}  b_{g'}(t)\right\} \\
            &= \sigma^2 \left\{\sum_{g=1}^{\ell} \sum_{g'=1}^{\ell} b_{g}(s) b_{g'}(t) \min(g, g')\right\} + \tau^2 \norm{\bb(s)}_1 \norm{\bb(t)}_1 \\
            &= \sigma^2 \left\{\sum_{g=1}^{\ell} \sum_{g'=1}^{\ell} b_{g}(s) b_{g'}(t) \min(g, g')\right\} + \tau^2,
        \end{align*}
        where we used the fact that $\norm{\bb(t)}_1 = 1$ for any $t \in [0,1]$ because the elements of $\bb(t)$ form a B-spline basis~\citepSup{deboor1978}.
    \end{proof}

\section{Additional Details on the SVI Algorithm}\label{sec:init}
    
This section outlines the remaining technical details of the proposed SVI algorithm. To ease notation, we define the following expectations taken with respect to the $\alpha$-variational posterior $q(\mathcal{W}, \brho)$ used throughout the remainder of the supplementary material:
\begin{align*}
    &\mathbb{E}[\beta_k(t)] = \mu_{\beta_k}(t) = \bmu_{\bw_k}^{\top} \bb(t), \quad \mathbb{E}[\{\beta_k(t)\}^2] = \{\mu_{\beta_k}(t)\}^2 + \bb(t)^{\top} \bSigma_{\bw_k} \bb(t), \\
    &\mathbb{E}[u_{ih}(t)] = \mu_{ih}(t) = \bmu_{\bw_{ih}}^{\top}\bb(t), \quad \mathbb{E}[\{u_{ih}(t)\}^2] = \{\mu_{ih}(t)\}^2 + \bb(t)^{\top} \bSigma_{\bw_{ih}} \bb(t), \\
    &\mathbb{E}[\bu_i(t)] = \bmu_i(t) = (\mu_{i1}(t), \dots, \mu_{id}(t))^{\top}, \quad \mathbb{E}[\bbeta(t)] = \bmu_{\bbeta}(t) = (\mu_{\beta_1}(t), \dots, \mu_{\beta_p}(t))^{\top}, \\
    &\text{Var}(\bbeta(t)) = \bSigma_{\bbeta}(t) = \diag\{\bb(t)^{\top} \bSigma_{\bw_{\beta_1}} \bb(t), \dots, \bb(t)^{\top}\bSigma_{\bw_{\beta_p}}\bb(t)\}, \\
    &\text{Var}(\bu_i(t)) = \bSigma_{i}(t) = \diag\{\bb(t)^{\top} \bSigma_{\bw_{i1}} \bb(t), \dots, \bb(t)^{\top}\bSigma_{\bw_{id}}\bb(t)\}, \\
    &\mathbb{E}[(\bu_i^{\top} \bu_j)^2] = \text{tr}\{\bSigma_{i}(t)\bSigma_j(t)\} + \bmu_j(t)^{\top}\bSigma_i(t) \bmu_j(t) + \bmu_i^{\top} \bSigma_j(t) \bmu_i + \{\bmu_i(t)^{\top} \bmu_j(t)\}^2, \\
    &\mathbb{E}[\gamma_h] = \mu_{\gamma_h} = \prod_{s=1}^h \mathbb{E}[\nu_s], \quad \mathbb{E}[\omega_{ij,t_m}] = \mu_{\omega_{ij,t_m}}.
\end{align*}

\subsection{Algorithms for Updating \texorpdfstring{$q(\brho)$}{q(rho)} and \texorpdfstring{$q(\bomega)$}{q(omega)}}\label{sec:rem_svi}
    This section contains the updates for the variational factors of the local P\'olya-gamma latent variables in Algorithm~\ref{alg:svi_local} and the variance parameters in Algorithm~\ref{alg:svi_var}. In Algorithm~\ref{alg:svi_var}, we use $\text{GIG}(a,b,p)$ to denote a generalized inverse Gaussian (GIG) distribution with density  
\[
    \text{GIG}(x \mid a, b, p) = \frac{(a/b)^{p/2}}{2 K_p(\sqrt{ab})} x^{p-1} e^{-(ax + b/x)/2},
\]
where $K_p(x)$ is the modified Bessel function of the second kind. We derive these updates in Appendix~\ref{sec:der_svi}.
\begin{myalgorithm}[hb]
    \begin{framed}
    Given the previous parameters at step $s$, update the local parameters for dyad $(i, j)$ at time $t_m$ as follows:
    \begin{enumerate}
        \item Update $q(\omega_{ij,t_m}) = \PG(\alpha, c_{ij,t_m})$:
            \begin{align*}
                c_{ij,t_m}^2 &= [\bmu_{\bbeta}^{(s)}(t_m)^{\top} \bx_{ij,t_m} + \bmu_i^{(s)}(t_m)^{\top}\bmu_j^{(s)}(t_m)]^2 + \\
                &\qquad \text{tr}\set{\bx_{ij,t_m} \bx_{ij,t_m}^{\top} \bSigma^{(s)}_{\bbeta}(t_m)} + \text{tr}\set{\bSigma_{i}^{(s)}(t_m)\bSigma_j^{(s)}(t_m)} + \\
                &\qquad \bmu_j^{(s)}(t_m)^{\top}\bSigma_i^{(s)}(t_m) \bmu_j^{(s)}(t_m) + \bmu_i^{(s)}(t_m)^{\top} \bSigma_j^{(s)}(t_m) \bmu_i^{(s)}(t_m).
            \end{align*}
        \item Update the mean of the local P\'{o}lya-gamma latent variable under $q(\omega_{ij,t_m})$:
            \[
                \mu_{\omega_{ij,t_m}} = \frac{\alpha}{2 c_{ij,t_m}} \left(\frac{e^{c_{ij,t_m}} - 1}{1 + e^{c_{ij,t_m}}}\right). 
            \]
                
    \end{enumerate}
    \end{framed}
    \caption{SVI update for the P\'{o}lya-gamma latent variable of dyad $(i,j)$ at time $t_m$.}
    \label{alg:svi_local}
\end{myalgorithm}

\begin{myalgorithm}
    \begin{framed}
    Given the previous parameters at step $s$ and step size $\rho_s$, update the current parameters as follows:
    \begin{enumerate}
        \item Update $q(\sigma_i^2) = \GIG(d_{\sigma}, \bar{b}_i, (c_{\sigma} - d(\ell - 1))/2)$ for $i \in \set{1, \dots, n}$, where:
            \begin{align*}
                \bar{b}_i^{(s+1)} &= (1 - \rho_s) \bar{b}_i^{(s)} + \rho_s \sum_{h=1}^d \mu_{\gamma_h}^{(s)} \left[\bmu_{\bw_{ih}}^{(s) \, \top} \bD^{(1) \, \top}_{\ell} \bD_{\ell}^{(1)} \bmu_{\bw_{ih}}^{(s)} + \tr(\bD_{\ell}^{(1) \top}\bD_{\ell}^{(1)} \bSigma_{\bw_{ih}}^{(s)})\right].
            \end{align*}

        \item Update $q(\sigma_{\beta_k}^2) = \GIG(d_{\sigma}, \bar{b}_{\beta_k}, (c_{\sigma} - (\ell - r_k)/2))$ for $k \in \set{1, \dots, p}$, where:
            \begin{align*}
                \bar{b}_{\beta_k}^{(s+1)} = (1 - \rho_s) \bar{b}_{\beta_k}^{(s)} + \rho_s \left[\bmu_{\bw_k}^{(s) \, \top} \bD_{\ell}^{(r_k)  \top} \bD_{\ell}^{(r_k)} \bmu_{\bw_k}^{(s)} + \tr( \bD_{\ell}^{(r_k) \top} \bD_{\ell}^{(r_k)} \bSigma_{\bw_k}^{(s)})\right]
            \end{align*}
        \item Update $q(\nu_h) = \text{Gamma}(\bar{c}_h, \bar{d}_h)$ for $h \in \set{1, \dots, d}$:
            \begin{align*}
                \bar{c}_h^{(s+1)} &= \begin{cases}
                    a_2 + (d -h + 1)n \ell/2, & h > 1, \\
                    a_1 + d n \ell/2, & \text{otherwise}, 
                \end{cases} \\
                \bar{d}_h^{(s+1)} &= (1 - \rho_s) \bar{d}_h^{(s)} + \rho_s \bigg\{1 + \frac{1}{2}\sum_{t=h}^d \mu_{\gamma_{t, h}}^{(s)} \sum_{i=1}^n \mathbb{E}_{q(\bw_{it}, \sigma_i^2)}[\bw_{it}^{\top}\bOmega_i \bw_{it}]\bigg\},
            \end{align*}
            where
            \[
                \mu_{\gamma_{t,h}}^{(s)} = \begin{cases}
                    \prod_{g=1, g \neq h}^{t} \frac{\bar{c}_g^{(s)}}{\bar{d}_g^{(s)}}, & h > 1 \\
                    1, & \text{otherwise},
                \end{cases}
            \] 
            \begin{align*}
                \mathbb{E}_{q(\bw_{it}, \sigma_i^2)}[\bw_{it}^{\top} \bOmega_i \bw_{it}] &= \mathbb{E}_{q(\sigma_i^2)}\left[\frac{1}{\sigma_i^2}\right] \left\{\bmu_{\bw_{it}}^{(s) \, \top}\bD_{\ell}^{(1) \top} \bD_{\ell}^{(1)} \bmu_{\bw_{it}}^{(s)} + \tr(\bD_{\ell}^{(1) \top}\bD_{\ell}^{(1)} \bSigma_{\bw_{it}}^{(s)})\right\} \\
                &\qquad\qquad + \frac{1}{\tau^2}\{\mu_{\bw_{it}, 1}^{(s) \, 2} + [\bSigma_{\bw_{it}}^{(s)}]_{11}\}, 
            \end{align*}
            \[
                \mathbb{E}_{q(\sigma_i^2)}\left[\frac{1}{\sigma_i^2}\right] = \frac{\sqrt{d_{\sigma}} K_{v + 1}\left(\sqrt{d_{\sigma} \bar{b}_i^{(s)}}\right)}{\sqrt{\bar{b}_i^{(s)}} K_{v}\left(\sqrt{d_{\sigma}\bar{b}_i^{(s)}}\right)} - \frac{2v}{\bar{b}_i^{(s)}}, \qquad v = \frac{1}{2} (c_{\sigma} - d(\ell - 1)),
            \]
            and $K_v(\cdot)$ is the modified Bessel function of the second kind.
        
    \end{enumerate}
    \end{framed}
    \caption{SVI updates for the variance parameters.}
    \label{alg:svi_var}
\end{myalgorithm}
    
\clearpage

\subsection{Parameter Initialization}

     Since the ELBO is non-convex with multiple local minima, appropriate initialization for the parameters can significantly improve convergence. Algorithm~\ref{alg:init} outlines our proposed initialization procedure. First, the initialization method estimates the dyad-wise probability matrix using universal singular value thresholding (USVT)~\citepSup{chatterjee2015}. Then, it computes estimates for $\set{\bbeta(t_m)}_{m=1}^M$ and $\set{\bU(t_m)}_{m=1}^M$ based on the estimated log-odds matrices obtained by inverting the logit transform. The method estimates the coefficient functions by minimizing a least squares objective and the latent trajectories by performing a $d$-dimensional adjacency spectral embedding ($\text{ASE}_d$)~\citepSup{athreya2018} on the resulting residual matrix. The algorithm projects the estimated coordinate functions onto the span of the B-spline basis to obtain estimates for the basis coefficients. \citetSup{ma2020} and \citetSup{macdonald2023} proposed a similar procedure to initialize a static LSM and FASE, respectively.  The algorithm contains a sequential Procrustes alignment step that smooths the initial estimates of the latent trajectories.

    \begin{myalgorithm}
        \begin{framed}
        Given the latent space dimension $d$ and B-spline basis $\bb(t)$, perform the following steps: 
        \begin{enumerate}
            \item For $m = 1, \dots, M$:
            \begin{enumerate}
            \item {\it (USVT)}. Define the threshold $\tau = \sqrt{2.01 n \hat{p}}$, where $\hat{p} = (2/n(n+1)) \sum_{i \leq j } y_{ij,t_m}$. Let $\tilde{\bP}_{t_m} = \sum_{s_i \geq \tau} s_i \bu_i\bv_i^{\top}$, where $\sum_{i=1}^n s_i \bu_i\bv_i^{\top}$ is the singular value decomposition of $\bY_{t_m}$. Project $\tilde{\bP}_{t_m}$ elementwise to the interval $[0.01, 0.5]$ to obtain $\hat{\bP}_{t_m}$. Set $\hat{\bTheta}_{t_m} = \logit\{(\hat{\bP}_{t_m} + \hat{\bP}_{t_m}^{\top})/2\}$.
    
            \item {\it (Coefficient functions)}. Set $\hat{\bbeta}(t_m) = \argmin_{\bbeta(t_m)} \norm{\hat{\bTheta}_{t_m} - \mathcal{X}_{t_m} \bar{\times}_3 \bbeta(t_m)}_F^2$ and define the residual $\bE_{t_m} = \hat{\bTheta}_{t_m} - \mathcal{X}_{t_m} \bar{\times}_3 \hat{\bbeta}(t_m)$.
            \item {\it (Latent trajectories)}. Set $\hat{\bU}(t_m) = \text{ASE}_d(\bE_{t_m}) = \bV_d \bLambda_d^{1/2}$, where $\bV_d \in \Reals{n \times d}$ and $\bLambda_d \in \Reals{d \times d}$ correspond to the $d$ largest eigenvectors and eigenvalues of $\bE_{t_m}$.
            \end{enumerate}
        \item {\it (Align $\hat{\bU}(t_1), \dots \hat{\bU}(t_m)$)}. Moving sequentially forward in time starting at $m = 2$, project $\hat{\bU}(t_m)$ to the locations that are closest to its previous location $\hat{\bU}(t_{m-1})$ through a Procrustes rotation~\citepSup{hoff2002}.
        \item {\it (Project estimates onto the B-spline basis)}. Let $\bB = (\bb(t_1), \dots, \bb(t_M))^{\top} \in \Reals{M \times \ell}$.
            \begin{enumerate}
                \item For $i = 1, \dots, n$ and $h = 1, \dots, d$, set 
                \[
                    \bmu_{\bw_{ih}} = (\bB^{\top}\bB)^{-1} \bB^{\top} \begin{pmatrix}
                        \hat{u}_{ih}(t_1) \\
                            \vdots \\
                        \hat{u}_{ih}(t_M)
                        \end{pmatrix},
                \]
                where $\hat{u}_{ih}(t_m)$ is the $(i,h)$-th element of $\hat{\bU}(t_m)$.
                \item For $k = 1,\dots, p$, set 
                    \[
                        \bmu_{\bw_k} = (\bB^{\top}\bB)^{-1} \bB^{\top} \begin{pmatrix}
                            \hat{\beta}_k(t_1) \\
                            \vdots \\
                            \hat{\beta}_k(t_M)
                        \end{pmatrix}.
                    \]
            \end{enumerate}
        \end{enumerate}
        \end{framed}
        \caption{Initialization method for the basis coefficients.}
        \label{alg:init}
    \end{myalgorithm}
    
    It remains to initialize the other parameters of the $\alpha$-variational posterior. These parameters include the precision matrices of the basis coefficients' variational factors, the $\set{\bar{b}_i}_{i=1}^n$ and $\set{\bar{b}_{\beta_k}}_{k=1}^p$ parameters of the GIG factors associated with the transition variances, and the $\set{\bar{d}_h}_{h=1}^d$ parameters of the gamma factors associated with the multiplicative gamma process parameters. We set the precision matrices for the basis coefficients equal to the identity matrix, and the remaining $\set{\bar{b}_i}_{i=1}^n$, $\set{\bar{b}_{\beta_k}}_{k=1}^p$ and $\set{\bar{d}_h}_{h=1}^d$ parameters to 100.
    
    \subsection{Convergence Criteria and Post-Processing}
     
     To determine convergence of the SVI algorithm, we monitored the log-likelihood of the subsampled dyads. To account for the subsampling noise, we set the stopping criterion to whether the change in the median log-likelihoods of the subsampled dyads calculated over two consecutive windows of 20 iterations was less than $10^{-3}$ or the number of iterations exceeded 250. 

    Upon convergence, the estimated latent positions are only identified up to an orthogonal transformation that can differ between time points, which can hinder visualization. To improve visualization, we performed sequential Procrustes rotations~\citepSup{hoff2002} on these estimates where the estimated latent positions starting at time $t_2$ are projected to the locations closest to their previous location. Such post-processing is often used in dynamic latent space models~\citepSup{zhao2022, macdonald2023, zhao2023}.

    To formally describe the sequential Procrustes alignment procedure, let $\set{\hat{\bU}(t_m)}_{m=1}^M$ define the $\alpha$-variational posterior means of the latent trajectories evaluated at the observed time points. For $m = 2, \dots, M$, we sequentially solve
    \[
        \bO_m = \argmin_{\bO \in \mathcal{O}_d} \norm{\hat{\bU}(t_m)\bO - \hat{\bU}(t_{m-1}) \bO_{m-1}}_F^2. 
    \]
    We then set the final estimate of the latent trajectories evaluated at the observed time points to $\set{\hat{\bU}(t_m) \bO_m}_{m=1}^M$.  After this post-processing, the latent trajectories are identifiable up to a single orthogonal transformation shared across all time points.

\section{Derivation of the SVI Algorithm}\label{sec:der_svi}

This section contains the derivation of the stochastic variational inference algorithm detailed in Algorithm~\ref{alg:svi}, Algorithm~\ref{alg:svi_local}, and Algorithm~\ref{alg:svi_var}, which includes the proofs of Propositions~\ref{prop:unbiased_grads}--\ref{prop:beta_weights} in the main text. We assume that the reader is familiar with stochastic variational inference; however, we present a brief overview of the concepts necessary to understand the derivations in Appendix~\ref{sec:svi_overview}. Throughout this section, we continue to use the notation for the expectations of the model parameters taken with respect to the variational posterior defined at the beginning of Appendix~\ref{sec:init}. Also, for a general variational posterior $q(\btheta) = \prod_{k=1}^K q(\btheta_k)$, we use $\mathbb{E}_{-q(\btheta_k)}[\cdot]$ to denote an expectation taken with respect to all variational factors except $q(\btheta_k)$. Furthermore, we use $p(\btheta_k \mid \cdot)$ to refer to the full-conditional distribution of $\btheta_k$.

We start by re-expressing the augmented fractional likelihood in Equation (\ref{eq:augmented_lik}) in a way that is useful for deriving the full-conditional distributions of the latent variables. The derivation uses the fact that the density of a $\PG(b, c)$ random variable is
\begin{equation*}
    \PG(\omega \mid b, c) = \cosh^{b}(c/2) e^{-c^2 \omega/2} \PG(\omega \mid b, 0),
\end{equation*}
where $\cosh(x/2) = (1 + e^{x}) / (2e^{x/2})$ and $\text{PG}(\omega \mid b, c)$ is the density of a $\text{PG}(b,0)$ random variable; see \citetSup{polson2013}. As such, we can re-express the augmented fractional likelihood of the observed networks $\mathcal{Y}$ and P\'olya-gamma latent variables $\bomega$ as 
\begin{align}
    p_{\alpha}(\mathcal{Y}, \bomega \mid &\mathcal{W}, \mathcal{X}) = \prod_{m=1}^M \prod_{i \leq j} p_{\alpha}(y_{ij,t_m}, \omega_{ij,t_m} \mid [\bTheta_{t_m}]_{ij}) \nonumber \\
    &= \prod_{m=1}^M \prod_{i \leq j} \left\{\frac{e^{y_{ij,t_m}[\bTheta_{t_m}]_{ij}}}{1 + e^{[\bTheta_{t_m}]_{ij}}}\right\}^{\alpha} \PG(\omega_{ij,t_m} \mid \alpha, [\bTheta_{t_m}]_{ij}) \nonumber \\
    &= \prod_{m=1}^M \prod_{i \leq j} \left\{\frac{e^{y_{ij,t_m}[\bTheta_{t_m}]_{ij}}}{1 + e^{[\bTheta_{t_m}]_{ij}}}\right\}^{\alpha} \left\{\frac{1 + e^{[\bTheta_{t_m}]_{ij}}}{2e^{[\bTheta_{t_m}]_{ij}/2}} \right\}^{\alpha} e^{-\omega_{ij,t_m} [\bTheta_{t_m}]_{ij}^2/2} \PG(\omega_{ij,t_m} \mid \alpha, 0)  \nonumber \\
    &\propto \prod_{m=1}^M \prod_{i \leq j} \exp\left\{z_{ij,t_m} [\bTheta_{t_m}]_{ij} - \omega_{ij,t_m} [\bTheta_{t_m}]_{ij}^2/2\right\} \PG(\omega_{ij,t_m} \mid \alpha, 0) \nonumber \\
    &\propto \prod_{m=1}^M \prod_{i \leq j} \exp\left\{-\frac{\omega_{ij,t_m}}{2} \left(\frac{z_{ij,t_m}}{\omega_{ij,t_m}} - [\bTheta_{t_m}]_{ij}\right)^2\right\} e^{\frac{z_{ij,t_m}^2}{2\omega_{ij,t_m}}} \PG(\omega_{ij,t_m} \mid \alpha, 0) \label{eq:quad_lik},
\end{align}
where $z_{ij,t_m} = \alpha (y_{ij,t_m} - 1/2)$. The previous expression demonstrates that $\log p_{\alpha}(\mathcal{Y}, \bomega \mid \mathcal{W}, \mathcal{X})$ is quadratic in the basis coefficients, which combined with their Gaussian priors will result in conjugate full-conditional distributions.

\subsection{Proof of Proposition~\ref{prop:lp_weights}}

First, we show that the full-conditional distribution for $\bw_{ih}$ is $N(\bar{\bmu}_{ih}, \bar{\bSigma}_{\bw_{ih}})$ with  natural parameters $\boldeta_{ih,1} \in \Reals{\ell}$ and $\boldeta_{ih,2} \in \Reals{\ell \times \ell}$, that is, $\bar{\bmu}_{ih} = \boldeta_{ih,2}^{-1} \boldeta_{ih,1}$ and $\bar{\bSigma}_{\bw_{ih}} = \boldeta_{ih,2}^{-1}$. Throughout this section, we use $p(\bw_{ih})$ to denote $\bw_{ih}$'s prior density. Define the following residual
\[
    r_{ijh, t_m} = \frac{z_{ij,t_m}}{\omega_{ij,t_m}} - \bbeta(t_m)^{\top}\bx_{ij, t_m} - \sum_{g \neq h} \bw_{ig}^{\top}\bb(t_m)\bb(t_m)^{\top}\bw_{jg}.
\]
Starting with Equation~(\ref{eq:quad_lik}), standard manipulations show that
\begin{align*}
    p(\bw_{ih} \mid \cdot) &\propto p_{\alpha}(\mathcal{Y}, \bomega \mid \mathcal{W}, \mathcal{X}) p(\bw_{ih}) \\
    &\propto \prod_{m=1}^M \prod_{j = 1}^n \exp\left\{-\frac{\omega_{ij,t_m}}{2} \left(r_{ijh, t_m} - \bw_{ih}^{\top} \bb(t_m)\bb(t_m)^{\top}\bw_{jh}\right)^2\right\} p(\bw_{ih}) \\
    &\propto \prod_{m=1}^M \prod_{j =1}^n \exp\left\{-\frac{\omega_{ij,t_m}}{2} \left(r_{ijh, t_m} - u_{jh}(t_m) \bb(t_m)^{\top}\bw_{ih}\right)^2\right\} p(\bw_{ih}) \\
    &\propto \left[\prod_{m=1}^M \prod_{j = 1}^n N(r_{ijh,t_m} \mid u_{jh}(t_m)\bb(t_m)^{\top} \bw_{ih}, \, \omega_{ij,t_m}^{-1})\right] p(\bw_{ih}),
\end{align*}
where $N(\bx \mid \bmu, \bSigma)$ denotes a Gaussian density with mean $\bmu$ and covariance matrix $\bSigma$. The term in brackets is the likelihood for multiple linear regression with a Gaussian response $r_{ijh,t_m}$, covariate vector $u_{jh}(t_m) \bb(t_m)$, sample weight $\omega_{ij,t_m}$, and coefficients $\bw_{ih}$. Since $p(\bw_{ih}) = N(\bw_{ih} \mid \mathbf{0}_{\ell}, \gamma_h^{-1} \bOmega_i^{-1})$, a standard Bayesian linear regression-type calculation demonstrates that the full-conditional of $\bw_{ih}$ is multivariate Gaussian with the following natural parameters
\begin{align*}
    \boldeta_{ih,1} &= \sum_{m=1}^{M} \sum_{j=1}^n \omega_{ij,t_m} u_{jh}(t_m) r_{ijh,t_m} \bb(t_m), \\
    \boldeta_{ih,2} &= \sum_{m=1}^{M} \sum_{j =1}^n \omega_{ij,t_m} \{u_{jh}(t_m)\}^2 \, \bb(t_m) \bb(t_m)^{\top} + \gamma_h \bOmega_{i}.
\end{align*}
This result also demonstrates that $p(\bw_{ih} \mid \cdot)$ is conjugate within the exponential family.

Next, we derive the optimal variational distribution and the corresponding natural gradient updates. Since $p(\bw_{ih} \mid \cdot)$ is multivariate Gaussian, the optimal variational distribution for $\bw_{ih}$ is also multivariate Gaussian~\citepSup{bishop2006}. As such, we set $q(\bw_{ih}) = N(\bmu_{ih}, \bSigma_{\bw_{ih}})$ with natural parameters $\blambda_{ih} \in \Reals{\ell}$ and $\bLambda_{ih} \in \Reals{\ell \times \ell}$. Under this choice for $q(\bw_{ih})$ and the fact that $p(\bw_{ih} \mid \cdot)$ is within the same exponential family, according to \citetSup{hoffman2013} and detailed in Appendix~\ref{sec:svi_overview}, the natural gradients of the ELBO with respect to the variational factor's natural parameters are 
\begin{align}
    \nabla_{\blambda_{ih}} \textsf{ELBO}[q(\mathcal{W}, \brho)q(\bomega)] &= -\blambda_{ih} + \mathbb{E}_{-q(\bw_{ih})}[\boldeta_{ih,1}] \nonumber \\
    &=-\blambda_{ih} + \sum_{m=1}^M \sum_{j =1}^n \mathbb{E}_{-q(\bw_{ih})}\left[\omega_{ij, t_m} u_{jh}(t_m) r_{ijh,t_m} \right] \bb(t_m), \label{eq:lp_grad1}\\
    \nabla_{\bLambda_{ih}} \textsf{ELBO}[q(\mathcal{W}, \brho)q(\bomega)] &= -\bLambda_{ih} + \mathbb{E}_{-q(\bw_{ih})}[\boldeta_{ih,2}] \nonumber \\
    &= -\bLambda_{ih} + \mathbb{E}_{-q(\bw_{ih})}\left[\gamma_h \bOmega_{i}\right] \nonumber \\
    &\qquad+ \sum_{m=1}^{M} \sum_{j =1}^n \mathbb{E}_{-q(\bw_{ih})}\left[\omega_{ij,t_m} \{u_{jh}(t_m)\}^2 \right]\, \bb(t_m) \bb(t_m)^{\top} \label{eq:lp_grad2}.
\end{align}
Furthermore, we have that
\begin{align*}
    \mathbb{E}_{-q(\bw_{ih})}[\omega_{ij,t_m} u_{jh}(t_m) r_{ijh,t_m}] &= \mu_{jh}(t_m) \bigg[\alpha (y_{ij,t_m} - 1/2) \\
    &\qquad- \mu_{\omega_{ij,t_m}}\left\{\bmu_{\beta}(t_m)^{\top} \bx_{ij, t_m} + \sum_{g \neq h} \mu_{ig}(t_m) \mu_{jg}(t_m)\right\}\bigg], \\
    \mathbb{E}_{-q(\bw_{ih})}[\omega_{ij,t_m} \{u_{jh}(t_m)\}^2] &=  \mu_{\omega_{ij,t_m}} \left[\{\mu_{ih}(t_m)\}^2 + \bb(t_m)^{\top} \bSigma_{\bw_{ih}}\bb(t_m)\right], \\
    \mathbb{E}_{-q(\bw_{ih})}[\gamma_h \bOmega_{i}] &= \mu_{\gamma_h} \left\{\mathbb{E}_{q(\sigma_i^2)}\left[\frac{1}{\sigma_i^2}\right] \bD_{\ell}^{(1) \, \top} \bD_{\ell}^{(1)} + \frac{\be_{1}\be_{1}^{\top}}{\tau^2}\right\},
\end{align*}
where we used the independence of the latent variables under the variational posterior to simplify the expectations.

Lastly, we obtain the proposed unbiased estimates of the natural gradients by applying the time point and non-edge sampling scheme from Proposition~\ref{prop:unbiased_grads} to the summations in Equations (\ref{eq:lp_grad1}) and (\ref{eq:lp_grad2}).

\subsection{Proof of Proposition~\ref{prop:beta_weights}}

The proof proceeds similarly to the proof of Proposition~\ref{prop:lp_weights}. First, we show that the full-conditional distribution for $\bw_{k}$ is $N(\bar{\bmu}_{k}, \bar{\bSigma}_{k})$ with natural parameters $\boldeta_{k,1} \in \Reals{\ell}$ and $\boldeta_{k,2} \in \Reals{\ell \times \ell}$, that is, $\bar{\bmu}_{k} = \boldeta_{k,2}^{-1}\boldeta_{k,1}$ and $\bar{\bSigma}_{k} = \boldeta_{k,2}^{-1}$. Throughout this section, we use $p(\bw_{k})$ to denote $\bw_{k}$'s prior density. Define the following residual
\[
    e_{ijk, t_m} = \frac{z_{ij,t_m}}{\omega_{ij,t_m}} - \sum_{\ell \neq k} \beta_{\ell}(t_m) \, x_{ij\ell, t_m} - \bu_i(t_m)^{\top}\bu_j(t_m).
\]
Starting with Equation~(\ref{eq:quad_lik}), standard manipulations show that
\begin{align*}
    p(\bw_{k} \mid \cdot) &\propto p_{\alpha}(\mathcal{Y}, \bomega \mid \mathcal{W}, \mathcal{X}) p(\bw_{k}) \\
    &\propto \prod_{m=1}^M \prod_{i \leq j} \exp\left\{-\frac{\omega_{ij,t_m}}{2} \left(e_{ijk, t_m} - x_{ijk,t_m}\bb(t_m)^{\top}\bw_{k} \right)^2\right\} p(\bw_{k}) \\
    &\propto \left[\prod_{m=1}^M \prod_{i \leq j} N\left(e_{ijk,t_m} \mid x_{ijk,t_m}\bb(t_m)^{\top}\bw_{k}, \, \omega_{ij,t_m}^{-1}\right)\right]p(\bw_{k}).
\end{align*}
The term in brackets is the likelihood for multiple linear regression with a Gaussian response $e_{ijk,t_m}$, covariate vector $x_{ijk,t_m}\bb(t_m)$, sample weight $\omega_{ij,t_m}$, and coefficients $\bw_{k}$. Since $p(\bw_{k}) = N(\bw_{k} \mid \mathbf{0}_{\ell}, \bOmega_{\beta_k}^{-1})$, a standard Bayesian linear regression-type calculation demonstrates that the full-conditional of $\bw_{k}$ is multivariate Gaussian with the following natural parameters
\begin{align*}
    \boldeta_{k,1} &= \sum_{m=1}^M \sum_{i \leq j} \omega_{ij,t_m} x_{ijk,t_m} e_{ijk,t_m} \, \bb(t_m), \\
    \boldeta_{k,2} &= \sum_{m=1}^M \sum_{i \leq j} \omega_{ij,t_m} x_{ijk,t_m}^2 \, \bb(t_m) \bb(t_m)^{\top} + \bOmega_{\beta_k}.
\end{align*}
This result also demonstrates that $p(\bw_{k} \mid \cdot)$ is conjugate within the exponential family.

Next, we derive the optimal variational distribution and the corresponding natural gradient updates. Since $p(\bw_{k} \mid \cdot)$ is multivariate Gaussian, the optimal $q(\bw_{k})$ is also multivariate Gaussian~\citepSup{bishop2006}. As such, we set $q(\bw_{k}) = N(\bmu_{k}, \bSigma_{k})$ with natural parameters $\blambda_k \in \Reals{\ell}$ and $\bLambda_{k} \in \Reals{\ell \times \ell}$. Under this choice of $q(\bw_k)$ and the fact that $p(\bw_k \mid \cdot)$ is within the same exponential family, according to \citetSup{hoffman2013}, the natural gradients of the ELBO with respect to the variational factor's natural parameters are
\begin{align}
    \nabla_{\blambda_{k}} \textsf{ELBO}[q(\mathcal{W}, \brho)q(\bomega)] &= -\blambda_{k} + \mathbb{E}_{-q(\bw_k)}[\boldeta_{k,1}] \nonumber \\
    &= -\blambda_k + \sum_{m=1}^M \sum_{i \leq j} x_{ijk,t_m} \mathbb{E}_{-q(\bw_k)}\left[\omega_{ij, t_m} e_{ijk,t_m} \right] \bb(t_m) \label{eq:beta_grad1}, \\
    \nabla_{\bLambda_{k}} \textsf{ELBO}[q(\mathcal{W}, \brho)q(\bomega)] &= -\bLambda_{k} + \mathbb{E}_{-q(\bw_k)}[\boldeta_{k,2}] \nonumber \\
    &=-\bLambda_k + \mathbb{E}_{-q(\bw_k)}\left[\bOmega_{\beta_k}\right] \nonumber \\
    &\qquad + \sum_{m=1}^{M} \sum_{i \leq j} x_{ijk,t_m}^2 \mathbb{E}_{-q(\bw_k)}\left[\omega_{ij,t_m}\right] \, \bb(t_m) \bb(t_m)^{\top} \label{eq:beta_grad2}.
\end{align}
Furthermore, we have that
\begin{align*}
    \mathbb{E}_{-q(\bw_k)}[\omega_{ij,t_m} e_{ijk,t_m}] &=  \alpha (y_{ij, t_m} - 1/2) - \mu_{\omega_{ij,t_m}} \left\{\sum_{\ell \neq k} \mu_{\beta_{\ell}}(t_m) x_{ij\ell,t_m} - \bmu_i(t_m)^{\top}\bmu_j(t_m)\right\}, \\
    \mathbb{E}_{-q(\bw_k)}[\bOmega_{\beta_k}] &= \mathbb{E}_{q(\sigma_{\beta_k}^2)}\left[\frac{1}{\sigma_{\beta_k}^2}\right] \bD_{\ell}^{(r_k) \, \top} \bD_{\ell}^{(r_k)} + \sum_{s=1}^{r_k} \frac{\be_{s}\be_{s}^{\top}}{\tau_{\beta}^2},
\end{align*}
where we used the independence of the latent variables under the variational posterior to simplify the expectations.

Lastly, we obtain the proposed unbiased estimates of the natural gradients by applying the time point and non-edge sampling scheme from Proposition~\ref{prop:unbiased_grads} to the summations in Equations (\ref{eq:beta_grad1}) and (\ref{eq:beta_grad2}).

\subsection{Updating \texorpdfstring{$q(\sigma_i^2)$}{q(sigmai)}}

Starting with Equation~(\ref{eq:quad_lik}), standard calculations show that
\begin{align*}
    p(\sigma_i^2 \mid \cdot) &\propto \exp\left\{-\frac{1}{2}\left(\frac{1}{\sigma_i^2} \sum_{h=1}^d \gamma_h \norm{\bD_{\ell}^{(1)} \bw_{ih}}_2^2  + d_{\sigma} \sigma_i^2\right) \right\} (\sigma_i^2)^{c_{\sigma}/2- 1 - d(\ell - 1)/2)} \\
    &\propto \text{GIG}\left(\sigma_i^2 \mid d_{\sigma}, \sum_{h=1}^d \gamma_h \norm{\bD_{\ell}^{(1)} \bw_{ih}}_2^2, \frac{1}{2} \{c_{\sigma} - d(\ell - 1)\}\right),
\end{align*}
which is a generalized inverse Gaussian (GIG) distribution, which we denote by $\text{GIG}(a,b,p)$, with density
\[
    \text{GIG}(x \mid a, b, p) = \frac{(a/b)^{p/2}}{2 K_p(\sqrt{ab})} x^{p-1} e^{-(ax + b/x)/2},
\]
where $K_p(x)$ is the modified Bessel function of the second kind. The generalized inverse-Gaussian distribution is in the exponential family with natural parameters $-a/2$, $-b/2$, and $p - 1$. As such, the optimal variational factor is also a generalized inverse Gaussian, so we set $q(\sigma_i^2) = \text{GIG}(\bar{a}_{i}, \bar{b}_i, \bar{p}_i)$.

Applying the formula for the natural gradients in \citetSup{hoffman2013} and the chain-rule, we have
\begin{align*}
    \nabla_{\bar{a}_i} \textsf{ELBO}[q(\mathcal{W}, \brho)q(\bomega)] &= -2 \times \nabla_{-\bar{a}_i/2}\textsf{ELBO}[q(\mathcal{W}, \brho)q(\bomega)] = -\bar{a}_i + d_{\sigma}, \\
    \nabla_{\bar{p}_i} \textsf{ELBO}[q(\mathcal{W}, \brho)q(\bomega)] &= -\bar{p}_i + \frac{1}{2} \{c_{\sigma} - d(\ell-1)\}, \\
    \nabla_{\bar{b}_i} \textsf{ELBO}[q(\mathcal{W},\brho)q(\bomega)] &= -2 \times \nabla_{-\bar{b}_i/2} \textsf{ELBO}[q(\mathcal{W}, \brho)q(\bomega)] \\
    &= -b_i + \sum_{h=1}^d \mathbb{E}_{-q(\sigma_i^2)}[\gamma_h\norm{\bD_{\ell}^{(1)} \bw_{ih}}_2^2] \\
    &= -b_i + \sum_{h=1}^d \mu_{\gamma_h} \left[\bmu_{\bw_{ih}} \bD_{\ell}^{(1) \, \top} \bD_{\ell}^{(1)} \bmu_{\bw_{ih}} + \text{tr}(\bD_{\ell}^{(1) \, \top}\bD_{\ell}^{(1)} \bSigma_{\bw_{ih}})\right].
\end{align*}
The gradients for $\bar{a}_i$ and $\bar{p}_i$ do not depend on the other parameters of the variational distribution, so we can set these variational parameters to their maximizers, that is, $\bar{a}_i = d_{\sigma}$ and $\bar{p}_i = \{c_{\sigma} - d(\ell - 1)\}/2$. As such, we only need to update $\bar{b}_i$ at each iteration. Lastly, we need the following expectation for the other gradient updates:
\[
    \mathbb{E}_{q(\sigma_i^2)}\left[\frac{1}{\sigma_i^2}\right] = \frac{\sqrt{\bar{b}_i} K_{\bar{p}_i + 1}(\sqrt{\bar{a}_i\bar{b}_i})}{\sqrt{\bar{b}_i} K_{\bar{p}_i}(\sqrt{\bar{a}_i\bar{b}_i})} - \frac{2\bar{p}_i}{\bar{b}_i}.
\]

\subsection{Updating \texorpdfstring{$q(\sigma_{\beta_k}^2)$}{q(sigmabeta)}}

Starting with Equation~(\ref{eq:quad_lik}), standard calculations show that
\begin{align*}
    p(\sigma_{\beta_k}^2 \mid \cdot) &\propto \exp\left\{-\frac{1}{2}\left(\frac{1}{\sigma_{\beta_k}^2} \norm{\bD_{\ell}^{(r_k)} \bw_{k}}_2^2  + d_{\sigma} \sigma_{\beta_k}^2\right) \right\} (\sigma_{\beta_k}^2)^{c_{\sigma}/2- 1 - (\ell - r_k)/2)} \\
    &\propto \text{GIG}\left(\sigma_{\beta_k}^2 \mid d_{\sigma}, \norm{\bD_{\ell}^{(r_k)} \bw_{k}}_2^2, \frac{1}{2} \{c_{\sigma} - (\ell - r_k)\}\right), 
\end{align*}
which is a generalized inverse Gaussian distribution. The generalized inverse-Gaussian distribution is in the exponential family with natural parameters $-a/2$, $-b/2$, and $p - 1$. As such, the optimal variational factor is also a generalized inverse Gaussian, so we set $q(\sigma_{\beta_k}^2) = \text{GIG}(\bar{a}_{\beta_k}, \bar{b}_{\beta_k}, \bar{p}_{\beta_k})$. 

Applying the formula for the natural gradients in \citetSup{hoffman2013} and the chain-rule, we have
\begin{align*}
    \nabla_{\bar{a}_{\beta_k}} \textsf{ELBO}[q(\mathcal{W}, \brho)q(\bomega)] &= -2 \times \nabla_{-\bar{a}_{\beta_k}/2}\textsf{ELBO}[q(\mathcal{W}, \brho)q(\bomega)] = -\bar{a}_{\beta_k} + d_{\sigma}, \\
    \nabla_{\bar{p}_{\beta_k}} \textsf{ELBO}[q(\mathcal{W}, \brho)q(\bomega)] &= -\bar{p}_{\beta_k} + \frac{1}{2} \{c_{\sigma} - (\ell-r_k)\}, \\
    \nabla_{\bar{b}_{\beta_k}} \textsf{ELBO}[q(\mathcal{W},\brho)q(\bomega)] &= -2 \times \nabla_{-\bar{b}_{\beta_k}/2} \textsf{ELBO}[q(\mathcal{W}, \brho)q(\bomega)] \\
    &= -b_{\beta_k} + \mathbb{E}_{-q(\sigma_{\beta_k}^2)}[\norm{\bD_{\ell}^{(r_k)} \bw_{k}}_2^2] \\
    &= -b_{\beta_k} + \left[\bmu_{\bw_{k}} \bD_{\ell}^{(r_k) \, \top} \bD_{\ell}^{(r_k)} \bmu_{\bw_{k}} + \text{tr}(\bD_{\ell}^{(r_k) \, \top}\bD_{\ell}^{(r_k)} \bSigma_{\bw_{k}})\right].
\end{align*}
The gradients for $\bar{a}_{\beta_k}$ and $\bar{p}_{\beta_k}$ do not depend on the other parameters of the variational distribution, so we can set these variational parameters to their maximizers, that is, $\bar{a}_{\beta_k} = d_{\sigma}$ and $\bar{p}_{\beta_k} = \{c_{\sigma} - (\ell - r_k)\}/2$. As such, we only need to update $\bar{b}_{\beta_k}$ at each iteration.

\subsection{Updating \texorpdfstring{$q(\nu_h)$}{q(nuh)}}

Starting with Equation~(\ref{eq:quad_lik}), standard calculations show that the full-conditional distribution of each $\nu_h$ parameter is gamma distributed. For $h = 1, \dots, d$, let $p(\nu_h)$ denote the gamma prior distribution for $\nu_h$.  We have that
\[
    p(\nu_h \mid \cdot) \propto \nu_h^{(d - h + 1) n\ell/2} \exp\left\{-\frac{\nu_h}{2} \sum_{s=h}^d \gamma_{s,h}\sum_{i=1}^n \bw_{is}^{\top} \bOmega_{i} \bw_{is} \right\} p(\nu_h),
\]
where $\gamma_{s, h} = \prod_{g=1}^{h-1} \nu_g \times \prod_{g=h+1}^{s} \nu_g$. Based on these expressions, we recognize that 
\begin{align*}
    p(\nu_1 \mid \cdot) &= \text{Gamma}\left(a_1 + \frac{dn\ell}{2}, 1 + \frac{1}{2} \sum_{s=1}^d \gamma_{s, 1} \sum_{i=1}^n \bw_{is}^{\top} \bOmega_i \bw_{is}\right), \\
    p(\nu_h \mid \cdot) &= \text{Gamma}\left(a_2 + \frac{(d-h+1)n\ell}{2}, 1 + \frac{1}{2} \sum_{s=h}^d \gamma_{s, h} \sum_{i=1}^n \bw_{is}^{\top} \bOmega_i \bw_{is}\right),
\end{align*}
which are within the exponential family. Based on the above full-conditional distributions, we set $q(\nu_h)$ to their optimal forms. That is we set $q(\nu_h) = \text{Gamma}(\bar{c}_h, \bar{d}_h)$ with natural parameters $\bar{c}_h - 1$ and $\bar{d}_h$. 

Applying the formula for the natural gradients in \citetSup{hoffman2013}, we have that
\begin{align*}
    \nabla_{\bar{c}_1}\textsf{ELBO}[q(\mathcal{W}\brho)q(\bomega)] &= -\bar{c}_1 + a_1 + \frac{dn\ell}{2}, \\ 
    \nabla_{\bar{c}_h}\textsf{ELBO}[q(\mathcal{W}\brho)q(\bomega)] &= -\bar{c}_h + a_2 + \frac{(d -h + 1)n\ell}{2} \qquad 1 < h \leq d, \\ 
    \nabla_{\bar{d}_h}\textsf{ELBO}[q(\mathcal{W}, \brho)q(\bomega)] &= -\bar{d}_h + 1 + \frac{1}{2} \sum_{s=h}^d \mathbb{E}_{-q(\nu_h)}[\gamma_{s,1}] \sum_{i=1}^n \mathbb{E}_{q(\bw_{is}, \sigma_i^2)}[\bw_{is}^{\top} \bOmega_i \bw_{is}],
\end{align*}
where
\begin{align*}
    \mathbb{E}_{-q(\nu_h)}[\gamma_{s, h}] &= \prod_{g=1}^{h-1}\mathbb{E}_{q(\nu_g)}[\nu_g] \times \prod_{g=h+1}^{s} \mathbb{E}_{q(\nu_g)}[\nu_g] = \prod_{g=1}^{h-1} \frac{\bar{c}_g}{\bar{d}_g} \times \prod_{g=h+1}^{s} \frac{\bar{c}_g}{\bar{d}_g}, \\
    \mathbb{E}_{q(\bw_{is},\sigma_i^2)}[\bw_{is}^{\top} \bOmega_i \bw_{is}] &= \mathbb{E}_{q(\sigma_i^2)}\left[\frac{1}{\sigma_i^2}\right] \mathbb{E}_{q(\bw_{is})}[\bw_{is}^{\top} \bD_{\ell}^{(1) \, \top} \bD_{\ell}^{(1)}\bw_{is}] + \frac{1}{\tau^2} \mathbb{E}_{q(\bw_{is})}[\bw_{is,1}^2] \\
    &= \mathbb{E}_{q(\sigma_i^2)}\left[\frac{1}{\sigma_i^2}\right] \left\{\bmu_{\bw_{is}}^{\top}\bD_{\ell}^{(1) \, \top} \bD_{\ell}^{(1)} \bmu_{\bw_{is}} + \tr(\bD_{\ell}^{(1) \, \top}\bD_{\ell}^{(1)} \bSigma_{\bw_{is}})\right\} \\
    &\qquad\qquad+ \frac{1}{\tau^2} \{\bmu_{\bw_{is}, 1}^2 + [\bSigma_{\bw_{is}}]_{11}\}.
\end{align*}
The gradients for $\bar{c}_{h}$ do not depend on the other parameters of the variational distribution, so we can set these variational parameters to their maximizers, that is, $\bar{c}_{1} = a_1 + dn\ell/2$ and $\bar{c}_h = a_2 + (d-h+1)n\ell/2$ for $h > 1$. As such, we only need to update $\bar{d}_{h}$ at each iteration. 

\subsection{Updating \texorpdfstring{$q(\omega_{ij,t_m})$}{q(omega)}}

From the form of the augmented joint distribution in Equation~(\ref{eq:augmented_lik}), we have that the full-conditionals for each local P\'olya-gamma latent variable is
\[
    \omega_{ij,t_m} \mid \cdot \sim \text{PG}(\alpha, [\bTheta_{t_m}]_{ij}),
\]
which is in the exponential family with natural parameter $-[\bTheta_{t_m}]_{ij}^2 / 2$. Recall SVI sets  the variational factors of each local latent variable to their optimal forms at each iteration. In particular, the optimal variational factor $q(\omega_{ij,t_m}) = \PG(\alpha, c_{ij,t_m})$ with $c_{ij,t_m}^2 = \mathbb{E}_{-q(\omega_{ij,t_m})}\{[\bTheta_{t_m}]_{ij}^2\}$. A straightforward calculation shows that
\begin{align*}
    \mathbb{E}_{-q(\omega_{ij,t_m})}\{[\bTheta_{t_m}]_{ij}^2\} &= (\mathbb{E}_{-q(\omega_{ij,t_m})}\{[\bTheta_{t_m}]_{ij}\})^2 + \text{Var}_{-q(\omega_{ij,t_m})}([\bTheta_{t_m}]_{ij}) \\
    &= [\bmu_{\bbeta}(t_m)^{\top} \bx_{ij,t_m} + \bmu_i(t_m)^{\top}\bmu_j(t_m)]^2 + \\
    &\qquad \text{tr}\set{\bx_{ij,t_m} \bx_{ij,t_m}^{\top} \bSigma_{\bbeta}(t_m)} + \text{tr}\set{\bSigma_{i}(t_m)\bSigma_j(t_m)}+ \\
    &\qquad \bmu_j(t_m)^{\top}\bSigma_i(t_m) \bmu_j(t_m) + \bmu_i(t_m)^{\top} \bSigma_j(t_m) \bmu_i(t_m),
\end{align*}
where $\text{Var}_{-q(\omega_{ij,t_m})}(\cdot)$ denotes the variance taken with respect to all variational factors expect $q(\omega_{ij,t_m})$. In addition, the expectation of a $\PG(b,c)$ random variable is 
\[
    \frac{b}{2c} \left(\frac{e^c - 1}{1 + e^c}\right),
\]
so that
\[
    \mathbb{E}_{q(\omega_{ij,t_m})}[\omega_{ij,t_m}] = \frac{\alpha}{2c_{ij,t_m}} \left(\frac{e^{c_{ij,t_m}} - 1}{1 + e^{c_{ij,t_m}}}\right).
\]

    \subsection{Proof of Proposition~\ref{prop:unbiased_grads}}
    
     To prove the result, we only need to show that $\mathcal{B}(H_i)$ is an unbiased estimate of $H_i$. To start, we express
    \[
        H_i = \sum_{m=1}^M \sum_{j \in \mathcal{N}_{i,t_m}} h_{ij,m} + \sum_{m=1}^M \sum_{j \in \mathcal{N}_{i,t_m}^c} h_{ij,m}.
    \] 
    Next, let $(Z_1, \dots, Z_M) \in \set{0, 1}^M$ denote a collection of random variables indicating the selection of time point $t_m$, such that, $\sum_{m=1}^M Z_m = \abs{\mathcal{M}}$. Similarly, let $\set{W_{jm}^{(i)}}_{j\in \mathcal{N}_{i,t_m}^c} \in \set{0, 1}^{\abs{\mathcal{N}_{i,t_m}^c}}$ denote the collection of random variables indicating the selection of node $j$ at time $t_m$ to be in $\mathcal{N}_{i,t_m}^{c \, *}$, so that $\sum_{j\in \mathcal{N}_{i,t_m}} W_{jm}^{(i)} = \abs{\mathcal{N}_{i,t_m}^{c \, *}}$ for all $1 \leq i \leq n$. Under the uniform random sampling without replacement scheme, we have that
    \[
        \mathbb{P}(Z_m = 1) = \frac{\abs{\mathcal{M}}}{M}, \qquad \mathbb{P}(W_{jm}^{(i)} = 1 \mid Z_m) = \begin{cases}
            0, & Z_{m} = 0   \\
            \frac{\abs{\mathcal{N}_{i,t_m}^{c \, *}}}{\abs{\mathcal{N}_{i,t_m}^c}}, & Z_m = 1,
        \end{cases}
    \]
    so that $\mathbb{E}[Z_m] = \frac{\abs{\mathcal{M}}}{M}$ and
    \[
        \mathbb{E}[Z_m W_{jm}^{(i)}] = \mathbb{E}[Z_m \mathbb{E}[W_{jm}^{(i)} \mid Z_m]] = \frac{\abs{\mathcal{N}_{i,t_m}^{c \, *}}}{\abs{\mathcal{N}_{i, t_m}^c}} \mathbb{E}[Z_m \ind{Z_m = 1}] = \frac{\abs{\mathcal{N}_{i,t_m}^{c \, *}}}{\abs{\mathcal{N}_{i, t_m}^c}} \frac{\abs{\mathcal{M}}}{M}.
    \]
    In terms of these indicator variables, we have
    \begin{align*}
        \mathcal{B}(H_i) &= \frac{M}{\abs{\mathcal{M}}} \sum_{m=1} \sum_{j \in \mathcal{N}_{i,t_m}} Z_m  h_{ij,m} + \frac{M}{\abs{\mathcal{M}}} \frac{\abs{\mathcal{N}_{i,t_m}^c}}{\abs{\mathcal{N}_{i,t_m}^{c \, *}}} \sum_{m=1}^M \sum_{j \in \mathcal{N}_{i,t_m}^{c}} Z_m W_{jm}^{(i)} h_{ij,m}.
    \end{align*}
    Using the formulas for the previous expectations, we have $\mathbb{E}[\mathcal{B}(H_i)] = H_i$.

    \section{Proof of Theorem~\ref{thm:vb_consistency}}\label{sec:proof_vb_con}
 
    \subsection{Preliminaries}
    
    The proof of Theorem~\ref{thm:vb_consistency} is based on Theorem 3.3 in \citetSup{yang2020} and uses an argument based on the chain rule of KL divergences introduced by \citetSup{zhao2022}, who obtained consistency results for a discrete-time dynamic LSM. The proof consists of two parts. First, we show that the proposed P-spline prior for dynamic LSMs with appropriately chosen variance parameters places sufficient mass on KL neighborhoods centered at the true parameters. According to the theory developed by \citetSup{bhattacharya2019}, this result establishes that the fractional posterior contracts about the true parameters at the desired rate. Next, we verify the conditions of Theorem 3.3 in \citetSup{yang2020}  using a technique introduced by \citetSup{zhao2022} to demonstrate that the $\alpha$-variational posterior inherits the asymptotic properties of the fractional posterior without having to specify appropriate variance parameters. We establish auxiliary technical results in Appendix~\ref{subsec:aux_results}. In addition, the proofs use facts about spline approximations, which we briefly review in Appendix~\ref{sec:spline_overview}.

First, we layout some preliminaries results and definitions. For the remainder of this document, we define the following rate,
\begin{equation}\label{eq:rate}
    \epsilon_{n,M} = \max\left\{\left(\frac{L}{nM}\right)^{1/5},  \sqrt{\frac{\log nM}{nM}}\right\}.
\end{equation}
In addition, let $\Pi_{\mathcal{W} \mid \brho}$, $\Pi_{\mathcal{W}_u \mid \brho}$ and $\Pi_{\mathcal{W}_{\beta} \mid \brho}$ denote the conditional prior measures on the basis coefficients with densities
    \begin{align*}
        p(\mathcal{W} \mid \brho) &= p(\mathcal{W}_u \mid \brho) p(\mathcal{W}_{\beta} \mid \brho), \\
        p(\mathcal{W}_u \mid \brho) &\propto \prod_{i=1}^n \prod_{h=1}^d \exp\left(-\frac{\gamma_h}{2} \bw_{ih}^{\top}\bOmega_i \bw_{ih}\right), \\
        p(\mathcal{W}_\beta \mid \brho) &\propto \prod_{k=1}^p \exp\left(-\frac{1}{2} \bw_{k}^{\top}\bOmega_{\beta_k} \bw_{k}\right).
    \end{align*} 
Lastly, we state the following corollary to Lemma~\ref{lemma:spline_approx} in Appendix~\ref{sec:spline_overview} on the existence of certain spline approximations.
    \begin{corollary}\label{cor:w_approx}
        If the elements of $\bU_0(t)$ and $\bbeta_0(t) $ satisfy Assumptions \ref{assump:func_space}--\ref{assump:lp}, then there exists spline approximations $\tilde{\beta}_{0k}(t) = \bw_{0k}^{\top}\bb(t)$ for $1 \leq k \leq p$ and $\tilde{u}_{ih}(t) = \bw_{0ih}^{\top} \bb(t)$ for $1 \leq i \leq n$ and $1 \leq h \leq d$ with $\ell \asymp (nM)^{1/5}$, such that 
        \begin{align}
            \norm{\beta_{0k}(t) - \bw_{0k}^{\top}\bb(t)}_{L_{\infty}[0,1]} &\lesssim \ell^{-1} \norm{\beta'_{0k}}_{L_{\infty}[0,1]} \lesssim \ell^{-1}L  \lesssim \epsilon_{n,M} \nonumber, \\
            \norm{\bD_{\ell}^{(1)} \bw_{0k}}_2 &\lesssim \norm{\beta'_{0k}}_{L_{\infty}[0,1]} \lesssim L, \label{eq:b_approx}
    \end{align}
    for $1 \leq k \leq p$ and
        \begin{align}
        \norm{u_{0ih}(t) - \bw_{0ih}^{\top}\bb(t)}_{L_{\infty}[0,1]} &\lesssim \ell^{-1} \norm{u'_{0ih}}_{L_{\infty}[0,1]} \lesssim \ell^{-1} L  \lesssim \epsilon_{n,M}, \nonumber \\
            \norm{\bD_{\ell}^{(1)} \bw_{0ih}}_2 &\lesssim \norm{u'_{0ih}}_{L_{\infty}[0,1]} \lesssim L, \label{eq:u_approx}
    \end{align}
        for $1 \leq i \leq n$ and $1 \leq h \leq d$, where $\epsilon_{n,M}$ is defined in Equations (\ref{eq:rate}).
    \end{corollary}
    \noindent This corollary states a well-known result from classical B-spline theory that there exist splines that approximate the true functions at our desired rate when we choose the spline basis dimension $\ell^{-1} \asymp \epsilon_{n,M}$. As such, it remains to show that the $\alpha$-variational posterior concentrates on splines close to the ones found by applying Corollary~\ref{cor:w_approx}.
   
   \subsection{KL Support of the P-Spline Prior}

    In this section, we prove Lemma~\ref{lemma:kl_neighborhood}, which establishes  the support of the P-spline prior for dynamic LSMs on KL neighborhoods about the true parameters. The proof uses techniques similar to those used to prove Theorem 3.2 (a) in \citetSup{zhao2022}, which established the KL support of Gaussian random walk priors for discrete-time LSMs. However, unlike \citetSup{zhao2022} who placed Gaussian random walk priors on the discrete-time latent trajectories, we place them on the basis coefficients of the spline approximations.

    Before presenting the result, we need the following lemma on the small-ball probability of first-order Gaussian random walks. Recall that under our assumptions, the P-spline prior for dynamic LSMs takes the form of a first-order Gaussian random walk on the basis coefficients, so naturally the prior's KL support depends on its properties. The proof is provided in Appendix~\ref{subsec:aux_results}.
    \begin{lemma}\label{lemma:small_ball}
        If the components of $\bw = (w_1, \dots, w_T)^{\top}$ follow a first-order Gaussian random walk with initial variance $\tau^2$ and transition variance $\sigma^2$, that is,
        \[
            w_1 \sim N(0, \tau^2), \qquad w_t \mid w_{t-1} \sim N(w_{t-1}, \sigma^2), \quad t = 2, \dots T,
        \]
        then for any vector $\bw_0 = (w_{01}, \dots, w_{0T})^{\top} \in \Reals{T}$, we have that
        \[
            \mathbb{P}(\norm{\bw - \bw_0}_2 \leq \delta) \gtrsim  \exp\left(-\frac{\norm{\bD_T^{(1)}\bw_0}_2^2}{\sigma^2} - \frac{w_{0}^2}{2\tau^2}\right) \exp\left(-C \frac{T^3 \sigma^2}{\delta^2} - \log\frac{T}{\delta}\right).
        \]
        for some constant $C > 0$.
    \end{lemma}

    In what follows, let $p_0(\mathcal{Y} \mid \mathcal{X})$ denote the density under the true data-generating process and $p_{\mathcal{W}}(\mathcal{Y} \mid \mathcal{X}) = p(\mathcal{Y} \mid \mathcal{W}, \mathcal{X})$ denote the density with the latent functions approximated by B-splines with basis coefficients $\mathcal{W}$. 

    \begin{lemma}[KL support of the P-spline prior for dynamic LSMs]\label{lemma:kl_neighborhood}
        Suppose the true data generating process satisfies model (\ref{eq:dynlsm})--(\ref{eq:dynlsm_lr}) with true latent functions $\bU_0(t)$ and $\bbeta_0(t)$ and observed covariates $\mathcal{X}$ that satisfy Assumptions \ref{assump:func_space}--\ref{assump:x}.  Denote the $\epsilon$-ball for the KL neighborhood centered at $\set{\bU_0(t), \bbeta_0(t)}$ as
        \[
            B_{n,M}(\mathcal{W}; \epsilon) = \left\{ \mathcal{W}  \, : \, \int p_0 \log \frac{p_{0}}{p_{\mathcal{W}}} d\mu \leq \frac{1}{2} n(n+1)M \epsilon^2, \int p_{0} \log^2 \frac{p_{0}}{p_{\mathcal{W}}} d\mu \leq \frac{1}{2} n (n+1) M \epsilon^2 \right\},
        \]
        where $\mu$ is a common dominating measure. Define $\brho^* = \set{\sigma_i = \sigma_u^*, \sigma_{\beta_k} = \sigma_\beta^*, \nu_h^* = 1, 1 \leq i \leq n, 1 \leq k \leq p, 1 \leq h \leq d}$, where $\sigma_u^* = b_1 \tilde{L} \epsilon_{n,M}^2$ and $\sigma_{\beta}^* = b_2 \tilde{L} \epsilon_{n,M}^2$ for constants $b_1, b_2 > 0$, and $\tilde{L} = \max\{L, (nM)^{-3/2}\}$. Under $\Pi_{\mathcal{W} \mid \brho^*}$ with $r_1 = \dots = r_p = 1$ and $\bb(t)$ a B-spline basis of dimension $\ell \asymp (nM)^{1/5}$, we have for $\lambda$-almost all $\set{t_m}_{m=1}^M$ that
        \[
            \Pi_{\mathcal{W} \mid \brho^*}\left\{B_{n,M}(\mathcal{W}; \epsilon_{n,M})\right\} \gtrsim e^{-C\frac{1}{2} n(n+1) M  \epsilon_{n,M}^2},
        \]
        for some constant $C > 0$ and $\epsilon_{n,M}$ defined in Equation~(\ref{eq:rate}).
    \end{lemma}
    \begin{proof}

        We start by using Corollary~\ref{cor:w_approx} to find a spline approximation $\tilde{\beta}_{0k}(t) = \bw_{0k}^{\top}\bb(t)$ to $\beta_k(t)$ for $1 \leq k \leq p$ and $\tilde{u}_{ih}(t) = \bw_{0ih}^{\top} \bb(t)$ to $u_{ih}(t)$ for $1 \leq i \leq n$ and $1 \leq h \leq d$ that satisfy Equations (\ref{eq:b_approx})--(\ref{eq:u_approx}).
        In addition, define the events 
        \begin{align*}
            E_1 &= \bigcap_{k=1}^p \set{\bw_k \, : \, \norm{\bw_{0k} - \bw_k}_2 \leq c_1 \epsilon_{n,M}}, \\
            E_2 &= \bigcap_{i=1}^n \bigcap_{h=1}^d \set{\bw_{ih} \, : \, \norm{\bw_{0ih} - \bw_{ih}}_2 \leq c_2 \epsilon_{n,M}}
        \end{align*}
        for some constants $c_1 > 0$ and $c_2 > 0$ specified later. We begin by showing that $E_1 \cap E_2 \subset B_{n,M}(\mathcal{W}; \epsilon_{n,M})$ for $c_1$ and $c_2$ chosen appropriately.
        
        First, we upper-bound the two terms that define the KL neighborhood by the squared Frobenius norm between the log-odds matrices. By Lemma~\ref{lemma:ub_bern_kl} in Appendix~\ref{subsec:aux_results}, we have
        \[
            D_{KL}(p_0, p_{\mathcal{W}}) \lesssim \sum_{m=1}^M \sum_{i\leq j} ([\bTheta_{0t_m}]_{ij} - [\bTheta_{t_m}]_{ij})^2.
        \]
        Furthermore, we have
        \begin{align*}
            V_2(p_0, p_{\mathcal{W}}) := \int p_0 \log^2 \frac{p_0}{p_{\mathcal{W}}} d\mu \lesssim \sum_{m=1}^M \sum_{i \leq j} p_{0ij, t_m} \log^2 \frac{p_{0ij,t_m}}{p_{ij,t_m}} +  (1 - p_{0ij,t_m}) \log^2 \frac{1 - p_{0ij,t_m}}{1 - p_{ij,t_m}},
        \end{align*}
        where $p_{0ij,t_m}$ and $p_{ij,t_m}$ are the probabilities of forming and edge between node $i$ and $j$ at time $t_m$ according to $\bTheta_{0t_m}$ and $\bTheta_{t_m}$, respectively. By Assumption~\ref{assump:func_space}, the elements of both $\bU_0(t)$ and $\bbeta_0(t)$ are bounded for $\lambda$-almost all $\set{t_m}_{m=1}^M$. Furthermore, the elements of $\mathcal{X}_{t_m}$ are bounded by Assumption~\ref{assump:x} for $m = 1, \dots M$. As such, the probabilities of forming an edge are bounded away from $0$ and $1$ for $\lambda$-almost all $\set{t_m}_{m=1}^M$, so we can apply Lemma 1 of \citetSup{jeong2021} to show that the right hand side of the previous expression is bounded above by $\sum_{m=1}^M \sum_{i \leq j}([\bTheta_{0t_m}]_{ij} - [\bTheta_{t_m}]_{ij})^2$ multiplied by a positive constant. As such, we have
        \begin{equation}\label{eq:kl_ub}
            \max\set{D_{KL}(p_{0}, p_{\mathcal{W}}), V_2(p_{0}, p_{\mathcal{W}})} \lesssim \sum_{m=1}^M \sum_{i \leq j} ([\bTheta_{0t_m}]_{ij} - [\bTheta_{t_m}]_{ij})^2.
        \end{equation}
        Therefore, we only need to lower bound the prior probability of the set
        \begin{align*}
            \bigg\{\sum_{m=1}^M \sum_{i \leq j} ([\bTheta_{0t_m}]_{ij} &- [\bTheta_{t_m}]_{ij})^2 \leq C_1 \frac{1}{2}n (n+1) M \epsilon_{n,M}^2\bigg\} \\
            &\supset \bigg\{\max_m \max_{i,j} ([\bTheta_{0t_m} ]_{ij}- [\bTheta_{t_m}]_{ij})^2 \leq C_1 \epsilon_{n,M}^2\bigg\},
        \end{align*}
        for $C_1 > 0$ chosen to satisfy Equation~(\ref{eq:kl_ub}). Given $i,j$ and $m$, we have
        \begin{align*}
            \abs{[\bTheta_{0t_m}]_{ij} - [\bTheta_{t_m}]_{ij}} &\leq \norm{\bbeta_0(t_m) - \bbeta(t_m)}_2 \norm{\bx_{ij,t_m}}_2 + \abs{\bu_{0i}(t_m)^{\top}\bu_{0j}(t_m) - \bu_i(t_m)^{\top}\bu_j(t_m)} \\
            &\leq K_x \norm{\bbeta_0(t_m) - \bbeta(t_m)}_2 + \abs{\bu_{0i}(t_m)^{\top}\bu_{0j}(t_m) - \bu_i(t_m)^{\top}\bu_j(t_m)},
        \end{align*}
        where we used Assumption~\ref{assump:x} in the last line. Using the triangle and Cauchy-Schwarz inequalities, we can bound the second term as follows:
        \begin{align*}
            \abs{\bu_{0i}(t_m)^{\top}\bu_{0j}(t_m) - \bu_i(t_m)^{\top}\bu_j(t_m)} &\leq \abs{\{\bu_i(t_m)^{\top} - \bu_{i0}(t_m)\}^{\top}\bu_{0j}(t_m)} \\
            &\qquad\qquad + \abs{\bu_{i}(t_m)^{\top}\{\bu_{j}(t_m) - \bu_{0j}(t_m)\}} \\
            &\leq\max_i \norm{\bu_{0i}(t_m) - \bu_i(t_m)}_2 \{\norm{\bu_{0i}(t_m) - \bu_i(t_m)}_2 \\
            &\qquad\qquad + 2 \norm{\bu_{0i}(t_m)}_2\} \\
            &\leq \max_i \norm{\bu_{0i}(t_m) - \bu_i(t_m)}_2 \{\norm{\bu_{0i}(t_m) - \bu_i(t_m)}_2 + 2C_2\},
        \end{align*}
        where $C_2$ is a constant such that $C_2 > d^{1/2} \max_{ih} \norm{u_{0ih}}_{L_{\infty}[0,1]}$, which exists for $\lambda$-almost $\set{t_m}_{m=1}^M$ by Assumption~\ref{assump:lp}. Notice that when $\max_i \norm{\bu_{0i}(t_m) - \bu_i(t_m)}_2 \leq \epsilon_{n,M} / [(2 + c_0) C_2] \leq C_2$ for some constant $c_0 > 1$, we have
        \[
            \max_i \norm{\bu_{0i}(t_m) - \bu_i(t_m)}_2 \{\norm{\bu_{0i}(t_m) - \bu_i(t_m)}_2 + 2C_2\} \leq \frac{\epsilon_{n,M}}{(2 + c_0)C_2} 3C_2 \leq \epsilon_{n,M}.
        \] 
        As such, we define the events $\tilde{E}_1 = \set{\max_m \norm{\bbeta_0(t_m) - \bbeta(t_m)} \leq c_3 \epsilon_{n,M}}$ and $\tilde{E}_2 = \set{\max_{i,m} \norm{\bu_{0i}(t_m) - \bu_i(t_m)}_2 \leq c_4 \epsilon_{n,M}}$, where $c_3 = K_x^{-1}$, $c_4 = 1 / [(2 + c_0) C_2]$. Based on the previous observations, we have that $\tilde{E}_1 \cap \tilde{E}_2 \subset B_{n,M}(\mathcal{W};\epsilon_{n,M})$. 
    
        To establish that $E_1 \cap E_2 \subset B_{n,M}(\mathcal{W};\epsilon_{n,M})$, we will show that  $E_1 \subset \tilde{E}_1$ and $E_2 \subset \tilde{E}_2$. We start by showing $E_1 \subset \tilde{E}_1$ for an appropriately chosen constant $c_1$. We have that
        \begin{align*}
            \max_{m} \norm{\bbeta_0(t_m) - \bbeta(t_m)}_2 &\leq  \sqrt{p} \max_{m,k}\abs{\beta_{0k}(t_m) -\beta_k(t_m)} \\
            &= \sqrt{p} \max_{m,k} \abs{\beta_{0k}(t_m) - \tilde{\beta}_{0k}(t_m) + \tilde{\beta}_{0k}(t_m) - \beta_{0k}(t_m)} \\
            &\leq \sqrt{p} \max_{k} \norm{\beta_{0k} - \tilde{\beta}_{0k}}_{L_{\infty}[0,1]} + \sqrt{p} \max_{k,m} \abs{\tilde{\beta}_{0k}(t_m) - \beta_k(t_m)} \\
            &\lesssim \epsilon_{n,M} +  \max_{k,m} \abs{\tilde{\beta}_{0k}(t_m) - \beta_k(t_m)}.
        \end{align*}
        The third line holds for $\lambda$-almost all $\set{t_m}_{m=1}^M$ and the fourth line follows from the definition of $\tilde{\beta}_{0k}(t)$. In addition, for a given $k$, we have
        \[
            \max_m \abs{\tilde{\beta}_{0k}(t_m) - \beta_k(t_m)} = \max_m \abs{(\bw_{0k} - \bw_k)^{\top} \bb(t_m)} \leq \norm{\bw_{0k} - \bw_k}_2 \, \max_m \norm{\bb(t_m)}_2 \leq \norm{\bw_{0k} - \bw_k}_2,
        \]
        where we used the Cauchy-Schwarz inequality and the fact that $\norm{\bb(t_m)}_2 \leq \norm{\bb(t_m)}_1 = 1$, since $\bb(t)$ is a basis of B-splines. Therefore,
        \[
            \max_m \norm{\bbeta_0(t_m) - \bbeta(t_m)}_2 \lesssim \epsilon_{n,M} + \max_k \norm{\bw_{0k} - \bw_k}_2.
        \]
        As such, we can find a constant $c_1 > 0$ such that $E_1 \subset \tilde{E}_1$. Based on a similar argument, we can show for $\lambda$-almost all $\set{t_m}_{m=1}^M$ that
        \[
            \max_m \norm{\bu_{0i}(t_m) - \bu_{i}(t_m)} \lesssim \epsilon_{n,M} + \max_{h} \norm{\bw_{0ih} - \bw_{ih}}_2.
        \]
        As such, we can find a constant $c_2 > 0$ large enough such that $E_2 \subset \tilde{E}_2$. The inclusion of these events establishes that $E_1 \cap E_2 \subset B_{n,M}(\mathcal{W};\epsilon_{n,M})$. Accordingly, we have that
        \begin{equation}\label{eq:lower_bound}
            \Pi_{\mathcal{W} \mid \brho^*}\set{B_{n,M}(\mathcal{W};\epsilon_{n,M})} \geq \Pi_{\mathcal{W}_{\beta} \mid \brho^*}(E_1) \Pi_{\mathcal{W}_{u} \mid \brho^*}(E_2),
        \end{equation}
        where we used the independence of the basis coefficients for the coefficient functions and the latent trajectories under the prior. It remains to show that the two probabilities on the right-hand side of the previous display are lower-bounded at the proposed rate.
        
        Using the independence of the basis coefficients under the prior, we have that the first probability on the right-hand side of Equation~(\ref{eq:lower_bound}) is
        \[
            \Pi_{\mathcal{W}_{\beta} \mid \brho^*}(E_1) = \prod_{k=1}^p \Pi_{\bw_k \mid \brho^*}(\norm{\bw_{0k} - \bw_k}_2 \leq c_1 \epsilon_{n,M}),
        \]
        where $\Pi_{\bw_k \mid \brho^*}$ denotes the prior measure on $\bw_k$ conditioned on $\brho^*$. To bound this probability, we use Lemma~\ref{lemma:small_ball} to obtain
        \begin{align*}
            \Pi_{\bw_k \mid \brho^*}(\norm{\bw_{0k} - \bw_k}_2 \leq c_1 \epsilon_{n,M}) &\gtrsim\exp\left(-\frac{1}{2} \frac{\norm{\bD_{\ell}^{(1)} \bw_{0k}}_2^2}{\sigma_{\beta}^{* 2}} -\frac{w_{0k,1}^2}{2\tau_{\beta}^2}\right) \\
            &\qquad\times \exp\left(-C_3 \frac{\ell^3 \sigma_{\beta}^{* 2}}{\epsilon^2_{n,M}} - \log \frac{\ell}{\epsilon_{n,M}}\right) 
        \end{align*}
        for some constant $C_3 > 0$. Next, recalling that $\sigma_{\beta}^{* 2} = b_1 \tilde{L} \epsilon_{n,M}^2$ for some constant $b_1 > 0$, $\norm{\bD_{\ell}^{(1)} \bw_{0k}}_2 \lesssim \tilde{L}$ by the definition of $\bw_{0k}$, and $1 \leq \ell \lesssim \epsilon_{n,M}^{-1}$, we have 
        \begin{align*}
            \Pi_{\mathcal{W}_{\beta} \mid \brho^*}(E_1) &\gtrsim  \exp\left\{-C_4 \left(\frac{\tilde{L}}{\epsilon_{n,M}^3} +  p\log\frac{1}{\epsilon_{n,M}} \right)\right\}\\
            &\geq  \exp\left\{-C_4 n \left(\frac{\tilde{L}}{\epsilon_{n,M}^3} +  \log\frac{1}{\epsilon_{n,M}} \right)\right\}. 
        \end{align*}
        
        Using a similar argument, we bound the second prior probability on the right hand side of Equation (\ref{eq:lower_bound}). Since the basis coefficients are independent under the prior, we have
        \begin{align*}
            \Pi_{\mathcal{W}_{u} \mid \brho^*}(E_2) &= \prod_{i=1}^n \prod_{h=1}^d \Pi_{\bw_{ih} \mid \brho^*}(\norm{\bw_{0ih} - \bw_{ih}}_2 \leq c_2 \epsilon_{n,M}),
        \end{align*}
        where $\Pi_{\bw_{ih} \mid \brho^*}$ denotes the prior measure on $\bw_{ih}$ conditioned on $\brho^*$. To bound the probability inside the product, we use Lemma~\ref{lemma:small_ball} and the fact that $\gamma_1 = \dots = \gamma_d = 1$ under $\Pi_{\mathcal{W}\mid \brho^*}$ to obtain
        \begin{align*}
            \Pi_{\bw_{ih \mid \brho^*}}(\norm{\bw_{0ih} - \bw_{ih}}_2 \leq c_2 \epsilon_{n,M}) &\gtrsim  \exp\left(-\frac{1}{2} \frac{\norm{\bD_{\ell}^{(1)} \bw_{0ih}}_2^2}{\sigma_{u}^{* 2}} -\frac{w_{0ih,1}^2}{2\tau^2}\right) \\
            &\qquad \times \exp\left(-C_5 \frac{\ell^3 \sigma_{u}^{* 2}}{\epsilon^2_{n,M}} - \log \frac{\ell}{\epsilon_{n,M}}\right)
        \end{align*}
        for some constant $C_5 > 0$. Next, using the fact that $\sigma_u^{* 2} = b_2 \tilde{L} \epsilon_{n,M}^2$, $\norm{\bD_{\ell}^{(1)}\bw_{0ih}}_2 \lesssim \tilde{L}$ by the definition of $\bw_{0ih}$, and  $1 \leq \ell \lesssim \epsilon_{n,M}^{-1}$, we have
        \begin{align*}
            \Pi_{\mathcal{W}_u \mid \brho^*}(E_2) &\gtrsim  \exp\left(-\frac{1}{2} \frac{nd \tilde{L} }{\epsilon_{n,M}^2} -\frac{ d \sum_{i=1}^n w_{0ih,1}^2}{2\tau^2}\right) \exp\left(-C_5 \frac{nd \tilde{L}}{\epsilon_{n,M}^3} - 2nd \log \frac{1}{\epsilon_{n,M}}\right) \\
            &\gtrsim  \exp\left\{-C_6 n \left(\frac{\tilde{L}}{\epsilon_{n,M}^3} + \log \frac{1}{\epsilon_{n,M}}\right)\right\},
        \end{align*}
        for some constant $C_6 > 0$. The last inequality used the fact that $(2\tau^2)^{-1} \sum_{i=1}^n  w_{0ih,1}^2 \lesssim n$. 
        
        Based on the above two lower bounds, we have
        \[
            \Pi_{\mathcal{W} \mid \brho^*}\{B_{n,M}(\mathcal{W};\epsilon_{n,M})\} \gtrsim \exp\left\{-C_7n \left(\frac{\tilde{L}}{\epsilon_{n,M}^3} + \log \frac{1}{\epsilon_{n,M}}\right)\right\},
        \]
        where $C_7 = C_4 + C_6$. The rate
        \[
            \epsilon_{n,M} = \max\left\{\left(\frac{L}{nM}\right)^{1/5}, \sqrt{\frac{\log nM}{nM}}\right\},
        \]
        is obtained when $n(n+1) M \epsilon_{n,M}^2 \gtrsim n \max\{\tilde{L}/\epsilon_{n,M}^3, \log(1/\epsilon_{n,M})\}$. Also, to replace $\tilde{L}$ with $L$, we used the fact that when $\tilde{L}= \max\{L, (nM)^{-3/2}\} = (nM)^{-3/2}$, then $\epsilon_{n,M} = \sqrt{\log nM / nM}$. As such, we have that $\Pi_{\mathcal{W} \mid \brho^*}\{B_{n,M}(\mathcal{W};\epsilon_{n,M})\}\gtrsim \exp\set{-C_8 n(n+1) M \epsilon_{n,M}^2/2}$ for some $C_8 > 0$ with this choice of $\epsilon_{n,M}$. 

    \end{proof}
    
    \subsection{Proof of Theorem~\ref{thm:vb_consistency}}
    
    Before proceeding with the proof, we layout some more preliminary results and definitions. To demonstrate the error bound for the global variational solution to Equation~(\ref{eq:kl_vi}) under the variational family $\mathcal{Q}$ defined in Equation~(\ref{eq:var_fam}), we use the following lemma that restates Theorem 3.3 in \citetSup{yang2020} in the context of the proposed model.

    \begin{lemma}[Risk bound of the $\alpha$-variational posterior]\label{lemma:yang}
        It holds with $\mathbb{P}_0$-probability at least $1 - \zeta$ that for any probability measure $q \in \mathcal{Q}$ with $q \ll p_0$,
        \begin{align}
            &\int \frac{2}{n(n+1)M} D_{\alpha}\{p(\mathcal{Y} \mid \mathcal{W}, \mathcal{X}) \mid \mid p_0(\mathcal{Y} \mid \mathcal{X})\} \hat{q}(\mathcal{W}, \brho) d\mathcal{W}d\brho \nonumber \\
            &= \frac{2 \alpha}{n(n+1)M(1 - \alpha)}   \bigg[- \int \log \frac{p(\mathcal{Y} \mid \mathcal{W}, \mathcal{X})}{p_0(\mathcal{Y} \mid \mathcal{X})} q(\mathcal{W}, \brho) d\mathcal{W}d\brho  + \frac{D_{KL}\{q(\mathcal{W}, \brho) \mid \mid p(\mathcal{W}, \brho)\}}{\alpha} \nonumber \\
            &\qquad\qquad\qquad\qquad\qquad+ \frac{\log(1/\zeta)}{\alpha} \bigg], \label{eq:yang_bound}
        \end{align}
        where $\hat{q}(\mathcal{W}, \brho)$ is the global optimizer of Equation~(\ref{eq:kl_vi}) with $\mathcal{Q}$ defined in Equation~(\ref{eq:var_fam}) and
        \[
            D_{\alpha}\{p(\mathcal{Y} \mid \mathcal{W}, \mathcal{X}) \mid \mid p_0(\mathcal{Y} \mid \mathcal{X})\} = \frac{1}{\alpha - 1} \int \{p_0(\mathcal{Y} \mid \mathcal{X})\}^{\alpha} \{p(\mathcal{Y} \mid \mathcal{W},  \mathcal{X})\}^{1-\alpha} d\mu
        \]
        is the $\alpha$-divergence with respect to a common dominating measure $\mu$. 
    \end{lemma}

    The key ingredient of the proof is finding a member $q^*(\mathcal{W}, \brho) \in \mathcal{Q}$ for which the terms on the right hand side of Equation~(\ref{eq:yang_bound}) in Lemma~\ref{lemma:yang} are bounded by $\epsilon_{n,M}^2$. To do so,  we layout some more definitions. As in the proof of Lemma~\ref{lemma:kl_neighborhood}, we define the events
    \begin{align*}
        E_1 = &\bigcap_{k=1}^p \left\{\bw_k \, : \, \norm{\bw_k - \bw_{0k}} \leq c_1 \epsilon_{n,M}\right\}, \\
        E_2 = &\bigcap_{i=1}^n \bigcap_{h=1}^d \left\{\bw_{ih} \, : \, \norm{\bw_{ih} - \bw_{0ih}} \leq c_2 \epsilon_{n,M}\right\},
    \end{align*}
    where $\bw_{0ih}$ and $\bw_{0k}$ are the basis coefficients of the spline approximations to the true functions $\beta_{0k}(t)$ and $u_{0ih}(t)$ constructed according to Corollary~\ref{cor:w_approx} and $c_1,c_2$ are chosen as in Lemma~\ref{lemma:kl_neighborhood}. Also, define $\sigma_{\beta}^* = b_1 \tilde{L}^{1/2} \epsilon_{n,M}$ and $\sigma_u^* = b_2 \tilde{L}^{1/2} \epsilon_{n,M}$ for positive constants $b_1, b_2 > 0$, $\tilde{L} = \max\{L, (nM)^{-3/2}\}$, and $\nu_1^* = \dots = \nu_d^* = 1$. Based on these definitions, we define the variational density $q^*(\mathcal{W}, \brho) = q^*(\mathcal{W})q^*(\brho)$ as follows
    \begin{align}\label{eq:vb_funcs}
        q^*(\mathcal{W}) &\propto \prod_{k=1}^p p(\bw_k \mid \sigma_{\beta}^{*})  \prod_{i=1}^n \prod_{h=1}^d p(\bw_{ih} \mid \sigma_u^{*}, \set{\nu_h^*}_{h=1}^d) \times 1\{E_1 \cap E_2\} \nonumber \\
        &\propto \prod_{k=1}^p p(\bw_k \mid \sigma_{\beta}^{*}) 1\{\norm{\bw_k - \bw_{0k}}_2 \leq c_1 \epsilon_{n,M}\} \nonumber \\
        &\qquad\qquad \times \prod_{i=1}^n \prod_{h=1}^d p(\bw_{ih} \mid \sigma_u^*, \set{\nu_h^*}_{h=1}^d) 1\{\norm{\bw_{ih} - \bw_{0ih}}_2 \leq c_2 \epsilon_{n,M}\},
    \end{align}
    and
    \begin{align}\label{eq:vb_var}
        q^*(\brho) \propto &\prod_{k=1}^p p(\sigma_{\beta_k}^2) 1\{\sigma_{\beta}^{* 2} \leq \sigma_{\beta_k}^2 \leq \sigma_{\beta}^{*2} e^{\epsilon_{n,M}^2} \} \prod_{i=1}^n p(\sigma_i^2) 1\{\sigma_u^{*2} \leq \sigma_i^{2} \leq \sigma_u^{*2} e^{\epsilon_{n,M}^2}\} \nonumber \\
        &\qquad\qquad \times \prod_{h=1}^d p(\nu_h) 1\{e^{-\epsilon_{n,M}^2} \leq \nu_h \leq 1\},
    \end{align}
    where $1\set{A}$ denotes the indicator function for a set $A$,  $p(\bw_k \mid \sigma_{\beta}^*)$ and $p(\bw_{ih} \mid \sigma_u^*, \set{\nu_h^*}_{h=1}^d)$ are the densities of the first-order Gaussian random walk priors on the basis coefficients, $p(\sigma_{\beta_k}^2)$ and $p(\sigma_i^2)$ are the densities of the $\text{Gamma}(c_{\sigma}/2,d_{\sigma}/2)$ priors on the transition variances, and $p(\nu_h)$ are the densities of the $\text{Gamma}(a_1, 1)$ and $\text{Gamma}(a_2, 1)$ priors from the multiplicative gamma process. Note that $q^*(\mathcal{W}, \brho)$ belongs to the variational family $\mathcal{Q}$. Using Lemma~\ref{lemma:yang}, we will show that $q^*(\mathcal{W}, \brho)$ contracts about the truth at the appropriate rate.

    The proof of Theorem~\ref{thm:vb_consistency} relies on the following lemma that establishes an upper-bound on the expected log-likelihood ratio under $q^*(\mathcal{W})$. The result follows from an application of Chebyshev's inequality with the necessary moments bounded using the KL support property of the prior in Lemma~\ref{lemma:kl_neighborhood}. The full proof is given in Appendix~\ref{subsec:aux_results}.

    \begin{lemma}\label{lemma:vb_lr}
        Suppose the true data generating process satisfies model (\ref{eq:dynlsm})--(\ref{eq:dynlsm_lr}) with true parameters $\bU_0(t)$ and $\bbeta_0(t)$ and observed covariates $\mathcal{X}$ satisfying Assumptions \ref{assump:func_space}--\ref{assump:x}. For $q^*(\mathcal{W})$ defined in Equation~(\ref{eq:vb_funcs}) with $\ell \asymp (nM)^{1/5}$,  we have for $\lambda$-almost all $\set{t_m}_{m=1}^M$ and any $D > 1$ that
        \begin{equation}\label{eq:kl_cheb}
            -\int_{\mathcal{W}} q^*(\mathcal{W}) \log \frac{p(\mathcal{Y} \mid  \mathcal{W}, \mathcal{X})}{p_0(\mathcal{Y} \mid \mathcal{X})} d\mathcal{W} \leq \frac{1}{2} D n (n+1) \epsilon_{n,M}^2
    \end{equation}
        holds with $\mathbb{P}_0$-probability converging to one.
    \end{lemma}
    
    Now, we prove Theorem~\ref{thm:vb_consistency} using Lemma~\ref{lemma:yang} along with an argument based on the chain rule of KL divergences developed by \citetSup{zhao2022}.

    \begin{proof}[Proof of Theorem~\ref{thm:vb_consistency}]
    According to Lemma~\ref{lemma:yang}, we need to establish upper bounds for 
        \begin{equation}\label{eq:exp_lr}
            -\int \log\frac{p(\mathcal{Y} \mid \mathcal{W}, \mathcal{X})}{p_0(\mathcal{Y} \mid \mathcal{X})} q(\mathcal{W}, \brho) d\mathcal{W} d\brho 
        \end{equation}
        and
        \[
            D_{KL}\{q(\mathcal{W}, \brho) \mid \mid p(\mathcal{W}\mid \brho) p(\brho)\},
        \]
        where $q(\mathcal{W}, \brho) = q(\mathcal{W})q(\brho)$ is a member of the variational family $\mathcal{Q}$ and $p(\mathcal{W} \mid \brho) p(\brho)$ is the prior. To proceed, we choose $q(\mathcal{W}, \brho) = q^*(\mathcal{W}, \brho)$ defined in Equation~(\ref{eq:vb_funcs}) and Equation~(\ref{eq:vb_var}).
 
        Under our choice of variational family, Equation~(\ref{eq:exp_lr}) simplifies as follows: 
        \begin{equation}
            -\int \log\frac{p(\mathcal{Y} \mid  \mathcal{W}, \mathcal{X})}{p_0(\mathcal{Y} \mid \mathcal{X})} q^*(\mathcal{W}) q^*(\brho) d\mathcal{W} d\brho = - \int \log\frac{p(\mathcal{Y} \mid  \mathcal{W}, \mathcal{X})}{p_0(\mathcal{Y} \mid \mathcal{X})} q^*(\mathcal{W}) d\mathcal{W}. 
        \end{equation}
        According to Lemma~\ref{lemma:vb_lr}, this expression is less than $(1/2) D n(n+1) \epsilon_{n,M}^2$ for any $D > 1$ and $\lambda$-almost all $\set{t_m}_{m=1}^M$ with $\mathbb{P}_0$-probability converging to one.

        Next, we use the chain rule of KL divergences and the independence of the variances parameters under the prior and variational posterior to establish that
        \begin{align*}
            D_{KL}\{q^*(\mathcal{W}, \brho) &\mid \mid p(\mathcal{W}, \brho)\} = \sum_{k=1}^p D_{KL}\{q^*(\sigma_{\beta_k}) \mid \mid p(\sigma_{\beta_k})\} + \sum_{i=1}^n D_{KL}\{q^*(\sigma_i) \mid \mid p(\sigma_i)\} \\
            &+ \sum_{h=1}^d D_{KL}\{q^*(\nu_h) \mid \mid p(\nu_h)\} + \int q^*(\brho) \int q^*(\mathcal{W}) \log \frac{q^*(\mathcal{W})}{p(\mathcal{W} \mid \brho)} d\mathcal{W} d\brho.
        \end{align*}
        We begin by bounding the KL divergence terms involving the variance parameters. Recall that for any probability measure $\mu$ and measurable set $A$ with $\mu(A) > 0$, we have $D_{KL}\set{\mu(\cdot \cap A) / \mu(A) \mid \mid \mu} = -\log \mu(A)$. As such, 
        \[
            \sum_{i=1}^n D_{KL}\{q^*(\sigma_i) \mid \mid p(\sigma_i)\} = -n \log \Pi_{\sigma_1}(\sigma_u^{*2} \leq \sigma_1^2 \leq \sigma_u^{*2} e^{\epsilon_{n,M}^2}),
        \]
        where we used the fact that $\sigma_1^2, \dots, \sigma_n^2 \iidsim \text{Gamma}(c_{\sigma}/2,d_{\sigma}/2)$ and $\Pi_{\sigma_1}$ is the corresponding probability measure under this prior. Let $f_{c_{\sigma}/2,d_{\sigma}/2}(\sigma_1)$ denote the density of a $\text{Gamma}(c_{\sigma}/2,d_{\sigma}/2)$ random variable. We have that
        \begin{align*}
            \Pi_{\sigma_1}(\sigma_u^{*2} \leq \sigma_1^2 \leq \sigma_u^{*2} e^{\epsilon_{n,M}^2}) &= \int_{\sigma_u^{*2}}^{\sigma_u^{*2}e^{\epsilon_{n,M}^2}} f_{c_{\sigma}/2,d_{\sigma}/2}(\sigma_1) d\sigma_1 \\
            &\geq \{\min_{\sigma_u^{*2} \leq \sigma_1^2 \leq \sigma_u^{*2} e^{\epsilon_{n,M}^2}} f_{c_{\sigma}/2,d_{\sigma}/2}(\sigma_1)\} \sigma_u^{*2} (e^{\epsilon_{n,M}^2} - 1) \\
            &\gtrsim \{\min_{\sigma_u^{*2} \leq \sigma_1^2 \leq \sigma_u^{*2} e^{\epsilon_{n,M}^2}} f_{c_{\sigma}/2,d_{\sigma}/2}(\sigma_1)\} \tilde{L} \epsilon_{n,M}^2 (e^{\epsilon_{n,M}^2} - 1) \\
            &\gtrsim \{\min_{\sigma_u^{*2} \leq \sigma_1^2 \leq \sigma_u^{*2} e^{\epsilon_{n,M}^2}} f_{c_{\sigma}/2,d_{\sigma}/2}(\sigma_1)\} \tilde{L} \epsilon_{n,M}^4,
        \end{align*}
        where the last inequality used the fact that $e^x - 1 \geq x$ for any $x$. For $M,n$ large enough such that $\epsilon_{n,M}^2 < 1$, we have on the interval $\sigma_u^{*2} \leq \sigma_1^2 \leq \sigma_u^{*2} e^{\epsilon_{n,M}^2}$ that
        \[
            -\log \{\min_{\sigma_u^{*2} \leq \sigma_1^2 \leq \sigma_u^{*2} e^{\epsilon_{n,M}^2}} f_{c_{\sigma}/2,d_{\sigma}/2}(\sigma_1)\} \lesssim \sigma_u^{*2} - \log\sigma_u^{*2} \lesssim \epsilon_{n,M}^2 + \log(1/\epsilon_{n,M}) - \log \tilde{L}.
        \]
        Since $\epsilon_{n,M} \geq (\log(nM)/nM)^{1/2}$ and $nM > 1$, we have that $\log(1/\epsilon_{n,M}) \lesssim \log(nM)$. As such,
        \begin{align*}
            -n\log \Pi_{\sigma_1}(\sigma_u^{*2} \leq \sigma_1^2 \leq \sigma_u^{*2} e^{\epsilon_{n,M}^2}) &\lesssim n \epsilon_{n,M}^2 + n\log(1/\epsilon_{n,M}) -n \log \tilde{L} \\
            &\lesssim n \epsilon_{n,M}^2 + n \log(nM) \\
            &\lesssim \frac{1}{2} n (n+1)M \epsilon_{n,M}^2,
        \end{align*}
        where the second inequality used the fact that $\tilde{L} \geq (nM)^{-3/2}$. We can apply a similar argument to show that
        \begin{align*}
            \sum_{k=1}^p D_{KL}\{q^*(\sigma_{\beta_k}) \mid\mid p(\sigma_{\beta_k})\} &= -p \log \Pi_{\sigma_{\beta_1}}(\sigma_{\beta}^{*2} \leq \sigma_{\beta_1}^2 \leq \sigma_{\beta}^{*2} e^{\epsilon_{n,M}^2}) \\
            &\lesssim \epsilon_{n,M}^2 + \log(1/\epsilon_{n,M}) - \log \tilde{L} \\
            &\lesssim \frac{1}{2} n(n+1)M \epsilon_{n,M}^2.
        \end{align*}
        Lastly,
        \[
            \sum_{h=1}^d D_{KL}\{q^*(\nu_h) \mid \mid p(\nu_h)\} = -\log \Pi_{\nu_1}(e^{-\epsilon_{n,M}^2} \leq \nu_1 \leq 1) -(d-1) \log \Pi_{\nu_2}(e^{-\epsilon_{n,M}^2} \leq \nu_2 \leq 1).
        \]
        For $n,M$ large enough so that $\epsilon_{n,M}^2 < 1$, we have that
        \begin{align*}
            \Pi_{\nu_1}(e^{-\epsilon_{n,M}^2} \leq \nu_1 \leq 1) &= \int_{e^{-\epsilon_{n,M}^2}}^1 f_{a_1, 1}(\nu_1) d\nu_1 \\
            &\geq \{\min_{e^{-\epsilon_{n,M}^2} \leq \nu_1 \leq 1} f_{a_1, 1}(\nu_1)\} (1 - e^{-\epsilon_{n,M}^2}) \\
            &\gtrsim \epsilon_{n,M}^2,
        \end{align*}
        where we used the fact that $1 - e^{-x} \geq x/2$ for $0 \leq x \leq 1$ and that a $\text{Gamma}(a, 1)$ density is lower-bounded by a constant when $e^{-1} \leq \nu_1 \leq 1$. The same argument shows that $\Pi_{\nu_2}(e^{-\epsilon_{n,M}^2} \leq \nu_2 \leq 1) \gtrsim \epsilon_{n,M}^2$. Thus,
        \[
            \sum_{h=1}^d D_{KL}\{q^*(\nu_h) \mid \mid p(\nu_h)\} \lesssim -d \log(\epsilon_{n,M}^2) \lesssim \frac{n(n+1)M}{2} \epsilon_{n,M}^2, 
        \]
        where the last inequality is due to the fact that $\epsilon_{n,M}^2 \geq \log(nM)/nM$.
        
        Now, we bound the third term of the KL divergence. Let $\brho^* = \set{\sigma_i = \sigma_u^*, \sigma_{\beta_k} = \sigma_\beta^*, \nu_h^* = 1, 1 \leq i \leq n, 1 \leq k \leq p, 1 \leq h \leq d}$. We have that
        \begin{align*}
            \int q^*(\mathcal{W}) \log \frac{q^*(\mathcal{W})}{p(\mathcal{W} \mid \brho)} d\mathcal{W} &= \int q^*(\mathcal{W}) \log \frac{q^*(\mathcal{W})}{p(\mathcal{W} \mid \brho)} \\
            &\qquad + q^*(\mathcal{W}) \log \frac{q^*(\mathcal{W})}{p(\mathcal{W} \mid \brho^*)} - q^*(\mathcal{W}) \log \frac{q^*(\mathcal{W})}{p(\mathcal{W} \mid \brho^*)} d\mathcal{W} \\
            &= \int_{E_1 \cap E_2} q^*(\mathcal{W}) \log\frac{p(\mathcal{W} \mid \brho^*)}{p(\mathcal{W} \mid \brho)}  d\mathcal{W} - \log \Pi_{\mathcal{W} \mid \brho^*}(E_1 \cap E_2).
        \end{align*}
        Based on the proof of Lemma~\ref{lemma:kl_neighborhood}, we have $-\log \Pi_{\mathcal{W} \mid \brho^*}(E_1 \cap E_2) \lesssim n (n+1)\epsilon_{n,M}^2/2$. Furthermore,
        \[
            \log \frac{p(\mathcal{W} \mid \brho^*)}{p(\mathcal{W} \mid \brho)} = \log \frac{p(\mathcal{W}_{u} \mid \brho^*)}{p(\mathcal{W}_{u} \mid \brho)} + \log \frac{p(\mathcal{W}_{\beta} \mid \brho^*)}{p(\mathcal{W}_{\beta} \mid \brho)}.
        \]
        The first term on the right-hand side of the previous expression is
        \begin{align*}
            \log \frac{p(\mathcal{W}_{u} \mid \brho^*)}{p(\mathcal{W}_{u} \mid \brho)} &= \sum_{i=1}^n \frac{(\ell - 1) d}{2}\log(\sigma_i) - \frac{n(\ell - 1)d}{2} \log(\sigma_u^*)  - \sum_{h=1}^d \sum_{s=1}^h \frac{n\ell}{2} \log(\nu_{s}) \\
            &\qquad + \sum_{i=1}^n \sum_{h=1}^d \gamma_h \frac{\norm{\bD_{\ell}^{(1)} \bw_{ih}}_2^2}{2\sigma_i^2} - \sum_{i=1}^n \sum_{h=1}^d \frac{\norm{\bD_{\ell}^{(1)} \bw_{ih}}_2^2}{2\sigma_u^{*2}}.
        \end{align*}
        In the constrained region of $q^*(\brho)$, that is, $\set{\sigma_i^2, i = 1, \dots, n \, : \, \sigma_u^{*2} \leq \sigma_i^2 \leq \sigma_u^{*2} e^{\epsilon_{n,M}^2}}$ and $\set{\nu_h, h = 1, \dots, d \, : \, e^{-\epsilon_{n,M}^2} \leq \nu_h \leq 1}$, we have
        \begin{align*}
            \log \frac{p(\mathcal{W}_{u} \mid \brho^*)}{p(\mathcal{W}_{u} \mid \brho)} \leq \frac{n(\ell-1)d}{2} \epsilon_{n,M}^2 + \frac{d(d+1) n \ell}{4} \epsilon_{n,M}^2 \lesssim \frac{n(n+1)M}{2} \epsilon_{n,M}^2,
        \end{align*}
        where we used that fact that $\gamma_h \leq 1$ for $1 \leq h \leq d$ in the constrained region and that fact that $\ell - 1 < \ell \lesssim (nM)^{1/5} \lesssim n(n+1)M$. Similarly, we have that
        \begin{align*}
            \log \frac{p(\mathcal{W}_{\beta} \mid \brho^*)}{p(\mathcal{W}_{\beta} \mid \brho)} &= \sum_{k=1}^p \frac{(\ell - 1)}{2}\log(\sigma_{\beta_k}) - \frac{p(\ell - 1)}{2} \log(\sigma_{\beta}^*) \\
            &\qquad + \sum_{k=1}^p \frac{\norm{\bD_{\ell}^{(1)} \bw_{k}}_2^2}{2\sigma_{\beta_k}^2} - \sum_{k=1}^p \frac{\norm{\bD_{\ell}^{(1)} \bw_{k}}_2^2}{2\sigma_{\beta}^{*2}}.
        \end{align*}
        In the constrained region of $q^*(\brho)$, that is, $\set{\sigma_{\beta_k}^2, k = 1, \dots, p \, : \, \sigma_{\beta}^{*2} \leq \sigma_{\beta_k}^2 \leq \sigma_{\beta}^{*2} e^{\epsilon_{n,M}^2}}$, we have
        \begin{align*}
            \log \frac{p(\mathcal{W}_{\beta} \mid \brho^*)}{p(\mathcal{W}_{\beta} \mid \brho)} &\leq  \frac{p(\ell - 1)}{2} \epsilon_{n,M}^2 \lesssim \frac{n(n+1)M}{2} \epsilon_{n,M}^2,
        \end{align*}
        where we used the fact that $\ell \lesssim (n^2M)^{1/5} \lesssim n(n+1)M$. 

        Based on these bounds, we can apply Lemma~\ref{lemma:yang} to conclude that with $\mathbb{P}_0$-probability converging to one
        \begin{align*}
            &\int \frac{2}{n(n+1)M} D_{\alpha}\{p(\mathcal{Y} \mid \mathcal{W}, \mathcal{X}) \mid \mid p_0(\mathcal{Y} \mid \mathcal{X})\} \hat{q}(\mathcal{W}, \brho) d\mathcal{W} d\brho = \\
            &\qquad \int \frac{2}{n(n+1)M} D_{\alpha}\{p(\mathcal{Y} \mid \mathcal{W},\mathcal{X}) \mid \mid p_0(\mathcal{Y} \mid \mathcal{X})\} \hat{q}(\mathcal{W}) d\mathcal{W} \lesssim \epsilon_{n,M}^2.
        \end{align*}
        The final result follows from Lemma~\ref{lemma:lb_bern_div}, which states that the $\alpha$-divergence between Bernoulli densities is lower-bounded by the squared loss up to a constant factor, and an application of Jensen's inequality since the squared loss is a convex function. 
    \end{proof}
    
    \subsection{Proof of Corollary~\ref{thm:vb_function_consistency}}
    
        We establish the result by showing that the error metrics for recovering the coefficient functions and latent trajectories are upper bounded by the squared Frobenius norm between the true and estimated log-odds matrices appearing in Theorem~\ref{thm:vb_consistency}. The proof uses various technical lemmas stated in Appendix~\ref{subsec:aux_results}. We start by demonstrating the error bound for recovering the latent trajectories. We have that
        \begin{align*}
            \frac{1}{Mnd} \sum_{m=1}^M \min_{\bO_m \in \mathcal{O}_d} \norm{\hat{\bU}(t_m) - \bU_0(t_m)\bO_m}_F^2 &\lesssim \frac{1}{Mnd \sigma_{min}^2(\bU_0)} \\
            &\qquad\quad\times \sum_{m=1}^M \lVert \hat{\bU}(t_m)\hat{\bU}(t_m)^{\top} - \bU_0(t_m) \bU_0(t_m)^{\top}\rVert_F^2 \\
            &\leq \frac{1}{Mnd \sigma_{min}^2(\bU_0) \kappa_{\mathcal{X}, d}} \sum_{m=1}^M \norm{\hat{\bTheta}_{t_m} - \bTheta_{0t_m}}_F^2 \\
            &= \frac{n}{\sigma_{min}^2(\bU_0) \kappa_{\mathcal{X}, d}} \, \frac{1}{Mn^2} \sum_{m=1}^M \norm{\hat{\bTheta}_{t_m} - \bTheta_{0t_m}}_F^2 \\
            &\lesssim  \frac{n}{\sigma_{min}^2(\bU_0) \kappa_{\mathcal{X}, d}} \max\left\{\left(\frac{L}{nM}\right)^{2/5},  \frac{\log nM}{nM}\right\},
        \end{align*} 
        where the first inequality uses Lemma~\ref{lemma:tu2016}, the second inequality uses Lemma~\ref{lemma:logit_ub}, and the last inequality holds with $\mathbb{P}_0$-probability converging to one by Theorem~\ref{thm:vb_consistency}. A similar argument establishes the error bound for the coefficient functions. We have that
        \begin{align*}
            \frac{1}{Mn^2} \sum_{m=1}^M \sum_{1 \leq i,j \leq n } \left[\{\hat{\bbeta}(t_m) - \bbeta_0(t_m)\}^{\top}\bx_{ij,t_m}\right]^2 &\leq \frac{1}{Mn^2 \kappa_{\mathcal{X}, d}} \sum_{m=1}^M \norm{\hat{\bTheta}_{t_m} - \bTheta_{0t_m}}_F^2 \\
            &\lesssim  \frac{1}{\kappa_{\mathcal{X}, d}}  \max\left\{\left(\frac{L}{nM}\right)^{2/5},  \frac{\log nM}{nM}\right\},
        \end{align*} 
        where the first inequality uses Lemma~\ref{lemma:logit_ub} and the second inequality holds with $\mathbb{P}_0$-probability converging to one according to Theorem~\ref{thm:vb_consistency}.

    \section{Auxiliary Technical Results}\label{subsec:aux_results}

    This section contains various auxiliary results used to prove the main theorems in the paper. 
    First, we establish Lemma~\ref{lemma:small_ball} concerning the small-ball probability of Gaussian random walk priors. The proof is based on a similar result in \citetSup{zhao2022}. The proof relies on quantifying the small-ball probabilities of Gaussian processes. In particular, we use the following lemma, whose proof is presented after the proof of Lemma~\ref{lemma:small_ball}.
    \begin{lemma}\label{lemma:gp_small_ball}
        Let $\{X(t), t \geq 0\}$ be a real-valued Gaussian process with mean zero, finite variance, and $X(0) = 0$. Assume that there exists a function $u(h)$ that is non-decreasing on $[0,1]$ and strictly positive and concave on $(0, 1)$ such that $\mathbb{E}\{X(t + h) - X(t)\}^2 \leq u^2(h)$ for $0 \leq t \leq t + h \leq 1$. If $u(h)/h^{\alpha}$ is non-decreasing on $(0,1)$ for some $\alpha > 0$, then
        \[
            \mathbb{P}\left\{\sup_{0 \leq t \leq 1} \abs{X(t)} \leq u(x) + u(3x)\sqrt{\frac{e^2\pi}{\alpha}} \right\} \geq \exp(-2/x).
        \]
    \end{lemma}

    \begin{proof}[Proof of Lemma~\ref{lemma:small_ball}]
        Define the events $E_1 = \set{\sup_{t=2,\dots,T} \abs{(w_t - w_1) - (w_{0t} - w_{01})} \leq \delta'}$, and $E_2 = \set{\abs{w_1 - w_{01}} \leq \delta'}$, where $\delta' = \delta / 2 \sqrt{T}$. We have that
        \[
            \mathbb{P}(\norm{\bw - \bw_0}_2 \leq \delta) \geq \mathbb{P}(E_1) \mathbb{P}(E_2),
        \]
        which follows from the independence of the increments $w_t - w_1$ from $w_1$ for $t \geq 2$ and the following bound
        \begin{align*}
            \norm{\bw - \bw_0}_2 &= \norm{(\bw - w_1\mathbf{1}_T) - (\bw_0 - w_{01}\mathbf{1}_T) + (w_1 - w_{01})\mathbf{1}_T}_2 \\
            &\leq \sqrt{T} \norm{(\bw - w_1\mathbf{1}_T) - (\bw_0 - w_{01}\mathbf{1}_T)}_{\infty} + \sqrt{T} \abs{w_1 - w_{01}},
        \end{align*}
        where $\mathbf{1}_T$ is the $T$-dimensional vector of ones.  

        We start by providing a lower-bound for $\mathbb{P}(E_1)$. For $t \geq 1$, let $Z_t \iidsim N(0, \sigma^2)$ and $\tilde{w}_t = \sum_{s=1}^{t} Z_s$. Denote $\tilde{\bw} = (\tilde{w}_1, \dots, \tilde{w}_{T-1})^{\top}$, which is equal in distribution to $(w_2 - w_1, \dots, w_T - w_1)^{\top}$, and let $\tilde{\bw}_0 =  (\tilde{w}_{01}, \dots, \tilde{w}_{0T-1})^{\top} = (w_{02} - w_{01}, \dots, w_{0T} - w_{01})^{\top}$. For clarity, let $L = T - 1$. Applying Anderson's inequality for the concentration of multivariate Gaussian random variables, we have
        \begin{align*}
            \mathbb{P}(E_1) &= \mathbb{P}\left(\sup_{t = 1, \dots, L} \abs{\tilde{w}_t - \tilde{w}_{0t}} \leq \delta'\right) \\
            &\geq \exp\left(-\frac{\norm{\bD_{T}^{(1)}\bw_0}_2^2}{2\sigma^2}\right) \mathbb{P}\left(\sup_{t=1, \dots, L}\abs{\tilde{w}_t} \leq \delta'\right).
        \end{align*}  
        To lower-bound the small-ball probability on the right-hand side of the previous expression, we consider a Gaussian process $\set{W(s), 0 \leq s \leq 1}$ induced by $(\tilde{w}_1, \dots, \tilde{w}_{L})$ using linear interpolation. Specifically, let $W(s) = \tilde{w}_{\floor{Ls}}  + (Ls - \floor{Ls}) Z_{\floor{Ls}+ 1}$ where $\tilde{w}_0 = 0$, $W(t/L) = \tilde{w}_t$, and $W(0) = 0$. Based on this construction, we have that
        \[
            \mathbb{P}\left(\sup_{t=1, \dots, L}\abs{\tilde{w}_t} \leq \delta'\right) \geq \mathbb{P}\left(\sup_{0 \leq s \leq 1} \abs{W(s)} \leq \delta'\right).
        \]
        
        We will use Lemma~\ref{lemma:gp_small_ball} to bound the small-ball probability on the right-hand side of the previous expression. To do so, we analyze $\mathbb{E}\{W(s+h) - W(s)\}^2$ for $0 \leq s \leq s+h \leq 1$. In particular, we start by showing that $\mathbb{E}\{W(s+h) - W(s)\}^2 \leq L^2 h \sigma^2$. For $0 \leq s \leq s + h \leq 1$, the increments are
        \begin{align*}
            W(s+h) - W(s) &= \tilde{w}_{\floor{L(s+h)}} - \tilde{w}_{\floor{Ls}} + \{L(s+h) - \floor{L(s+h)}\} Z_{\floor{L(s+h)} + 1} \\
            &\qquad - \{Ls - \floor{Ls}\} Z_{\floor{Ls} + 1}. 
        \end{align*}
        We separately demonstrate the bound for the three cases: (1) $\floor{L(s + h)} = \floor{Ls}$, (2) $\floor{L(s + h)} = \floor{Ls} + 1$, and (3) $\floor{L(s+h)} > \floor{Ls} + 1$. For $\floor{L(s + h)} = \floor{Ls}$, we have 
        \[
           \mathbb{E}\{W(s+h) - W(s)\}^2 = L^2 h^2 \sigma^2 \leq L^2 h \sigma^2,
        \]
        since $h \in (0, 1)$. For $\floor{L(s + h)} = \floor{Ls} + 1$, we have that
        \begin{align*}
            W(s+h) - W(s) &= \tilde{w}_{\floor{L(s+h)}} - \tilde{w}_{\floor{Ls}} + \{L(s+h) - \floor{L(s+h)}\} Z_{\floor{L(s+h)} + 1} \\
            &\qquad - \{Ls - \floor{Ls}\} Z_{\floor{Ls} + 1} \\
            &= \{L(s+h) - \floor{Ls} - 1\} Z_{\floor{Ls}+ 2} \\
            &\qquad - \{Ls - \floor{Ls} - 1\} Z_{\floor{Ls} + 1},
        \end{align*}
        so that
        \begin{align*}
            \mathbb{E}\{W(s+h) - W(s)\}^2 &= \{(L(s+h) - \floor{Ls} - 1)^2 + (1 - (Ls - \floor{Ls}))^2\} \sigma^2 \\
            &= L^2 \left\{\left(\frac{L(s+h) - \floor{Ls} - 1}{L}\right)^2 + \left(\frac{1 - (Ls - \floor{Ls})}{L}\right)^2\right\} \sigma^2 \\
            &\leq L^2 \left\{\left(\frac{L(s+h) - \floor{Ls} - 1}{L}\right) + \left(\frac{1 - (Ls - \floor{Ls})}{L}\right)\right\} \sigma^2 \\
            &= L^2 h \sigma^2,
        \end{align*}
        where the inequality used the fact that $x^2 \leq x$ for $x \in [0, 1]$.
        Lastly, for $\floor{L(s+h)} > \floor{Ls} + 1$, we have
        \begin{align*}
            W(s+h) - W(s) &= W(s+h) - W\left(\frac{\floor{L(s + h)}}{L}\right) \\
            &\quad+ \sum_{k=\floor{Ls}+2}^{\floor{L(s+h)}} W\left(\frac{k}{L}\right) - W\left(\frac{k-1}{L}\right) \\
            &\quad+ W(\frac{\floor{Ls}+1}{L}) - W(s) \\
            &= \{L(s+h) - \floor{L(s+h)}\} Z_{\floor{L(s+h)}+1} + \sum_{k=\floor{Ls}+2}^{\floor{L(s+h)}} Z_k \\
            &\quad- \{Ls - \floor{Ls} - 1\} Z_{\floor{Ls} + 1}. 
        \end{align*}
        Therefore,
        \begin{align*}
            \mathbb{E}\{W(s+h) - W(s)\}^2 = &\{L(s+h) - \floor{L(s+h)}\}^2 \sigma^2 \\
            &+ (\floor{L(s+h)} - \floor{Ls} - 1) \sigma^2 \\
            &+ \{1 - (Ls - \floor{Ls})\}^2 \sigma^2 \\
            &\leq L^2\bigg\{\left(\frac{L(s+h) - \floor{L(s+h)}}{L}\right)^2 \\
            &\quad+ \frac{\floor{L(s+h)} - \floor{Ls} - 1}{L} + \left(\frac{1 - (Ls - \floor{Ls})}{L}\right)^2\bigg\}\sigma^2 \\
            &\leq L^2\bigg\{\frac{L(s+h) - \floor{L(s+h)}}{L} \\
            &\quad+ \frac{\floor{L(s+h)} - \floor{Ls} - 1}{L} + \frac{1 - (Ls - \floor{Ls})}{L}\bigg\}\sigma^2 \\
            &= L^2 h \sigma^2.
        \end{align*}
        
        Define $u(h) = L h^{1/2} \sigma$ so that based on the previous inequalities, we have that $\mathbb{E}\{W(s+h) - W(s)\}^2 \leq u^2(h)$ for all $0 \leq s \leq s + h \leq 1$. Furthermore, $u(h)/h^{1/2} = L \sigma$ is non-decreasing on $(0,1)$. As such, using Lemma~\ref{lemma:gp_small_ball}, we have
        \[
            \mathbb{P}\left(\sup_{0 \leq s \leq 1} \abs{W(s)} \leq \delta'\right) \geq \exp\left(-C \frac{L^2\sigma^2}{\delta^{'2}}\right) \geq \exp\left(-4 C \frac{T^3\sigma^2}{\delta^2}\right),
        \]
        for some constant $C > 0$. For the second probability,  we have
        \begin{align*}
            \mathbb{P}(E_2) &\geq \frac{1}{\sqrt{2\pi\tau^2}} \exp\left(-\frac{w_{01}^2}{2\tau^2}\right) \left(\frac{\delta}{\sqrt{T}}\right) \\
            &\geq \frac{1}{\sqrt{2\pi\tau^2}} \exp\left(-\frac{w_{01}^2}{2\tau^2}\right) \left(\frac{\delta}{T}\right) \\
            &\gtrsim \exp\left\{-\frac{w_{01}^2}{2\tau^2} - \log\left(\frac{T}{\delta}\right)\right\}.
        \end{align*}
        Finally, combining the previous lower bounds gives the result.    
    \end{proof}
    
    The proof of Lemma~\ref{lemma:gp_small_ball} is a minor modification of the proof of Theorem 1.1 in \citetSup{shao1993}. As such, we will need the following two lemmas from \citetSup{shao1993} stated without proof.
    \begin{lemma}[Lemma 2.3 in \citetSup{shao1993}]\label{lemma:sup_max}
        Let $\{X(t), t \geq 0\}$ be a real-valued Gaussian process with mean zero and finite variance. Assume that there exists a non-decreasing function $u(h)$ on $[0,1]$ such that $\mathbb{E}\{X(t+h) - X(t)\}^2 \leq u^2(h)$ for all $0 \leq t \leq t+ h\leq 1$, then
        \[
            \mathbb{P}\left\{\sup_{0 \leq t \leq 0} \abs{X(t)} \leq x + 2 e \int_0^{\infty} u\left(\frac{e \cdot e^{-y^2}}{R}\right) dy \right\} \geq e^{-R} \ \mathbb{P}\left(\max_{0 \leq i \leq R} \abs*{X\left(\frac{i}{R}\right)} \leq x\right),
        \]
        for every $R \geq 1, x > 0$.
    \end{lemma}
    
    \begin{lemma}[Lemma 2.4 in \citetSup{shao1993}]\label{lemma:max_lb}
        Let $\{\xi_i, 1 \leq i \leq n\}$ be Gaussian random variables with mean zero and finite variances. Then for every $x > 0$
        \[
            \mathbb{P}\left(\max_{1 \leq i \leq n} \abs*{\sum_{j=1}^i \xi_j} \leq x\right) \geq \prod_{i=1}^n \sqrt{\frac{1}{2\pi}} \int_0^{2x/\rho_i} e^{-y^2/2}dy,
        \]
        where $\rho_i^2 = \sum_{j=1}^n \abs{\mathbb{E}(\xi_j \xi_i)}$.
    \end{lemma}
    
    \begin{proof}[Proof of Lemma~\ref{lemma:gp_small_ball}]
        Taking $R = 1/x$ in Lemma~\ref{lemma:sup_max}, we have that
        \begin{align*}
            \mathbb{P}\left\{\sup_{0 \leq t \leq 0} \abs{X(t)} \leq u(x) + 2e \int_0^{\infty} u\left(e \cdot e^{-y^2} \cdot x\right) dy \right\} \geq e^{-1/x} \ \mathbb{P}\left(\max_{0 \leq i \leq 1/x} \abs*{X(ix)} \leq u(x)\right).
        \end{align*}
        Apply Lemma~\ref{lemma:max_lb}, we get
        \[
            \mathbb{P}\left(\max_{0 \leq i \leq 1/x} \abs*{X(ix)} \leq u(x)\right) \geq \prod_{i=1}^{\lfloor 1/x \rfloor} \sqrt{\frac{1}{2\pi}} \int_0^{2u(x)/\rho_i} e^{-y^2/2} dy,
        \]
        where $\rho_i^2 = \sum_{j=1}^{\lfloor 1/x \rfloor} \abs*{\mathbb{E}[\{X(jx) - X((j-1)x)\}\{X(ix) - X((i-1)x)\}]}$ for $1 \leq i \leq \lfloor 1/x \rfloor$. From the Cauchy-Schwarz inequality and the concavity of $u^2(h)$ on $(0,1)$, we obtain that
        \begin{align*}
            \rho_i^2 &\leq \sum_{j=1}^{\lfloor 1/x \rfloor} \sqrt{\mathbb{E}\{X(jx) - X((j-1)x)\}^2 \mathbb{E}\{X(ix) - X((i-1)x)\}^2} \\
            &\leq \sum_{j=1}^{\lfloor 1/x \rfloor} u^2(x) \leq 2 u^2(x).
        \end{align*}
        Therefore,
        \begin{align*}
            \mathbb{P}\left(\max_{0 \leq i \leq 1/x} \abs*{X(ix)} \leq u(x)\right) &\geq \prod_{i=1}^{\lfloor 1/x \rfloor} \sqrt{\frac{1}{2\pi}} \int_0^{\sqrt{2}} e^{-y^2/2} dy \\
            &= \exp\left\{\bigg\lfloor\frac{1}{x}\bigg\rfloor \log \left(\frac{2 \Phi(\sqrt{x}) - 1}{2}\right)\right\} \\
            &\geq \exp(-0.87/x),
        \end{align*}
        where $\Phi(\cdot)$ stands for the cumulative distribution function of a standard normal random variable, and  we used the fact that $\log(2 \Phi(\sqrt{x}) - 1) \geq -0.87$ in the last line. 

        A combination of the above inequalities yields
        \[
            \mathbb{P}\left\{\sup_{0 \leq t \leq 0} \abs{X(t)} \leq u(x) + 2e \int_0^{\infty} u\left(e \cdot e^{-y^2} \cdot x\right) dy \right\} \geq e^{-1.87/x}.
        \]
        Lastly, the fact that $u(h)/h^{\alpha}$ is non-decreasing on $(0,1)$ for some $\alpha > 0$ implies that
        \begin{align*}
            \int_0^{\infty} u(e \cdot x \cdot e^{-y^2}) dy &= \int_0^{\infty} e^{\alpha} x^{\alpha} e^{-\alpha y^2} \frac{u(e \cdot x \cdot e^{-y^2})}{e^{\alpha} x^{\alpha} e^{-\alpha y^2}} dy \\
            &\leq u(ex) \int_0^{\infty} e^{-\alpha y^2} dy \\
            &= \frac{u(e \cdot x)}{2} \sqrt{\frac{\pi}{\alpha}} \\
            &\leq \frac{u(3x)}{2} \sqrt{\frac{\pi}{\alpha}}.
        \end{align*}

    \end{proof}
    
    Next, we establish Lemma~\ref{lemma:vb_lr}, which bounds the expected log-likelihood ratio under the $\alpha$-variational posterior $q^*(\mathcal{W})$. The proof uses the notation and definitions outlined in Appendix~\ref{sec:proof_vb_con}.

    \begin{proof}[Proof of Lemma~\ref{lemma:vb_lr}]
         The proof uses Chebyshev's inequality to lower-bound the probability of the event in Equation (\ref{eq:kl_cheb}). We begin by characterizing the first two moments of $-\int_{\mathcal{W}} q^*(\mathcal{W}) \log \{p(\mathcal{Y} \mid  \mathcal{W}, \mathcal{X}) / p_0(\mathcal{Y} \mid \mathcal{X})\} d\mathcal{W}$ under $\mathbb{P}_0$. Let $\mathbb{E}_0$ and $\text{Var}_0$ denote the expectation and variance under $\mathbb{P}_0$. We have
        \begin{align*}
            \mathbb{E}_{0}\bigg(-\int_{\mathcal{W}} q^*(\mathcal{W}) \log &\frac{p(\mathcal{Y} \mid \mathcal{W}, \mathcal{X})}{p_0(\mathcal{Y} \mid \mathcal{X})} d\mathcal{W} \bigg) = \int_{\mathcal{W}} q^*(\mathcal{W}) \mathbb{E}_0\left\{-\log \frac{p(\mathcal{Y} \mid  \mathcal{W}, \mathcal{X})}{p_0(\mathcal{Y} \mid \mathcal{X})} \right\} d\mathcal{W} \\
            &\propto \int_{E_1 \cap E_2} p(\mathcal{W}) D_{KL}\left\{p_0(\mathcal{Y} \mid \mathcal{X}) \mid \mid p(\mathcal{Y} \mid  \mathcal{W}, \mathcal{X}) \right\} d\mathcal{W} \\
            &\leq \int_{B_{n,M}(\mathcal{W}; \epsilon_{n,M})} p(\mathcal{W}) D_{KL}\left\{p_0(\mathcal{Y} \mid \mathcal{X}) \mid \mid p(\mathcal{Y} \mid  \mathcal{W}, \mathcal{X}) \right\} d\mathcal{W} \\
            &\leq \frac{1}{2} n(n+1) \epsilon_{n,M}^2,
        \end{align*}
        where we used the definition of $q^*(\mathcal{W})$ in the second line and the definition of the KL-neighborhood in the last line. Similarity, by applying Fubini's theorem and Jensen's inequality, we have
        \begin{align*}
            \text{Var}_0\bigg(\int_{\mathcal{W}} q^*(\mathcal{W}) \log \frac{p(\mathcal{Y} \mid  \mathcal{W}, \mathcal{X})}{p_0(\mathcal{Y} \mid \mathcal{X})} &d\mathcal{W}\bigg) \leq \mathbb{E}_0\left(\int_{\mathcal{W}} q^*(\mathcal{W}) \log \frac{p(\mathcal{Y} \mid  \mathcal{W}, \mathcal{X})}{p_0(\mathcal{Y} \mid \mathcal{X})} d\mathcal{W}\right)^2 \\
            &\leq \mathbb{E}_0\left\{\int_{\mathcal{W}} q^*(\mathcal{W}) \log^2 \frac{p(\mathcal{Y} \mid  \mathcal{W}, \mathcal{X})}{p_0(\mathcal{Y} \mid \mathcal{X})} d\mathcal{W}\right\} \\
            &= \int_{\mathcal{W}} q^*(\mathcal{W}) \mathbb{E}_0\left\{\log^2 \frac{p(\mathcal{Y} \mid  \mathcal{W}, \mathcal{X})}{p_0(\mathcal{Y} \mid \mathcal{X})}\right\} d\mathcal{W} \\
            &= \int_{\mathcal{W}} q^*(\mathcal{W}) V_2\left\{p_0(\mathcal{Y} \mid \mathcal{X}) \mid \mid p(\mathcal{Y} \mid  \mathcal{W}, \mathcal{X}) \right\} d\mathcal{W} \\
            &\lesssim \int_{B_{n,M}(\mathcal{W};\epsilon_{n,M})} p(\mathcal{W})V_2\left\{p_0(\mathcal{Y} \mid \mathcal{X}) \mid \mid p(\mathcal{Y} \mid \mathcal{W}, \mathcal{X}) \right\} d \mathcal{W} \\
            &\leq \frac{1}{2} n (n+1) M \epsilon_{n,M}^2.
        \end{align*}
        Therefore, by Chebyshev's inequality, for any $D > 1$, we have
        \begin{align*}
            &\mathbb{P}_0\left(\int_{\mathcal{W}} q^*(\mathcal{W}) \log \frac{p(\mathcal{Y} \mid \mathcal{W}, \mathcal{X})}{p_0(\mathcal{Y} \mid \mathcal{X})} d\mathcal{W} \leq -D \frac{1}{2} n (n+1) M \epsilon_{n,M}^2 \right)  \\
            &\leq \mathbb{P}_0\bigg\{\int_{\mathcal{W}} q^*(\mathcal{W}) \log \frac{p(\mathcal{Y} \mid  \mathcal{W}, \mathcal{X})}{p_0(\mathcal{Y} \mid \mathcal{X})} d\mathcal{W} \\
            &\qquad-\mathbb{E}_0\left(\int_{\mathcal{W}} q^*(\mathcal{W}) \log \frac{p(\mathcal{Y} \mid  \mathcal{W}, \mathcal{X})}{p_0(\mathcal{Y} \mid \mathcal{X})} \right) \leq -(D-1) \frac{1}{2} n (n+1) M \epsilon_{n,M}^2\bigg\} \\
            &\leq \text{Var}_0\left(\int_{\mathcal{W}} q^*(\mathcal{W}) \log \frac{p(\mathcal{Y} \mid  \mathcal{W}, \mathcal{X})}{p_0(\mathcal{Y} \mid \mathcal{X})} d\mathcal{W}\right) \bigg/ \left\{(D-1)^2 \frac{1}{4}n^2 (n+1)^2 M^2 \epsilon_{n,M}^4\right\} \\
            &\leq \frac{2}{(D-1)^2 n (n+1) M \epsilon_{n,M}^2}.
        \end{align*} 
        Therefore, for any $D > 1$, we have for $\lambda$-almost all $\set{t_m}_{m=1}^m$ that
        \[
            -\int_{\mathcal{W}} q^*(\mathcal{W}) \log \frac{p(\mathcal{Y} \mid \mathcal{W}, \mathcal{X})}{p_0(\mathcal{Y} \mid \mathcal{X})} d\mathcal{W} \leq \frac{1}{2} D n (n+1) \epsilon_{n,M}^2
        \]
        holds with probability at least $1 - 4/\{(D-1)^2 n (n+1) M \epsilon_{n,M}^2\}$. This proves that when $n(n+1)M\epsilon_{n,M}^2/2 \rightarrow \infty$, we have for $\lambda$-almost all $\set{t_m}_{m=1}^M$ that
        \[
            -\int_{\mathcal{W}} q^*(\mathcal{W}) \log \frac{p(\mathcal{Y} \mid \mathcal{W}, \mathcal{X} )}{p_0(\mathcal{Y} \mid \mathcal{X})} d\mathcal{W} \leq \frac{1}{2} D n (n+1) \epsilon_{n,M}^2
        \]
        holds with $\mathbb{P}_0$-probability converging to one.
    \end{proof}
    
    The following two lemmas present an upper bound for the KL divergence and a lower bound for the $1/2$-divergence between two Bernoulli distributions in terms of the squared difference of their success probabilities. Proofs can be found in \citetSup{zhao2022}.

    \begin{lemma}[Lemma A.4 in \citetSup{zhao2022}]\label{lemma:ub_bern_kl}
        Let $p_a = 1 / (1 + e^{-a})$ and $p_b = 1 / (1 + e^{-b})$. Define $P_a$ and $P_b$ as the Bernoulli measures with success probability $p_a$ and $p_b$, respectively. Then we have
        \[
            D_{KL}(P_a \mid\mid P_b) + D_{KL}(P_b \mid\mid P_a) \leq (a-b)^2.
        \]
    \end{lemma}

    \begin{lemma}[Lemma A.5 in \citetSup{zhao2022}]\label{lemma:lb_bern_div}
        Let $p_a = 1 / (1 + e^{-a})$ and $p_b = 1 / (1 + e^{-b})$. Define $P_a$ and $P_b$ as the Bernoulli measures with success probability $p_a$ and $p_b$, respectively. Suppose there exists constants $c,C > 0$ such that $c < a,b < C$, then we have
        \[
            D_{1/2}(P_a, P_b) \gtrsim (b - a)^2.
        \]
    \end{lemma}
    
    The next result provides a lower-bound to the squared Frobenius-norm between the true and estimated log-odds matrices in terms of error metrics for the coefficient functions and latent trajectories. The lemma is a modification of Lemma 24 in \citetSup{ma2020} to account for more than one dyadic covariate. To simplify the proof, we introduce some new notation. For two matrices $\bA$ and $\bB$, we denote the trace inner-product as $\langle \bA, \bB \rangle = \tr(\bA^{\top}\bB)$. Also, for a matrix $\bA$, we denote its nuclear norm as $\norm{\bA}_*$. In addition, we let $\Delta_{\beta_k(t)} = \hat{\beta}_k(t) - \beta_{0k}(t)$ and $\bX_{k,t} \in \Reals{n \times n}$ denote the covariate matrix at time $t \in \set{t_m}_{m=1}^M$ with entries $[\bX_{k,t}]_{ij} = x_{ijk,t}$. 

    \begin{lemma}\label{lemma:logit_ub}
        If Assumption~\ref{assump:stable_rank} holds, then for all $t \in \set{t_m}_{m=1}^M$ 
    \[
        \norm{\hat{\bTheta}_t - \bTheta_{0t}}^2_F \geq \left(1 - \sqrt{\frac{2d}{r(\mathcal{X})}}\right)\left(\norm{\hat{\bU}(t)\hat{\bU}(t)^{\top} - \bU_0(t)\bU_0(t)^{\top}}_F^2 + \norm{\mathcal{X}_t \, \bar{\times}_3 \, \{\hat{\bbeta}(t) - \bbeta_0(t)\}}_F^2\right).
    \]
    \end{lemma}
    \begin{proof}
        From the definition of the log-odds matrix, we have
        \begin{align}\label{eq:delta_def}
            \norm{\hat{\bTheta}_t - \bTheta_{0t}}_F^2 &= \norm{\hat{\bU}(t)\hat{\bU}(t)^{\top} - \bU_0(t)\bU_0(t)^{\top}}_F^2 + \norm{\sum_{k=1}^p \Delta_{\beta_k(t)} \bX_{k,t}}_F^2 \nonumber\\
            &\qquad+ 2 \braket*{\hat{\bU}(t)\hat{\bU}(t)^{\top} - \bU_0(t)\bU_0(t)^{\top}, \sum_{k=1}^p \Delta_{\beta_k(t)} \bX_{k,t}}.
        \end{align}
        By H\"{o}lder's inequality, we have
        \begin{align*}
            |\langle\hat{\bU}(t)\hat{\bU}(t)^{\top} &- \bU_0(t)\bU_0(t)^{\top}, \sum_{k=1}^p \Delta_{\beta_k(t)} \bX_{k,t}\rangle| \leq \norm{\hat{\bU}(t)\hat{\bU}(t)^{\top} - \bU_0(t)\bU_0(t)^{\top}}_* \\
            &\qquad\qquad\qquad\qquad\qquad\qquad\qquad\qquad\times\norm{\sum_{k=1}^p \Delta_{\beta_k(t)}\bX_{k,t}}_{op} \\
            &\leq\sqrt{2d} \norm{\hat{\bU}(t)\hat{\bU}(t)^{\top} - \bU_0(t)\bU_0(t)^{\top}}_F \norm{\sum_{k=1}^p \Delta_{\beta_k(t)}\bX_{k,t}}_{op}  \\
            &\leq\sqrt{2d} \norm{\hat{\bU}(t)\hat{\bU}(t)^{\top} - \bU_0(t)\bU_0(t)^{\top}}_F \frac{\norm{\sum_{k=1}^p \Delta_{\beta_k(t)}\bX_{k,t}}_F}{\sqrt{r(\mathcal{X})}} \\
            &\leq \sqrt{\frac{d}{2r(\mathcal{X})}}\left(\norm{\hat{\bU}(t)\hat{\bU}(t)^{\top} - \bU_0(t)\bU_0(t)^{\top}}_F^2 + \norm{\sum_{k=1}^p \Delta_{\beta_k(t)} \bX_{k,t}}_F^2\right),
        \end{align*}
        where the third inequality used Assumption~\ref{assump:stable_rank} and the last inequality used the fact that $2ab \leq a^2 + b^2$ for any $a,b \in \Reals{}$. Substituting the previous inequality into Equation~(\ref{eq:delta_def}), we have
        \begin{align*}
            \norm{\hat{\bTheta}_t - \bTheta_{0t}}_F^2 &\geq \left(1 - 2 \sqrt{\frac{d}{2r(\mathcal{X})}}\right) \left(\norm{\hat{\bU}(t)\hat{\bU}(t)^{\top} - \bU_0(t)\bU_0(t)^{\top}}_F^2 + \norm{\sum_{k=1}^p \Delta_{\beta_k(t)} \bX_{k,t}}_F^2\right)  \\
            &= \bigg(1 - \sqrt{\frac{2d}{r(\mathcal{X})}}\bigg) \bigg(\norm{\hat{\bU}(t)\hat{\bU}(t)^{\top} - \bU_0(t)\bU_0(t)^{\top}}_F^2 \\
            &\qquad\qquad\qquad\qquad\qquad+ \norm{\mathcal{X}_t \, \bar{\times}_3 \, \{\hat{\bbeta}(t) - \bbeta_0(t)\}}_F^2\bigg).  
        \end{align*}
    \end{proof}
    
    Lastly, we state the following lemma from \citetSup{tu2016} that relates two common metrics for comparing matrices.
    \begin{lemma}[Lemma 5.4 in \citetSup{tu2016}]\label{lemma:tu2016}
        For any $\bU_1, \bU_2 \in \Reals{n \times d}$, we have
        \[
            \min_{\bO \in \mathcal{O}_d} \norm{\bU_1 - \bU_2 \bO}_F^2 \leq  \frac{1}{2(\sqrt{2} - 1) \sigma_d^2(\bU_2)} \norm{\bU_1 \bU_1^{\top} - \bU_2\bU_2^{\top}}_F^2,
        \]
        where $\sigma_d(\bU_2)$ is the $d$-th largest singular value of $\bU_2$.
    \end{lemma} 

\section{Additional Empirical Results}\label{sec:more_results}

This section contains more details about the simulation studies, additional results on simulated data, and the remaining figures from the real data application.

\subsection{Settings for the Competing Methods}

Section~\ref{subsec:comp} of the main text compared our methodology with the GP model of \citetSup{durante2014} and FASE~\citepSup{macdonald2023}. The remaining details on how we estimated these competitors are as follows. The GP model of \citetSup{durante2014} places GP priors with exponential covariance functions on the latent functions. We set the length scale of the exponential covariance function $b = 0.1$. We generated samples from the model's posterior using the Hamiltonian Monte Carlo with adaptive parameter tuning~\citepSup{neal2011, hoffman2014} implementation in NumPyro~\citepSup{phan2019, bingham2019}. We estimated FASE using the R package \texttt{fase} with default hyperparameter values. We selected the model parameters, that is, the latent space dimension $d$ and basis dimension $\ell$, by minimizing the network generalized cross-validation (NGCV) criterion recommended by \citetSup{macdonald2023} through an exhaustive search over an 18-parameter grid $\set{(d,\ell) \, : \, 1 \leq d \leq 6, \ell = 5, 7, 9}$.

\subsection{Model Comparison for Different Network Densities}

Here, we present results on the performance of the competing methods for sparser and denser networks compared to the ones used in Section~\ref{subsec:comp} of the main text. The results are on networks generated from the same data-generating process used in Section~\ref{subsec:comp}; however, we set the expected density equal to 0.1 for the sparser case and 0.3 for the denser case. For all models, we used the same estimation procedure and hyperparameter settings as the study presented in Section~\ref{subsec:comp} of the main text.

Table~\ref{tab:comp} reports the results aggregated over 50 independent replicates for the same network sizes used in the original simulation study. Overall, our conclusions remain the same. All methods recovered the true dyad-wise probabilities with high accuracy, with the proposed method performing the best or equivalent to the best in all scenarios. Furthermore, the proposed method's computation time remained an order of magnitude faster than the competitors in most scenarios. We also observe that the computation time of the proposed method decreased as the network's density increased because the SVI algorithm scales with the network's density as opposed to the number of possible dyadic relations.

\begin{table}
    \centering\small
    \begin{tabular}{@{}cclcc@{}} \toprule
        $(n, M)$ & Density & Method & PCC & Computation Time (seconds) \\
    \bottomrule
        $(100, 10)$ & 0.1 & GP & 0.91   & 1075 (677) \\
                    &     & FASE & 0.92 & 95 (15) \\
                    &     & P-Spline (Proposed) & 0.93 & 7 (1)\\[0.5em]
                    & 0.3 & GP & 0.95   & 974 (268)\\
                    &     & FASE & 0.96  & 75 (13)\\
                    &     & P-Spline (Proposed) & 0.97  & 15 (2) \\[0.5em]
        \hline
        $(100, 20)$ & 0.1 & GP & 0.94   & 14160 (99) \\
                    &     & FASE & 0.95   & 150 (34) \\
                    &     & P-Spline (Proposed) & 0.95  & 8 (1)\\[0.5em]
                    & 0.3 & GP & 0.97   &  14454 (4771)\\
                    &     & FASE & 0.97  & 103 (23) \\
                    &     & P-Spline (Proposed) & 0.98   & 13 (2)  \\[0.5em]
        \hline 
        $(200, 10)$ & 0.1 & GP & 0.95   & 4486 (1838) \\
                    &     & FASE & 0.95   & 283 (21) \\
                    &     & P-Spline (Proposed) & 0.96  & 22 (5)\\[0.5em]
                    & 0.3 & GP & 0.98   &  3555 (1014)\\
                    &     & FASE & 0.98  & 215 (16) \\
                    &     & P-Spline (Proposed) & 0.98   & 44 (14)  \\[0.5em]
    \bottomrule
    \end{tabular}
    \caption{Average PCCs and computation times for the competing methods over the 50 replications. The values in parentheses indicate one standard deviation. The standard deviations for the PCCs are not included because they are all less than 0.01.}
    \label{tab:comp_appendix}
\end{table}

\subsection{Sensitivity to  Subsampling Fractions}\label{sec:subsample_sens}

In this simulation, we evaluated the effect of the subsample fractions $\gamma_n$ and $\gamma_M$ on the performance of the proposed SVI algorithm. We generated synthetic networks from the data-generating process described in Section~\ref{subsec:sim_setup} of the main text for varying network sizes and expected edge densities. We estimated the model using the proposed SVI algorithm with the same hyperparameter values used in the simulation study in Section~\ref{sec:sim_study}; however, we varied the non-edge fraction $\gamma_n$ and time point fraction $\gamma_M$ used to construct the unbiased estimates of the natural gradients. Furthermore, we set $m_0 = \ceil{\gamma_M M}$ instead of $\min(\ceil{\gamma_M M}, 100)$ to quantify the effect of subsamples  of time points larger than 100. We calculated the RMSE for recovering the true log-odd matrices as defined in Section~\ref{subsec:recovery} to measure performance. In all settings, we calculated the error metric over 50 independent replicates.

In Figure~\ref{fig:subsample_nonedges}, we report the results for synthetic networks with $n = 250$ nodes, $M = 100$ time points, and expect edge densities equal to 0.05, 0.1, 0.2, and 0.3. In this scenario, we varied $\gamma_n \in \set{1, 2, 3, 4, 5}$ and fixed $\gamma_M = 0.25$. Starting at $\gamma_n = 1$ when the number of non-edges associated with a node equals the degree of that node, the average errors subsequently decreased for all expected densities. For expected densities equal to  0.1, 0.2, and 0.3, the average errors remained roughly equal for $\gamma_n \geq 2$. For the sparsest setting where the expected density is 0.05, the average error is minimized at $\gamma_n = 3$, and subsequently increased afterward. However, the performance remained roughly constant after accounting for the variance over the simulations. Based on these results, we recommend setting $\gamma_n = 2$, which performed well across all settings and leads to a faster run time.

\begin{figure}[htb]
\centering
\begin{subfigure}[b]{0.48\textwidth}
    \centering 
    \includegraphics[width=\textwidth, keepaspectratio]{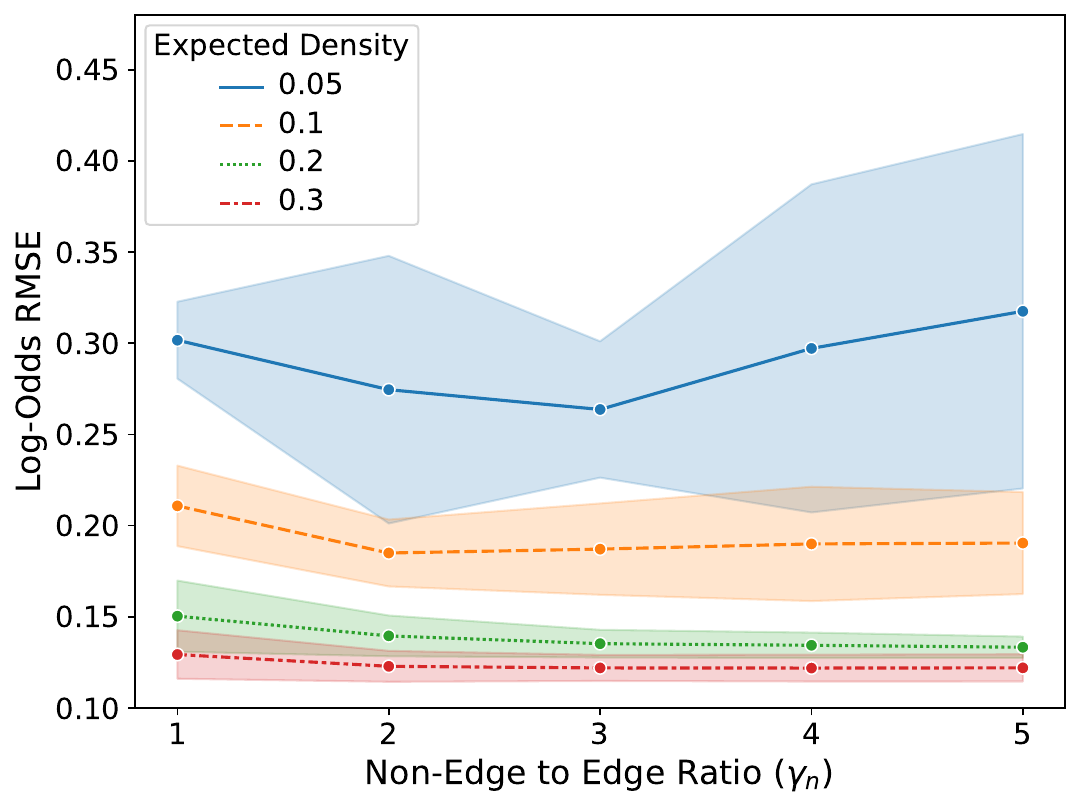}
    \caption{}
    \label{fig:subsample_nonedges}
\end{subfigure}
\begin{subfigure}[b]{0.48\textwidth}
    \centering 
    \includegraphics[width=\textwidth, keepaspectratio]{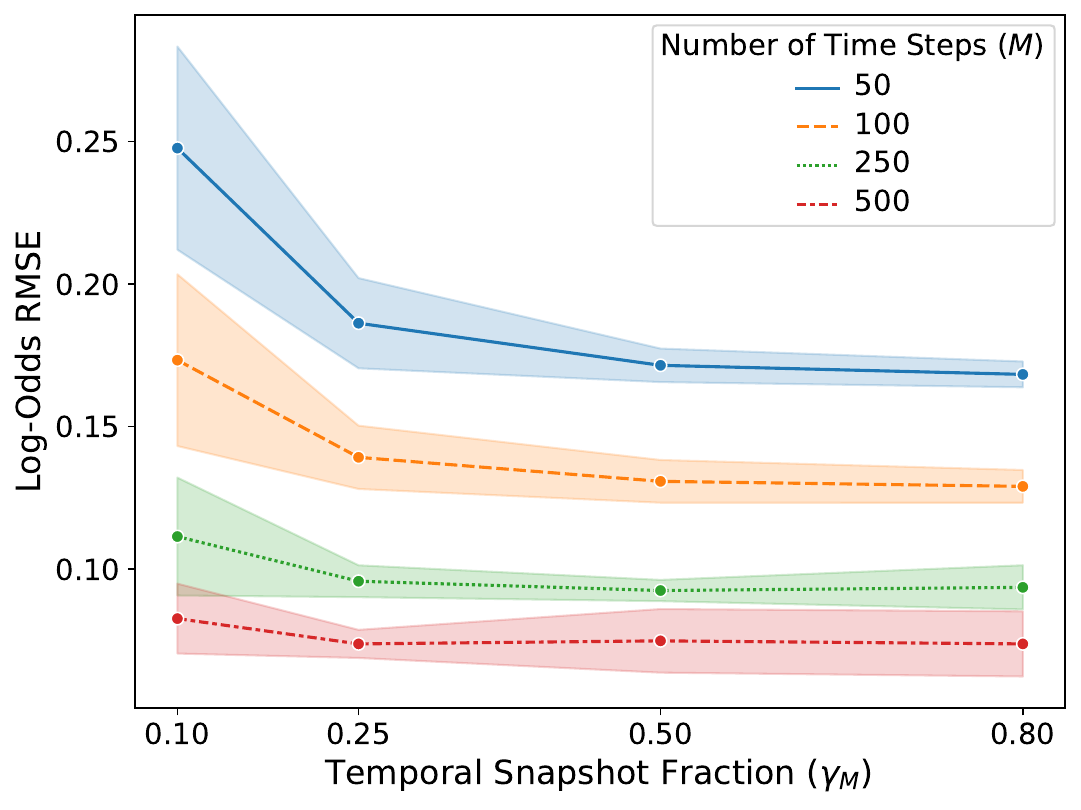}
    \caption{}
    \label{fig:subsample_time}
\end{subfigure}
    \caption{RMSE of recovering the true log-odds matrices as (a) $\gamma_n$ and (b) $\gamma_M$ increase for various network sizes and expected densities. The curves and shaded regions indicate averages and one standard deviation over 50 independent replicates, respectively.}
\end{figure}

In Figure~\ref{fig:subsample_time}, we report the results for synthetic networks with $n = 250$ nodes,  $M \in \set{50, 100, 250, 500}$, and an expected edge density of 0.2. In this scenario, we varied $\gamma_M \in \set{0.1, 0.25, 0.5, 0.8}$ and fixed $\gamma_n = 2$. The errors significantly decreased as $\gamma_M$ increased from 0.1 to 0.25 but remained roughly constant afterward. As such, we recommend setting $\gamma_M = 0.25$, which performed well across all settings.

\subsection{Additional Figures from the Real Data Application}\label{sec:additional_figures}

Figure~\ref{fig:nodewise_vars} displays the ten nations with the largest nodewise transition variances for the international conflict network analyzed in Section~\ref{sec:application} of the main text.  Ukraine, Venezuela, and Ethiopia are in the top five nations with the highest transition variances. Figure~\ref{fig:shrinkage} reports the estimated shrinkage parameters estimated on the same network.  The shrinkage parameters decreased significantly until $\hat{\gamma}_3^{-1}$ at which the curve leveled out. We chose a latent space dimension of $d = 2$ based on this observation. 

\begin{figure}[htb]
\centering \includegraphics[width=\textwidth, height=0.3\textheight, keepaspectratio]{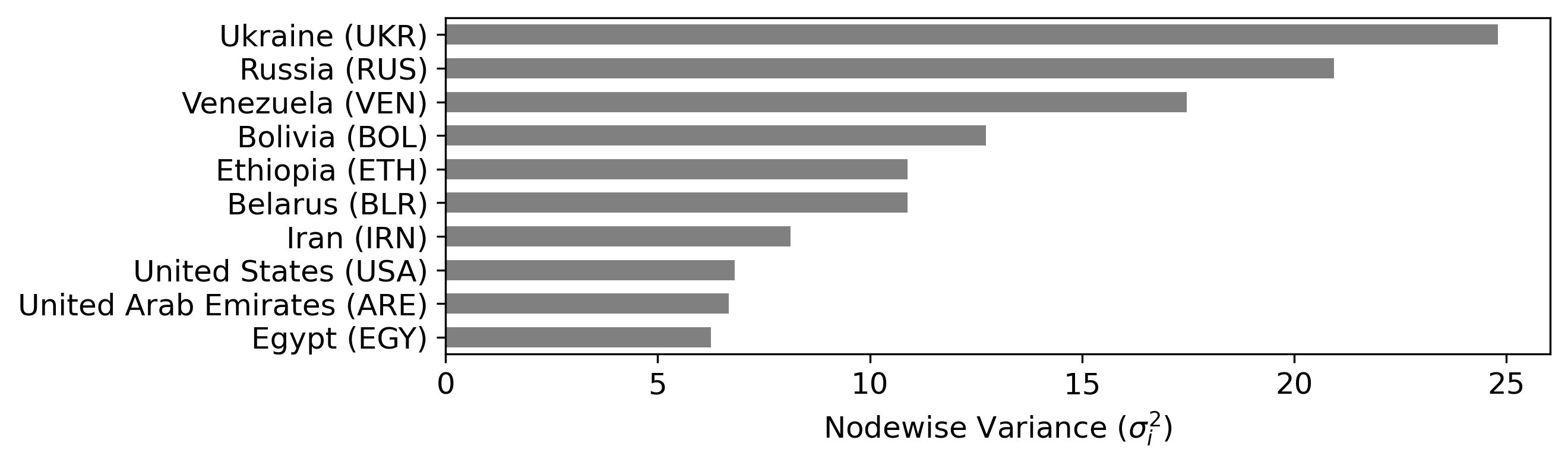}
    \caption{The $\alpha$-variational posterior means of the nodewise transition variances $\sigma_i^2$ for the international conflict network. The plot is restricted to the 10 nations with the largest transition variances.}
\label{fig:nodewise_vars}
\end{figure}

\begin{figure}[htb]
\centering 
\includegraphics[width=\textwidth, height=0.25\textheight, keepaspectratio]{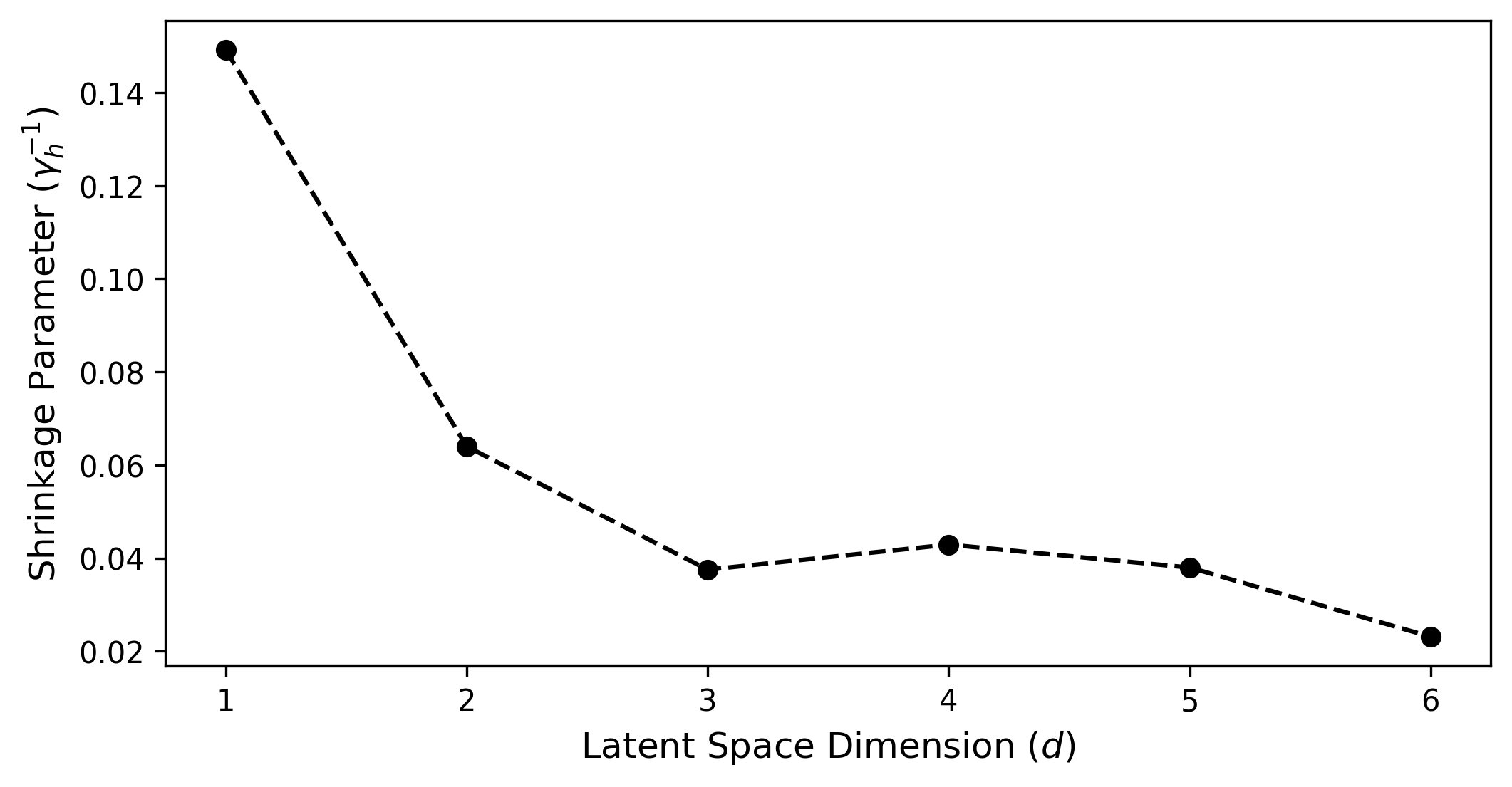}
\caption{The $\alpha$-variational posterior means of the shrinkage parameters for the international conflict network.}
\label{fig:shrinkage}
\end{figure}

\section{Properties of B-Spline Basis Functions}\label{sec:spline_overview}
    
    This section reviews the properties of B-splines used to prove Theorem~\ref{thm:vb_consistency}. Let $\bb_q(t) = (B_{1, q}(t), \dots, B_{\ell, q}(t))^{\top}$ denote a basis of B-spline functions of degree $q$ (or order $q+1$) with $K$ equally spaced internal knots so that $\ell = K + q + 1$. We denote the knot sequence by $\{\kappa_i\}_{i=1}^{K + 2 (q + 1)}$ with uniform knot spacing $h = \kappa_i - \kappa_{i-1} =  1/(K + 1)$ so that
    \[
        \kappa_i = (i - q - 1) h = \frac{i - q - 1}{K + 1}, \qquad i = 1, \dots, K + 2 (q + 1).
    \]
    Note that $[\kappa_{q+1}, \kappa_{\ell + 1}] = [0, 1]$.

    \subsection{Derivatives of B-Splines}

    In what follows, let $\bw \in \Reals{\ell}$ denote a vector of basis coefficients. From Equation (14) on page 117 in \citetSup{deboor1978} , we have for $t \in [0, 1]$ that 
    \[
        D^1\{\bw^{\top}\bb_q(t)\} = q \sum_{j=2}^{\ell} \frac{\Delta w_j}{\kappa_{j + q} - \kappa_j} B_{j,q-1}(t) = (K+1) (\bD_{\ell}^{(1)} \bw)^{\top} \bb_{q - 1}(t),
    \] 
    where $\Delta w_j = w_j - w_{j-1}$ and $D^r f$ denotes the $r$-th derivative of an $r$-times differentiable function $f$. Based on Corollary 8 on page 133 of \citetSup{deboor1978} or Theorem 4.38 on page 143 in \citetSup{schumaker2007},  we have for some constant $C_1 > 0$ that only depends on $q$ that
    \begin{align}\label{eq:deboor}
        \norm{\bD_{\ell}^{(1)} \bw}_2 \leq \ell \norm{\bD_{\ell}^{(1)} \bw}_{\infty} &\leq C_1 \left(\frac{\ell}{K+1}\right) \norm{(K+1) (\bD_{\ell}^{(1)} \bw)^{\top} \bb_{q-1}(t)}_{L_{\infty}[0,1]} \\
        &= C_1 \left(\frac{\ell}{K+1}\right) \norm*{D^1\{\bw^{\top} \bb_q(t)\}}_{L_{\infty}[0,1]},
    \end{align}
    where we used the definition of the first-derivative of a B-spline function in the last equality.
    
    \subsection{Approximation Properties of B-Splines}
    
    The following theorem quantifies the ability of B-splines to approximate a function $f$, which is a member of a certain smooth function space. In particular, let $L_p^{\sigma}[a,b]$ denote the Sobolev space in $L_p[a,b]$, that is, the Lebesgue space of real-valued functions on the interval $[a, b] \subset \Reals{}$, with absolutely continuous derivatives up to order $\sigma-1$. In other words,
    \[
        L_p^{\sigma}[a,b] = \left\{f \, : \, D^{\sigma-1} f \textrm{ is absolutely continuous on } [a, b] \textrm{ and } D^{\sigma}f \in L_p[a,b] \right\}.
    \]
    Also let $\norm{\cdot}_{L_p[a,b]}$ denote the $L_p$ norm on $[a,b]$. We have the following theorem concerning the approximation properties of B-splines of degree $m$ (order $m + 1$) with equally spaced internal knots.

    \begin{lemma}[Theorem 6.25 in \citetSup{schumaker2007}]\label{lemma:approx}
        Let $1 \leq p \leq q \leq \infty$ and $1 \leq \sigma \leq m+1$. Then for every $f \in L_p^{\sigma}[a,b]$ there exists a $\bw_0 \in \Reals{\ell}$ with $\ell = m + K + 1$ and a constant $C_2 > 0$ that only depends on $m$ and $p$ such that
        \[
        \begin{cases}
            \underset{r=0, \dots, \sigma - 1}{\norm{D^r[f - \bw_0^{\top}\bb_{m}(t)]}_{L_q[a,b]}} \\
            \underset{r=\sigma, \dots, m}{\norm{D^r [\bw_0^{\top}\bb_{m}(t)]}_{L_q[a,b]}}
        \end{cases} \leq C_2 h^{\sigma - r + 1/q - 1/p} \norm{D^{\sigma} f}_{L_p[a,b]},
    \]
        where $h = 1 / (K + 1)$.
    \end{lemma}
    
    An immediate corollary to this lemma is that if $f \in L_{\infty}^1[0,1]$, then there exists a $\bw_0 \in \Reals{\ell}$  such that
    \begin{align*}
        \norm{f(t) - \bw_0^{\top} \bb_m(t)}_{L_{\infty}[0,1]} &\leq \frac{C_2}{K+1} \  \norm*{D^1 f}_{L_{\infty}[0,1]}, \\
        \norm{D^1 \set{\bw_0^{\top} \bb_{m}(t)}}_{L_{\infty}[0,1]} &\leq C_2  \norm*{D^1 f}_{L_{\infty}[0,1]}.
    \end{align*}
    Combining Equation~(\ref{eq:deboor}) and the previous expression, we have the upper bound
    \[
        \norm{\bD_{\ell}^{(1)} \bw_0}_2 \leq C_1 \left(\frac{\ell}{K+1}\right) \norm{D^1\set{\bw_0^{\top}\bb_m(t)}}_{L_{\infty}[0,1]} \leq  C_1C_2 \left(\frac{\ell}{K+1}\right)\norm*{D^1 f}_{L_{\infty}[0,1]}.
    \]

    To summarize, we have the following lemma used to analyze the approximating properties of the proposed P-spline prior for dynamic LSMs when $\ell \rightarrow \infty$.
    \begin{lemma}\label{lemma:spline_approx}
        For $f \in L_{\infty}^1[0,1]$, there exists a $\bw_0 \in \Reals{\ell}$ such that
        \begin{align*}
            \norm{f(t) - \bw_0^{\top} \bb_{m}(t)}_{L_{\infty}[0,1]} &\lesssim \ell^{-1} \  \norm*{D^1 f}_{L_{\infty}[0,1]}, \\ 
            \norm{\bD_{\ell}^{(1)} \bw_0}_2 &\lesssim \norm*{D^1 f}_{L_{\infty}[0,1]},
        \end{align*}
        where $\bb_{m}(t)$ is a basis of B-spline functions of degree $m \geq 1$ with $K$ equally spaced internal knots so that $\ell = K + m + 1$.
    \end{lemma}
    
\section{Overview of Stochastic Variational Inference}\label{sec:svi_overview}

Here, we briefly review the concepts behind stochastic variational inference (SVI) necessary to understand the derivations in this article and refer to \citetSup{hoffman2013} for a comprehensive overview. SVI applies to a class of models for a set of $n$ observations $\by = \set{y_1, \dots, y_n}$ with $K$ blocks of global latent variables $\bw = \set{\bw_1, \dots, \bw_K}$ and $n$ local latent variables $\bomega = \{\omega_1, \dots, \omega_n\}$. In particular, the joint distribution should factorize as
\[
    p(\by, \bomega, \bw) = p(\bw) \prod_{i=1}^n p(y_i, \omega_i \mid \bw),
\]
so that the $i$-th local latent variable is associated with the $i$-th observation. Furthermore, SVI requires that the full-conditional distribution of the latent variables be members of the exponential family, that is, 
\begin{align*}
    p(\bw_k \mid \cdot) &\propto \exp\{\boldeta_{\bw_k}(\by, \bomega, \bw_{-k})^{\top} \bt_{\bw_k}(\bw_k) - \psi_{\bw_k}(\bw_k)\}, \qquad 1 \leq k \leq K, \\
    p(\omega_i \mid \cdot) &\propto \exp\{\boldeta_{\omega_i}(y_i, \bw)^{\top}\bt_{\omega_i}(\omega_i) - \psi_{\omega_i}(\omega_i)\}, \qquad 1 \leq i \leq n,
\end{align*}
where $\psi_{\bw_k}(\cdot)$ and $\psi_{\omega_i}$ are cumulant functions, $\bt_{\bw_k}(\cdot)$ and $\bt_{\omega_i}(\cdot)$ are vectors of sufficient statistics, $\boldeta_{\bw_k}(\cdot)$ and $\boldeta_{\omega_i}(\cdot)$ are the vector of natural parameters, and $\bw_{-k}$ denotes the collection of all global latent variables except $\bw_k$. Such a relationship is satisfied by the augmented model developed in this article.

In the SVI framework, we seek a variational approximation to the posterior by maximizing the ELBO 
\begin{align*}
    \hat{q}(\bw, \bomega) &= \argmax_{q(\bw, \bomega) \in \mathcal{Q}} \mathbb{E}_{q(\bw, \bomega)}\left[\log\left\{\frac{p(\by, \bw, \bomega)}{q(\bw, \bomega)}\right\}\right], 
\end{align*}
for variational distributions within the variational family 
\[
    \mathcal{Q} = \left\{q(\bw, \bomega) \,: \, q(\bw, \bomega) = q(\bw)q(\bomega) = \prod_{k=1}^K q(\bw_k) \prod_{i=1}^n q(\omega_i)\right\}.
\]
In this section, we will denote the ELBO by $\textsf{ELBO}[q(\bw, \bomega)]$ to highlight its depends on the variational distribution. As outlined in \citetSup{bishop2006}, the optimal variational factor of each latent variable  in $\mathcal{Q}$ is a member of the same exponential family as its full-conditional distribution, that is, 
\begin{align*}
    q(\bw_k) &\propto \exp\{\blambda_{\bw_k}^{\top} \bt_{\bw_k}(\bw_k) - \psi_{\bw_k}(\bw_k)\}, \qquad 1 \leq k \leq K, \\
    q(\omega_i) &\propto \exp\{\bphi_{\omega_i}^{\top} \bt_{\omega_i}(\omega_i) - \psi_{\omega_i}(\omega_i)\}, \qquad 1 \leq i \leq n.
\end{align*}
Using the fact that the full conditionals and the variational factors have the same exponential family representation, \citetSup{hoffman2013} showed that the natural gradient of the ELBO with respect to the variational factors' natural parameters are
\begin{align*}
    \nabla_{\blambda_{\bw_k}} \textsf{ELBO}[q(\bw, \bomega)] &= \mathbb{E}_{-q(\bw_k)}[\boldeta_{\bw_k}(\by, \bomega, \bw_{-k})] - \blambda_{\bw_k}, \qquad 1 \leq k \leq K,  \\
    \nabla_{\bphi_{\omega_i}} \textsf{ELBO}[q(\bw, \bomega)] &= \mathbb{E}_{q(\bw)}[\boldeta_{\omega_i}(y_i, \bw)] - \bphi_{\omega_i}, \qquad 1 \leq i \leq n.
\end{align*}
Setting these gradients to zero provides the solutions to the well known coordinate ascent variational inference (CAVI) algorithm~\citepSup{blei2017}. 

A severe computational bottleneck is that these gradients must be computed over the entire data set. To make this bottleneck clear, under the class of models under study, we can decompose the gradients associated with the global latent variables into three terms
\begin{align}
    \nabla_{\blambda_{\bw_k}} \textsf{ELBO}[q(\bw, \bomega)] &= -\blambda_{\bw_k} + \mathbb{E}_{-q(\bw_{k})}[\boldeta_{\bw_k}(\bw_{-k})] \nonumber \\
    &\qquad+ \sum_{i=1}^n \mathbb{E}_{-q(\bw_{k})}[\mathbb{E}_{q(\omega_i)}[\boldeta_{\bw_k}(y_i, \omega_i, \bw_{-k})]], \label{eq:global_grad}
\end{align}
The second term only depends on the global latent variables and the third term is a sum over the individual observations and local latent variables. Motivated by this decomposition of the gradients, \citetSup{hoffman2013} proposed SVI, which replaces the full gradients with cheaper to compute stochastic estimates.

SVI uses unbiased estimates of the natural gradients associated with the global latent variables obtained by subsampling the observations and local latent  variables used in the summation in Equation~(\ref{eq:global_grad}). Given a subsample of observations, the algorithm alternates between two steps until convergence. Let $s$ be the current iteration of the algorithm. The first step sets the natural parameters of the local variational factors associated with the subsampled observations to their optimal values given the current estimate of the global variational factors $\hat{q}(\bw)$ by setting their natural gradients to zero, that is, 
\[
    \hat{\bphi}_{\omega_i} = \mathbb{E}_{\hat{q}(\bw)}[\boldeta_{\omega_i}(y_i, \bw)].
\]
Then based only on the subsampled observations and local variational factors with optimal values, an unbiased estimate of the natural gradients of the global latent variables are calculated $\widehat{\nabla}_{\bw_k} \textsf{ELBO}[q(\bw, \bomega)]$ 
and a step of size $\rho_s$ is take in their direction, that is,
\[
    \blambda_{\bw_k}^{(s)} = \blambda_{\bw_k}^{(s-1)} + \rho_s \widehat{\nabla}_{\blambda_{\bw_k}} \textsf{ELBO}[q(\bw, \bomega)] \mid_{\blambda_{\bw_k}^{(s)}}, \qquad k = 1,\dots, K.
\]
To ensure convergence of the global variational parameters, the step size $\rho_s$ should satisfy $\sum_{s} \rho_s + \infty$ and $\sum_s \rho_s^2 < \infty$~\citepSup{robbinsmonro1951}. 

In summary, determining the natural gradients of the ELBO used in an SVI algorithm involves the following two steps: (1) Determining the full-conditional distribution of the latent variables to identify the optimal form of the variational factors, and (2) Taking the expectation of the full conditional's natural parameters under the variational posterior to calculate the gradients according to Equation~(\ref{eq:global_grad}). A cheap stochastic approximations of the natural gradients are then obtained by defining an appropriate unbiased estimate of the summation in Equation~(\ref{eq:global_grad}).

\bibliographystyleSup{apalike}
\bibliographySup{references}

\end{document}